\documentclass[journal]{IEEEtran}
\ifCLASSINFOpdf
\else
\fi
\usepackage{stmaryrd}
\usepackage{color}
\usepackage{dsfont}    
\usepackage{graphicx}
\usepackage{amsmath}
\usepackage{amssymb}
\usepackage{amsthm}
\usepackage{epstopdf}  
\usepackage{subfigure}
\usepackage{float}
\usepackage{algorithm}
\usepackage{algorithmicx}
\usepackage{algpseudocode}
\usepackage{multicol}
\usepackage{balance}
\usepackage{enumitem}

\usepackage{url}

\usepackage{mathrsfs}
\usepackage{xcolor}

\newtheorem{theorem}{Theorem}
\newtheorem{define}[theorem]{Definition}
\newtheorem{lem}[theorem]{Lemma}

\newtheorem{assump}[theorem]{Assumption}

\IEEEoverridecommandlockouts

\begin{document}

\title{Privacy Preserving Event Detection}

\author{Xiaoshan~Wang
        and~Tan~F.~Wong
\thanks{Manuscript received December 1, 2023; revised December 3, 2024.}
\thanks{The authors are with the Department of Electrical and Computer Engineering, University of Florida, Gainesville, FL 32611,  USA (e-mail: xiaoshanwang@ufl.edu; twong@ufl.edu). This material is based upon work supported by the National Science Foundation under Grant No. 2106589. Any opinions, findings, and conclusions or recommendations expressed in this material are those of the author(s) and do not necessarily reflect the views of the National Science Foundation.}
}

\maketitle

\begin{abstract}
  This paper presents a privacy-preserving event detection scheme
  based on measurements made by a network of sensors. A diameter-like
  decision statistic made up of the marginal types of the measurements
  observed by the sensors is employed. The proposed detection scheme
  can achieve the best type-I error exponent as the type-II error rate
  is required to be negligible.  Detection performance with
  finite-length observations is also demonstrated through a simulation
  example of spectrum sensing.  Privacy protection is achieved by
  obfuscating the sensors' marginal types with random zero-modulo-sum
  numbers that are generated and distributed via the exchange of
  encrypted messages among the sensors. The privacy-preserving
  performance against ``honest but curious'' adversaries, including
  colluding sensors, the fusion center, and external eavesdroppers, is
  analyzed through a series of cryptographic games. It is shown that
  the probability that any probabilistic polynomial time adversary
  successfully estimates the sensors' measured types cannot be much
  better than independent guessing, when there are at least two
  non-colluding sensors.
\end{abstract}

\begin{IEEEkeywords}
Privacy protection, event detection, $K$-sample problem, cryptographic game, wireless sensor network.
\end{IEEEkeywords}

%
\IEEEpeerreviewmaketitle


\section{Introduction}

A typical event detection system consists of a network of sensors
distributed in a target area for collecting and reporting measurement
data to a fusion center which aggregates the reported data to make a
detection decision. In this paper, we develop a privacy-preserving
event detection scheme, in which the sensors obfuscate the square
roots of the marginal types (empirical distributions) of their
measurements with random zero-modulo-sum (ZMS) numbers before
uploading them to the fusion center, which then performs a binary
hypothesis test based on a \emph{diameter-like} statistic that
measures the similarity level of the uploaded types. In this way, the
target event can be detected without exposing the data of individual
sensors. The proposed scheme results from a joint detection-privacy
design approach aiming to achieve two interconnected objectives.  One
objective is to construct a test that can achieve the best detection
error exponents, and the other is to develop a privacy-preserving
protocol that can minimize the probability of potential attackers
successfully estimating the sensors' types.


\subsection{$K$-sample problem} \label{sec:k-sample} 

Crowd-sensing of spectrum occupancy in, e.g., the citizen broadband
radio service (CBRS) band by smart phones is a practical application
that motivates the event detection problem considered (see
Section~\ref{sec:simulation} for a more detailed
example). Intuitively, the distributions of the received powers
measured by distributed sensors would be different when a potential
source is transmitting because of the different distances between the
source and the sensors. On the other hand, when the source is silent,
the received power distributions would be similar because only noise
is present in the sensors' received signals. Thus, comparing the
similarity of the power distributions across the sensors would allow
us to determine whether the source is transmitting or not, without any
\emph{a priori} information of the sensors' power distributions.

This simple intuitive approach is effectively a generalization to the
classical $K$-sample problem, which is to test whether the multiple
samples are drawn from the same unspecified distribution. During the
past decades, a variety of tests have been proposed to solve this
problem. Many of these tests, such as
\cite{kruskal1952use}--\nocite{terry1952some,
  puri1965some,mack1981k,scholz1987k}\cite{murakami2006k},
are based on ranking the samples.
The ranking operation requires the sensors to report all their
observations to the fusion center for calculating the decision
statistic. As a result, privacy protection would be required for the
raw data of all sensors.
Some non-ranking tests, such as ~\cite{quade1966analysis,szekely2004testing},
also have the same requirement of the
raw sensor observations be available at the fusion center, and hence
require complicated privacy protection mechanism.  The decision
statistics of the tests proposed in~\cite{kiefer1959k,chen2014bayesian},
on the other hand, can be expressed as functions of distributed
components that can be calculated at the sensors.  Nonetheless, these
functions have complicated forms, which may require multiple rounds of
obfuscation to protect the privacy of the distributed components.  In
general, the detection performance of all the above tests is analyzed
based on the limiting and/or approximate distributions of the
statistics, and is verified through the simulations with artificial or
real world data sets (see \cite{szekely2004testing}).  While there
appears to be no error-exponent analysis specific for the $K$-sample
problem available in the literature, results for the general composite
hypothesis testing problem~\cite{Zeitouni1991} apply.

We consider the generalization to the basic $K$-sample problem that
the marginal distributions of the sensors' observations do not need to
be exactly the same under the null hypothesis. In
Section~\ref{sec:detection}, we propose a novel specialization of the
composite hypothesis testing formulation to this generalized
$K$-sample problem by ways of a diameter measure that characterizes
the level of similarity between the sensors' marginal distributions.
The proposed formulation allows us to establish the important result
that the marginal types of the sensors' observations are sufficient in
achieving optimal worst-case type-I error exponent, whereas this
result is not readily available from the general large deviation
theory based analysis in~\cite{Zeitouni1991}. This result is critical
to our goal of protecting the privacy of the sensors' data because it
supports the use of a simple zero-modulo-sum (ZMS) obfuscation scheme
to hide the sensors' marginal types when a diameter measure based on
the Hellinger distance is employed, as will be further discussed in
Section~\ref{sec:zms} below.
The performance of a hypothesis test that employs the marginal
sensors' types to form the decision statistic is investigated in
Section~\ref{sec:simulation} by a simulation example of a spectrum
sensing scenario employing the CBRS channel model.

\subsection{Privacy protection} \label{sec:protect}
In many applications, it is desirable to protect the sensors' data
and/or statistics reported to the fusion center as they may expose
private information about the sensors themselves. A number of
mechanisms have been proposed to provide some form of privacy
protection.

An intuitive approach is to homomorphically encrypt the measurement
data and have the message transmission, statistic computation, and
event detection all conducted in the ciphertext domain. For example,
ref. \cite{lihong} designs a received-signal-strength fingerprinting
localization scheme called PriWFL by leveraging the Paillier
cryptosystem to preserve location privacy of users and data privacy of
the service provider. The PriWFL scheme is extended in \cite{mine} to
support channel state information fingerprinting localization and to
give further protection on position privacy of the localization
infrastructure.
The main drawback of homomorphic encryption is its high computational
overheads.

Another approach is based on compressive sensing (CS). Pseudo-random
measurement matrices are employed to linearly encode the sensors'
measurements, which are then recovered at the fusion center.
Based on this approach, a privacy-preserving federated learning (FL)
scheme for spectrum detection in CBRS is proposed
in~\cite{wang2020privacy}, and
a multi-level privacy-preserving scheme for the users with different
privilege levels to acquire and analyze the data is constructed
in~\cite{multics2022}.
A critical issue of the CS approach is how to generate and distribute
the secret measurement matrices. In the above examples, the required
secrecy is generated from channel reciprocity between wireless
transceivers~\cite{wang2020privacy} and from a chaotic
system~\cite{multics2022}. It is however difficult to obtain
verifiable secrecy from both of these mechanisms. 


Another large category of privacy-preserving techniques involves
perturbing the sensors' original data with well designed noise such
that the perturbed data can still yield an acceptable level of
performance. A popular design methodology of perturbation is based on
differential privacy (DP)~\cite{dwork2006}, which aims to constrain
the distance between any pair of outputs provided the input
collections only differ in one data point. Under the DP constraint,
ref.~\cite{wei2020federated} develops a FL scheme called NbAFL, which
lets the fusion center perturb the global model in the downlink
transmission and the users perturb the local models in the uplink
transmission. In~\cite{seif2020wireless}, FL is implemented under the
DP constraint over a Gaussian multiple access channel to extract
privacy benefit from the underlying physical layer characteristics.
The main shortcoming of perturbation is the inevitable performance
degradation caused by the introduced noise. In addition, when the
number of sensors and the dimension of data are large, the DP
guarantee may not be practically sufficient. More importantly, the DP
guarantee is derived from a defender's perspective rather than against
the objective and/or capability of a potential attacker.


\subsection{Zero-modulo-sum obfuscation} \label{sec:zms} 

As discussed in Section~\ref{sec:k-sample}, our main result of the
generalized $K$-sample problem is that the marginal types of the
sensors' observations are sufficient to achieve the optimal worst-case
type-I error exponent.  In particular, if the Hellinger diameter
measure of the sensors' marginal types is employed to construct the
decision statistic used by a test performed at the fusion center, the
resulting statistic can be expressed in terms of the sum of the square
roots of the sensors' marginal types. This simple but key observation
allows us to employ a classical ZMS obfuscation scheme to protect the
sensors' data privacy in lieu of the other approaches with their
respective shortcomings summarized in Section~\ref{sec:protect}.

ZMS obfuscation is widely used in many different applications. We
highlight here some related recent works. A zero-sum (but not modulo
sum) obfuscating mechanism is adopted in \cite{taoshu} as an
intermediate step to achieve privacy-preserving
localization. Ref.~\cite{ukil2010privacy} applies ZMS obfuscation to
perform data aggregation in wireless sensor networks, where the data
are obfuscated in a round-robin order through all sensors. Similarly,
ZMS obfuscation is applied in~\cite{danezis2013smart} to a smart grid,
where data aggregation is conducted with the help of hash
functions. In \cite{modulo}, the protocols of secret sharing and
multi-party anonymous authentication are developed with ZMS
obfuscation, and the detection of dishonest participants is discussed.
Another related work is~\cite{secagg}, in which a secure aggregation
protocol, called SecAgg, is proposed for FL. The protocol utilizes
random numbers generated by pseudo random generators (PRGs) to
obfuscate model updates from the FL participants. The seeds of PRGs
are negotiated via a Diffie-Hellman exchange between each participant
pair, including any malicious participants. 

In our case, each sensor generates a collection of uniform ZMS random
numbers, among which one number is kept secret to the sensor itself,
and other numbers are confidentially sent to other sensors by way of a
public key cryptosystem. Then, each sensor obfuscates its measured
square-root type by calculating the modulo sum of the type, the
self-kept number and the received numbers such that all the
obfuscation can be eventually canceled out at the fusion center.  The
detailed protocol to apply this ZMS obfuscation scheme is discussed in
Section~\ref{sec:protocol}.


In Section~\ref{sec:analysis}, we analytically quantify the privacy
protection performance of the ZMS obfuscation scheme under an ``honest
but curious'' threat model in which the adversary may include external
eavesdroppers, the fusion center, and a subset of sensors all
colluding to estimate the other sensors' marginal types. We apply the
standard attacker-challenger formalism in cryptographic analysis to
show that any probabilistic polynomial time (PPT) attacker cannot
improve the probability of correctly estimating the sensors' marginal
types beyond independent guessing given the information that she can
obtain from her own measurement and that is ``leaked'' to her via the
proposed protocol, provided that the public key cryptosystem to used
distribute the ZMS random numbers is secure under the chosen plaintext
attack (CPA) criterion. 

The prevailing privacy analysis methodology for the ZMS obfuscation
approach against the honest-but-curious attacker is through the notion
of \emph{view}~\cite{secagg,so2022lightsecagg,liu2023long}.  The view
approach essentially establishes that all internal states and received
messages (and hence any estimator generated from this set of
information) of the attacker during the execution of the ZMS protocol
can be emulated by a PPT simulator taking the same set of inputs in
that the distribution of the simulated internal states and messages is
indistinguishable from that of the real ones obtained during the
protocol execution. This leads to the interpretation that the attacker
cannot learn anything new more than its own inputs.  The view notion
is inadequate as a proof of achieving privacy in that it fails to
directly bound the performance of every estimator that the attacker
may construct using the information available to it.  The privacy
analysis in Section~\ref{sec:analysis}, by contrast, gives a direct
and strong bound on the estimation performance of the attacker. This
is particularly important for rigorous integration of the privacy
constraint in the detection design. Specially, this privacy bound
ensures that the additional requirement of privacy protection does not
fundamentally require any tradeoff in achieving the optimal worst-case
type-I error exponent in the generalized $K$-sample problem.

\section{Notation and Assumptions} \label{sec:notation}

\subsection{Basic Notation}

We use uppercase letters and the corresponding lowercase letters to
denote random variables and the values taken by the random variables,
respectively. We use boldface letters to denote an indexed collection
of random variables and values. Script letters are generally reserved
for index sets and alphabets. When an index set is employed as a
subscript, we refer to the collection of random variables (or values)
indexed over the set. For convenience, we slightly abuse notation by
using a single index to also denote a singleton index set containing
only that index. For example, given a sensor network with $K$ sensors,
$\mathcal{K}=\{1,2,\ldots,K\}$ denotes the set of sensor indices,
$\mathcal{X}$ denotes the finite alphabet of sensor measurements,
$\mathbf{X}_k=[X_{k,i}]_{i=1}^t \in\mathcal{X}^t$ denotes the
$t$-length measurement sequence of the $k$th sensor, and
$\mathbf{X}_{\mathcal{K}}=[\mathbf{X}_k]_{k\in\mathcal{K}}
\in\mathcal{X}^{Kt}$ denote the collection of measurement sequences
from all sensors.

For the rest of the paper, we assume the sensor measurement alphabet
$\mathcal{X}$ is finite with $\mathcal{P}(\mathcal{X})$ denoting the
set of distributions (probability mass functions) over
$\mathcal{X}$. The distribution of a random variable $X$ over
$\mathcal{X}$ is denoted by $p_X$. When convenient, we may write a
distribution $p \in \mathcal{P}(\mathcal{X})$ as a vector, i.e.,
$\mathbf{p}=[p(x)]_{x\in\mathcal{X}}$. For any $t$-length measurement
sequence $\mathbf{X}_{k}$,
$\tilde{q}_{\mathbf{X}_{k}}(x)=\frac{1}{t} \cdot \left(\text{number of
    occurrences of~} x \text{~in~} \mathbf{X}_k \right)$ denotes the
\emph{type (empirical distribution)} of $\mathbf{X}_{k}$. The set of
all possible types of $t$-length sequences is denoted by
$\tilde{\mathcal{Q}}_t(\mathcal{X})$. Note that
$\bigcup_{t=1}^{\infty} \tilde{\mathcal{Q}}_t(\mathcal{X})$ is dense
in $\mathcal{P}(\mathcal{X})$. Furthermore, for any
$\tilde{q} \in \tilde{\mathcal{Q}}_t(\mathcal{X})$, we denote its
\emph{type class} by
$\mathcal{T}(\tilde{q}) =\{\mathbf{x} \in \mathcal{X}^{t}:
\tilde{q}_{\mathbf{x}} = \tilde{q}\}$.

A vector of $K$ marginal distributions is denoted by
$\mathbf{p}_{\mathcal{K}} = [p_k]_{k\in\mathcal{K}} \in
\mathcal{P}^K(\mathcal{X})$, with each
$p_k\in\mathcal{P}(\mathcal{X})$.
With a slight abuse of notation, we also use the same notation
$\mathbf{p}_{\mathcal{K}}$ to denote a general joint distribution in
$\mathcal{P}(\mathcal{X^K})$. When necessary to highlight the former
case, we will explicitly state
$\mathbf{p}_{\mathcal{K}} \in \mathcal{P}^K(\mathcal{X})$. The same
convention applies to vectors of marginal types in
$\tilde{\mathcal{Q}}_t^K(\mathcal{X})$ and joint types in
$\tilde{\mathcal{Q}}_t(\mathcal{X}^K)$.

We will use the following \emph{diameter} measure to characterize the
degree of similarity between marginal distributions:
\begin{define} \label{hd}
Let $d:  \mathcal{P}^K(\mathcal{X}) \rightarrow [0,\infty)$. We call
$d(\cdot)$ a diameter measure if 
\begin{itemize}
\item $d(\mathbf{p}_{\mathcal{K}}) =0$ if and only if the marginal
  distributions in $\mathbf{p}_{\mathcal{K}}$ are identical, and
\item $d(\cdot)$ is continuous in $\mathcal{P}^K(\mathcal{X})$.
\end{itemize}
\end{define}
The diamater measure naturally extends to any general
$\mathbf{p}_{\mathcal{K}} \in \mathcal{P}(\mathcal{X}^K)$ in that the
marginals of $\mathbf{p}_{\mathcal{K}}$ are employed when calculating
$d(\mathbf{p}_{\mathcal{K}})$.

For the privacy-preserving protocol and its performance analysis
in~Sections~\ref{sec:protocol} and~\ref{sec:analysis}, we will
specialize to the following choice of the diameter measure based
on the Hellinger distance:
\begin{align}
d(\mathbf{p}_{\mathcal{K}})
&= \sum_{k\in\mathcal{K}}\sum_{l\in\mathcal{K}}d^2_H(p_k,p_l) 
= 
K^{2}- \sum_{x\in\mathcal{X}}
              \left(\sum_{k=1}^{K}\sqrt{{p}_{k}(x)}\right)^2
\label{equ:d}
\end{align}
where the $p_k$'s in~\eqref{equ:d} are the corresponding marginals of
$\mathbf{p}_{\mathcal{K}}$, and for any marginal pair
$p_k, p_l \in \mathcal{P}(\mathcal{X})$,
\begin{align*}
d^2_H(p_k,p_l) 
&= \frac{1}{2}\sum_{x\in\mathcal{X}}\left
  (\sqrt{{p}_{k}(x)}-\sqrt{{p}_{l}(x)}\right)^{2}
\end{align*}
is the Hellinger distance square between them~\cite{hd}.  For
convenience, we will call
the specialized diameter measure in~\eqref{equ:d} the \emph{Hellinger
  diameter}. It is not hard to show that the Hellinger diameter is
bounded, i.e.,  for every 
$\mathbf{p}_{\mathcal{K}} \in \mathcal{P}(\mathcal{X}^K)$, 
$0 \leq d(\mathbf{p}_{\mathcal{K}}) \leq d_{\max}$  with 
\begin{align}
d_{\max} = 
& K(K-1) -  \left\lfloor \frac{K}{|\mathcal{X}|} \right\rfloor (K
                                  - |\mathcal{X}| +
                                  K\bmod |\mathcal{X}|).
\label{equ:dmax}
\end{align}

For each $x_0 \in \mathcal{X}$, the indicator function
$\delta_{x_0}(x) = 1$ if $x = x_0$, and $\delta_{x_0}(x) = 0$
otherwise. This definition naturally extends when the arguments are
collections. Any other function $F(X)$ in this paper, unless otherwise
stated, is assumed stochastic. That is, $F(X)$ is random and is
conditionally independent of all other random variables given its
input $X$.

\subsection{Fixed-point Arithmetics}
Let $N$ be a positive integer and $\mathcal{N}_{m}$ be the collection
of all $m$-bit fixed-point numbers that quantize the interval
$[0, N)$, i.e.,
$\mathcal{N}_{m} = \left\{0,\frac{N}{2^m},\ldots,\frac{(2^m-1)N}{2^m}
\right\}$. We ``quantize'' each
$\tilde{q} \in \mathcal{\tilde Q}_t(\mathcal{X})$ by mapping
$\sqrt{\tilde{q}(x)}$, for each $x \in \mathcal{X}$, to its closest
value in $\mathcal{N}_{m}$.  The set of these quantized square-root
types is denoted by $\mathcal{Q}_t(\mathcal{X})$. More specifically,
every $\tilde{q} \in \mathcal{\tilde Q}_t(\mathcal{X})$ is mapped to a
$q \in \mathcal{Q}_t(\mathcal{X})$ that satisfies
$q(x) \in \mathcal{N}_m$, $0 \leq q(x) < 1$, and
$|\sqrt{\tilde{q}(x)} - q(x)| \leq 2^{-m-1}$ for every
$x \in \mathcal{X}$. Note that $q^2$ may not be a true type; however
it must satisfy
$\left|\sum_{x\in\mathcal{X}} q^2(x) - 1\right| \leq 2^{-m}
|\mathcal{X}|$. We assume that $m$ is chosen large enough to guarantee
$q^2$ is sufficiently close to a true type.  We also note that
$|\mathcal{Q}_t(\mathcal{X})| \leq |\mathcal{\tilde Q}_t(\mathcal{X})|
\leq (t+1)^{|\mathcal{X}|}$~\cite[Theorem~11.1.1]{info}. With a large
enough $m$ (i.e., $m = \mathcal{O}(\log_2 t)$), we assume
$|\mathcal{Q}_t(\mathcal{X})|$ to have the same order as
$|\mathcal{\tilde Q}_t(\mathcal{X})|$.

Let $\oplus$ and $\ominus$ denote addition and subtraction modulo $N$
over the fixed-point numbers in $\mathcal{N}_{m}$, respectively. Note
that $\mathcal{N}_{m}$ is closed under both the operations. If the
operands of $\oplus$ or $\ominus$ are indexed collections, it means
performing the $\oplus$ or $\ominus$ operation elementwise. For any
$x_0 \in \mathcal{N}_{m}$, the indicator function
$\delta_{x_0}(x) = \delta_0(x \ominus x_0)$ for all
$x \in \mathcal{N}_{m}$. We will omit the subscript $0$ in $\delta_0$
and write $\delta(x \ominus x_0)$ as the indicator function.

For a collection of random variables $[Y_k(x)]$ on $\mathcal{N}_{m}$
indexed by $k \in \mathcal{K}$ and $x\in\mathcal{X}$, we write
$\mathbf{Y}_{\mathcal{L}}(x) = [Y_{k}(x)]_{k\in\mathcal{L}}$ and
$\mathbf{Y}_{\mathcal{L}} =
[\mathbf{Y}_{\mathcal{L}}(x)]_{x\in\mathcal{X}}$ for any
$\mathcal{L} \subseteq \mathcal{K}$. We use the notation
$\mathbf{Y}_{\mathcal{K}} \sim {u}(\mathcal{N}^{K|\mathcal{X}|}_{m})$
to say $\mathbf{Y}_{\mathcal{K}}$ is uniformly distributed on
$\mathcal{N}^{K|\mathcal{X}|}_{m}$, i.e., all elements in
$\mathbf{Y}_{\mathcal{K}}$ are independent and identically distributed
(i.i.d.) according to ${u}(\mathcal{N}_{m})$.
Similarly, for a collection of random variables $[R_{k,l}(x)]$ on
$\mathcal{N}_{m}$ indexed by $(k,l)\in \mathcal{K}^2$ and
$x\in\mathcal{X}$, we write
$\mathbf{R}_{\mathcal{I},\mathcal{J}}(x) =
[R_{k,l}(x)]_{k\in\mathcal{I},l\in\mathcal{J}}$ and
$\mathbf{R}_{\mathcal{I},\mathcal{J}} =
[\mathbf{R}_{\mathcal{I},\mathcal{J}}(x)]_{x\in\mathcal{X}}$ for any
$\mathcal{I},\mathcal{J} \subseteq \mathcal{K}$. In addition, we
define
$\mathbf{\Sigma}_{\mathbf{Y}_\mathcal{L}}(x) =
\bigoplus_{k\in\mathcal{L}}Y_{k}(x)$,
$\mathbf{\Sigma}_{\mathbf{Y}_\mathcal{L}} =
[\Sigma_{\mathbf{Y}_\mathcal{L}}(x)]_{x\in\mathcal{X}}$,
$\mathbf{\Sigma}_{\mathbf{R}_{\mathcal{I},\mathcal{J}}}(x) =
[\bigoplus_{k\in\mathcal{I}}R_{k,l}(x)]_{l\in\mathcal{J}}$, and
$\mathbf{\Sigma}_{\mathbf{R}_{\mathcal{I},\mathcal{J}}} =
[\mathbf{\Sigma}_{\mathbf{R}_{\mathcal{I},\mathcal{J}}}(x)]_{x\in\mathcal{X}}$.

\subsection{Cryptographic Assumptions} \label{sec:crypto_assump}

We review here some standard cryptographic concepts and assumptions
useful for constructing and analyzing our privacy preserving mechanism
in later sections. In particular, we will follow the standard
cryptographic methodology that defines an attack experiment involving
two interactive parties, namely a \emph{challenger} and an
\emph{attacker}, and evaluates the probability advantage of the
attacker winning the experiment. To that end, the attack
experiment will be based on the well known chosen-plaintext attack (CPA)
model~\cite{intro}.

Consider 
\begin{itemize} 
\item a public-key cryptographic scheme $\Pi= (S,E,D)$ with security
  parameter $n$, and
\item a probabilistic, polynomial-time (PPT) attacker, whose running
  time is polynomial in $n$.
\end{itemize} 
The functions $S$, $E$, and $D$ represent the algorithms of key
generation, encryption, and decryption, respectively. The security
parameter $n$ is usually formulated in the unary form as $1^{n}$, a
string of $n$ $1$'s.
In
this paper, we restrict each plaintext in $\Pi$ to be an $m$-bit
message corresponding to a fixed-point number in $\mathcal{N}_m$, and
the resulting ciphertext space is denoted by $\mathcal{C}$. As shown
in Figure~\ref{fig:cpa}, the CPA experiment is defined as
follows: 
\begin{figure}
  \centering
  \includegraphics[width=0.42 \textwidth]
  {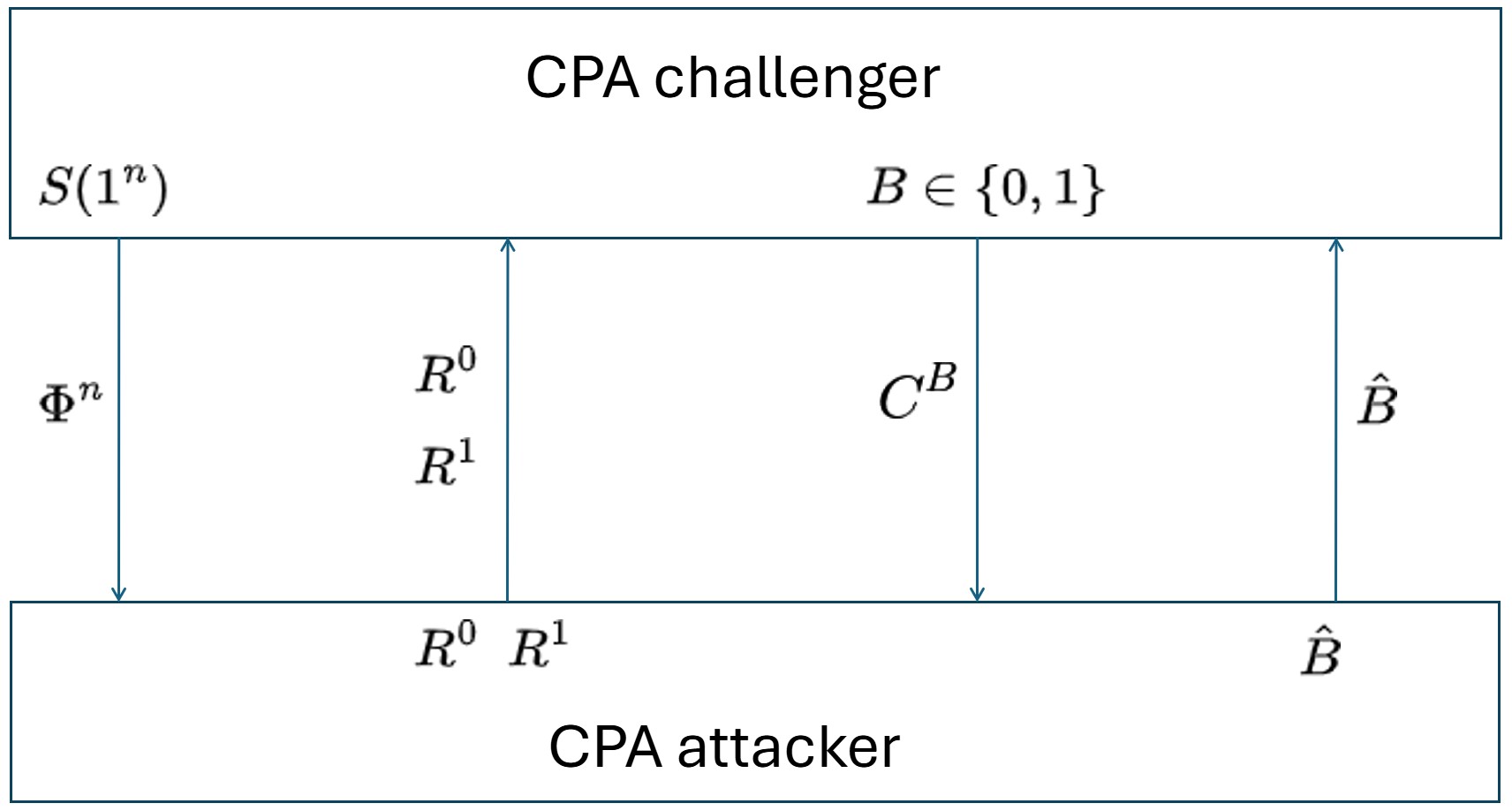}
  \caption{The CPA experiment.}\label{fig:cpa}
\end{figure}
\begin{enumerate}
\item The challenger runs $S(1^{n})$ to generate a pair of public key
  $\Phi^{n}\in\mathcal{E}_n$ and private key
  $\Psi^{n}\in\mathcal{D}_n$, and then gives $\Phi^n$ to the
  attacker. This means that the attacker can encrypt any plaintext by
  executing $E(\cdot;\Phi^n)$ herself.

\item The attacker generates a pair of challenge messages
  $[R^{0},R^{1}] \sim {u}(\mathcal{N}^2_m)$, and gives them to her
  challenger.

\item The challenger selects an independent random bit $B\in\{0,1\}$
  with equal probabilities, computes the ciphertext
  $C^B = E(R^{B};\Phi^n)\in\mathcal{C}$, and gives $C^B$ to
  the attacker.

\item The attacker outputs a bit
  $\hat{B} = \hat{B}(C^B,R^0,R^1,\Phi^n)$ as her estimate of
  $B$. Then, she reports $\hat{B}$ to her challenger.
\end{enumerate} 
If $\hat{B} = B$, it is said that the attacker wins the CPA
experiment.  Define the probability advantage of the attacker winning
the CPA experiment as
\begin{align}\label{equ:adv}
  F_{\mathrm{CPA}}(n) = \Pr(\hat{B} = B)-\frac{1}{2}.
\end{align} 

Based on the CPA experiment described above, we define the CPA
security of a public key scheme as follows~\cite{intro}:
\begin{define}
  A public key encryption scheme $\Pi= (S,E,D)$ is said to be
  CPA-secure if there exists a negligible function\footnote{A function
    $\varepsilon(n)$ is \emph{negligible} if for every polynomial
    function $\mathrm{poly}(n)$, there exists an $N$ such that for all
    integers $n \geq N$, it holds that
    $\varepsilon(n)\leq 1/\mathrm{poly}(n)$~\cite{intro}.}
  $\varepsilon(n)$ such that $F_{\mathrm{CPA}}(n)\leq \varepsilon(n)$
  for all PPT attackers.
\end{define}
Note that many practical public-key cryptographic schemes, such as
ElGamal~\cite{elgamal} and RSA-OAEP~\cite{rsa}, have been shown to be
CPA-secure.

For the rest of this paper, we make the following assumptions on the
cryptographic resources available to the sensors.
\begin{assump} \textbf{(Cryptographic resource)} \label{a:resource} 
  The same CPA-secure public-key cryptographic scheme $\Pi=(S,E,D)$
  with security parameter $n$ is made available to each sensor, which 
  maintains its own pair of public and private keys. The public keys
  for encrypting messages are known to all entities in the network
  while the private keys for decrypting messages are secret.
\end{assump}


\begin{assump} \textbf{(Independent Encryptions)} \label{a:crypto}
  Every use of the encryption function is conditionally independent of
  all other uses given the plaintexts and public keys. More precisely,
  let $M$ be any positive integer and $C_k={E}(R_k;\Phi_k^n)$ for
  $k=1,2,\ldots,M$. Then for any $c_k\in\mathcal{C}$,
  $r_k\in\mathcal{N}_m$, and $\phi_k\in\mathcal{E}_n$, 
\begin{align*}
  p_{C_1,\ldots,C_M\mid R_1,\ldots,R_M,\Phi_1^n,\ldots,\Phi_M^n}
  (c_1,\ldots,c_M\mid r_1,\ldots,r_M, \hspace{30pt} \\
  \phi_1^n,\ldots,\phi_M^n)
=\prod_{k=1}^{M}p_{C_k\mid R_k,\Phi_k^n}(c_k\mid r_k,\phi_k^n).
\end{align*}
\end{assump}

Furthermore, we adopt the following ``bar'' notation to simplify
discussion in later sections. For any
$\mathcal{I},\mathcal{J}\subseteq\mathcal{K}$ and a collection of
plaintexts
$\mathbf{R}_{\mathcal{I},\mathcal{J}}=[R_{k,l}(x)]_{k\in\mathcal{I},l\in\mathcal{J},x\in\mathcal{X}}$,
the corresponding ``barred'' collection is defined as ${\mathbf{\bar
    R}}_{\mathcal{I},\mathcal{J}}=[R_{k,l}(x)]_{k\in\mathcal{I},l\in\mathcal{J}\setminus
  k,x\in\mathcal{X}}$. Given a collection of public keys
$\mathbf{\Phi}_{\mathcal{J}}^n=[\Phi_l^n]_{l\in\mathcal{J}}$ with
$\Phi_l^n\in\mathcal{E}_n$, we write
${\mathbf{\bar C}}_{\mathcal{I},\mathcal{J}}= E(\mathbf{\bar
  R}_{\mathcal{I},\mathcal{J}};\mathbf{\Phi}_{\mathcal{J}}^n)$ as a
shorthand for the operation of computing the ciphertext
$C_{k,l}(x)=E(R_{k,l}(x);\Phi_{l}^n)$ for each $k\in\mathcal{I}$,
$l\in\mathcal{J}\setminus k$, $x\in\mathcal{X}$, and outputting the
whole ciphertext collection
$\mathbf{\bar
  C}_{\mathcal{I},\mathcal{J}}=[C_{k,l}(x)]_{k\in\mathcal{I},l\in\mathcal{J}\setminus
  k,x\in\mathcal{X}}$. Similarly, given a collection of private keys
$\mathbf{\Psi}_{\mathcal{J}}^n=[\Psi_l^n]_{l\in\mathcal{J}}$ with
$\Psi_l^n\in\mathcal{D}_n$, we write
${\mathbf{\bar R}}_{\mathcal{I},\mathcal{J}}=
D(\bar{\mathbf{C}}_{\mathcal{I},\mathcal{J}};\mathbf{\Psi}^n_{\mathcal{J}})$
as a shorthand for the operation of computing the plaintext
$R_{k,l}(x)=D(C_{k,l}(x);\Psi_{l}^n)$ for each $k\in\mathcal{I}$,
$l\in\mathcal{J}\setminus k$, $x\in\mathcal{X}$, and outputting the
whole plaintext collection
${\mathbf{\bar
    R}}_{\mathcal{I},\mathcal{J}}=[R_{k,l}(x)]_{k\in\mathcal{I},l\in\mathcal{J}\setminus
  k,x\in\mathcal{X}}$. It is obvious that if
$\mathcal{I}\cap\mathcal{J}=\emptyset$, then
$\mathbf{R}_{\mathcal{I},\mathcal{J}}={\mathbf{\bar
    R}}_{\mathcal{I},\mathcal{J}}$ and
$\mathbf{C}_{\mathcal{I},\mathcal{J}}={\mathbf{\bar
    C}}_{\mathcal{I},\mathcal{J}}$.

\section{Event Detection} \label{sec:detection} 

In this section, we introduce the formulation of the generalized
$k$-sample problem, propose a test that facilitates privacy
protection, and show that the proposed test can achieve good detection
performance.

\subsection{Problem Formulation} \label{sec:formulation}
As mentioned before, we consider a network of $K$ sensors together
with a fusion center that aggregates information from the sensors to
perform detection of a target event. The network size $K$ is assumed
to be fixed and known to all entities in the sensor
network. Communications between the fusion center and the sensors are
assumed public. All messages sent by any entity are observable by all
entities within and outside of the network. Moreover, all entities
agree on a positive number $N>K$, a large enough integer $m$, and thus
the resulting fixed-point domain $\mathcal{N}_m$ beforehand.

Let $\mathbf{X}_{k}$ be a $t$-length measurement vector made by the
$k$th sensor, for $k\in\mathcal{K} = \{1,2,\ldots,K\}$. The elements
of $\mathbf{X}_{k}$ are i.i.d. according to the marginal distribution
$p_{\theta,k}$ over the common finite alphabet $\mathcal{X}$. The
parameter $\theta\in\{0,1\}$ represents the system state indicating
whether the target event happens ($\theta=1$) or not ($\theta=0$). The
distributions $p_{0,k}$ and $p_{1,k}$ may contain \emph{private}
information about the $k$th sensor. For convenience, we write the
distributions as vectors:
$\mathbf{p}_{0,k}=[p_{0,k}(x)]_{x\in\mathcal{X}}$ and
$\mathbf{p}_{1,k}=[p_{1,k}(x)]_{x\in\mathcal{X}}$, and consider the
joint distributions
$\mathbf{p}_{0,\mathcal{K}} \in \mathcal{P}(\mathcal{X}^K)$ and
$\mathbf{p}_{1,\mathcal{K}} \in \mathcal{P}(\mathcal{X}^K)$, whose
marginals are respectively given by
$[\mathbf{p}_{0,k}]_{k\in\mathcal{K}}$ and
$[\mathbf{p}_{1,k}]_{k\in\mathcal{K}}$.  We assume that neither
$\mathbf{p}_{0,\mathcal{K}}$ nor $\mathbf{p}_{1,\mathcal{K}}$ is
known. However, it is known that they satisfy the condition
\begin{equation} \label{equ:restriction}
d\left(\mathbf{p}_{0,\mathcal{K}} \right)\leq
d_{0}< d_{1}\leq
d \left(\mathbf{p}_{1,\mathcal{K}} \right)
\end{equation}
for some $0 \leq d_{0} < d_{1}$, where $d(\cdot)$ is a
diameter measure satisfying the conditions in Defintion~\ref{hd}.

The objective of the fusion center is to make a decision on the system
state $\theta$ based on the whole set of sensor measurements
$\mathbf{X}_{\mathcal{K}}\in \mathcal{X}^{Kt}$. In this section, we
temporarily ignore any privacy concern and assume that any necessary
statistics (e.g., $\mathbf{X}_{\mathcal{K}}$) for decision are made
available to the fusion center. In Section~\ref{sec:protocol}, we will
present a protocol to protect privacy specifically for the application
of the following binary hypothesis test at the fusion center to make a
decision on $\theta$: The $k$th sensor calculates the type
$\tilde{Q}_k = \tilde{q}_{\mathbf{X}_k}$ from its measurement sequence
$\mathbf{X}_k$, and sends $\tilde{Q}_{k}$ to the fusion center. The
fusion center collects the whole set of sensor types
$\mathbf{\tilde{Q}}_{\mathcal{K}} \in
\tilde{\mathcal{Q}}_t^K(\mathcal{X})$ from the $K$ sensors, calculate
the diameter of $\mathbf{\tilde{Q}}_{\mathcal{K}}$, and then decides
\begin{equation}\label{equ:hypo}
\begin{aligned}
H_{0} &: \theta=0\quad \quad \text{if} \
          d({\mathbf{\tilde Q}}_{\mathcal{K}}) < \gamma \\
H_{1} &:  \theta=1\quad \quad \text{if} \
d({\mathbf{\tilde Q}}_{\mathcal{K}}) \geq \gamma
\end{aligned}
\end{equation}
where $\gamma \geq 0$ is a detection threshold. Note that the decision
statistic $d({\mathbf{\tilde Q}}_{\mathcal{K}})$ employed in the test
above depends only on the marginal types
$[\tilde{Q}_k]_{k \in \mathcal{K}}$, each of which can be calculated
at the corresponding sensor based on its own measurement vector.

\subsection{Error Exponents} \label{sec:errexp} 

In this section, we analyze the detection performance of the binary
test~\eqref{equ:hypo}.  To that end, define the following two sets of
joint distributions:
\begin{align*}
\mathcal{P}_{0,\mathcal{K}} &=\{\mathbf{p}_{\mathcal{K}}
\in\mathcal{P}(\mathcal{X}^K):d(\mathbf{p}_{\mathcal{K}})\leq d_{0}\}\\
\mathcal{P}_{1, \mathcal{K}} &=\{\mathbf{p}_{\mathcal{K}}
\in\mathcal{P}(\mathcal{X}^K):d(\mathbf{p}_{\mathcal{K}})\geq d_{1}\}.
\end{align*}
For a binary hypothesis test with acceptance region
$\mathcal{R}_t\subseteq \mathcal{X}^{Kt}$, 
we define the worst-case error probability of the first type as
\begin{equation}\label{equ:mu}
\mu_{t}=\max_{\mathbf{p}_{0,\mathcal{K}}\in\mathcal{P}_{0,
    \mathcal{K}}}
\mathbf{p}_{0,\mathcal{K}}(\mathcal{R}_t^c)
\end{equation}
and the worst-case error probability of the second type as
\begin{equation}\label{equ:lambda}
\lambda_{t}=\max_{\mathbf{p}_{1,\mathcal{K}}\in\mathcal{P}_{1,
    \mathcal{K}}}
\mathbf{p}_{1,\mathcal{K}}(\mathcal{R}_t).
\end{equation}
Based on these error probabilities, we define an achievable error
exponent pair as follows:
\begin{define} 
  A non-negative error exponent pair $(\alpha,\beta)$ is said to be
  \emph{achievable} if there is a sequence of acceptance regions
such that
\begin{align}
\label{equ:alpha}
\liminf_{t\rightarrow \infty}-\frac{1}{t}\log_2 \mu_{t} &\geq \alpha
\\
\label{equ:beta}
\liminf_{t\rightarrow \infty}-\frac{1}{t}\log _2\lambda_{t} &\geq \beta.
\end{align}
A non-negative error exponent of the first type $\alpha$ is said to be
achievable if there is a sequence of acceptance regions such
that~\eqref{equ:alpha} is satisfied and
$\lim_{t\rightarrow \infty} \lambda_{t} =0$.
\end{define}
Clearly, if $(\alpha,\beta)$ is achievable and $\beta>0$, then
$\alpha$ is achievable. Thus, $(\alpha,\beta)$ being an achievable
error exponent pair is a stronger condition.

For $\mathbf{p}_{\mathcal{K}} \in \mathcal{P}(\mathcal{X}^K)$, define
\begin{align*} 
\Delta_0(\mathbf{p}_{\mathcal{K}} ) &=
  \min_{\mathbf{p}_{0,\mathcal{K}}\in\mathcal{P}_{0,\mathcal{K}}}
  D(\mathbf{p}_{\mathcal{K}} \| \mathbf{p}_{0,\mathcal{K}}) \\
\Delta_1(\mathbf{p}_{\mathcal{K}} ) &=
  \min_{\mathbf{p}_{1,\mathcal{K}}\in\mathcal{P}_{1,\mathcal{K}}}
  D(\mathbf{p}_{\mathcal{K}} \| \mathbf{p}_{1,\mathcal{K}})
\end{align*} 
where $D(\cdot\|\cdot)$ is the Kullback-Leibler (KL) divergence.  For
$\gamma \geq 0$, define the function\footnote{By convention, we set
  the minimum (or infimum) value over an empty set to be $\infty$.}
\begin{equation*}
\alpha_*(\gamma) =
\min_{\mathbf{p}_{\mathcal{K}} \in \mathcal{P}(\mathcal{X}^K): d(\mathbf{p}_{\mathcal{K}}) \geq \gamma}
\Delta_0(\mathbf{p}_{\mathcal{K}}).
\end{equation*}
Further, for $\alpha \geq 0$, define
\begin{align*}
\gamma_*(\alpha) 
&= 
\inf\{ \gamma \geq 0: \alpha_*(\gamma) \geq \alpha\} ,
\\
\beta_*(\alpha) 
&=
\inf_{\mathbf{p}_{\mathcal{K}}\in \mathcal{P}(\mathcal{X}^K): d(\mathbf{p}_{\mathcal{K}}) < \gamma_*(\alpha)}
\Delta_1(\mathbf{p}_{\mathcal{K}}),
\\
\beta^*(\alpha) 
&=
\inf_{\mathbf{p}_{\mathcal{K}}\in \mathcal{P}(\mathcal{X}^K): \Delta_0(\mathbf{p}_{\mathcal{K}}) < \alpha}
\Delta_1(\mathbf{p}_{\mathcal{K}}).
\end{align*}

The formulation presented above is a specialization of the compositie
hypothesis testing framework to the generalized $K$-sample problem
using a diameter measure to characterize the degree of similarity of
the sensors' marginal types. The restriction imposed
by~\eqref{equ:restriction} allows us to consider non-trivial
worst-case type-I and type-II error exponents in~\eqref{equ:alpha}
and~\eqref{equ:beta}. That in turn allows us to describe the optimal
detection performance as the boundary of the achievable region of
error exponent pairs.  More importantly, all these conveniences lead
us to the following theorem which shows that the proposed
test~\eqref{equ:hypo} gives good detection performance. Note that this
result is difficult to obtain directly using the large deviation
analysis on the general compositie hypothesis testing formulation
in~\cite{Zeitouni1991}.
\begin{theorem} \label{thm:errexp} 
  Suppose $0 \leq d_{0} < d_{1}$. Then
\begin{enumerate}[label=(\roman*)]
\item $\beta_*(\alpha) \leq \beta^*(\alpha)$,
\item $\beta^*(\alpha) > 0$ if and only if $\beta_*(\alpha) > 0$,
\item $\beta^*(\alpha) = \sup\{ \beta: (\alpha,\beta) \text{~is
  achievable}\}$,
\item $\alpha_*(d_1) = \sup\{ \alpha: \alpha \text{~is
  achievable}\}$, and
\item the test~\eqref{equ:hypo} achieves the error exponent pair
  $(\alpha, \beta_*(\alpha))$ and the optimal error exponent
  $\alpha_*(d_1)$.
\end{enumerate}
\end{theorem}
\begin{proof}
  The proof of the theorem is given in Appendix~\ref{app:EXPproof}.
\end{proof}
As shown in the proof of part (iii) in Appendix~\ref{app:EXPproof},
The Hoeffding test~\cite{Hoeffding1965} using the decision statistic
$\Delta_0({\mathbf{\tilde Q}}_{\mathcal{K}})$, where
${\mathbf{\tilde Q}}_{\mathcal{K}} \in
\tilde{\mathcal{Q}}_t(\mathcal{X}^K)$ is the joint type of all sensor
measurements $\mathbf{X}_{\mathcal{K}}$ (see~\eqref{equ:besthypo}),
can achieve the best error exponent pair $(\alpha, \beta^*(\alpha))$.
However, $\mathbf{X}_{\mathcal{K}}$ must be made available at the
fusion center in order to calculate
$\Delta_0({\mathbf{\tilde Q}}_{\mathcal{K}})$. This in turn makes
protecting private information of the individual sensors much more
difficult.

The test~\eqref{equ:hypo} generally does not achieve the best error
exponent pair $(\alpha, \beta^*(\alpha))$, even when the joint
distributions in the sets $\mathcal{P}_{0,\mathcal{K}}$ and
$\mathcal{P}_{1,\mathcal{K}}$ are restricted to products of marginals
(i.e., the observations are independent across the sensors). To see
that, consider the simple case where $K=2$, $d(\cdot)$ is the
Hellinger diameter, $d_0=0$, $d_{\max}=2$, and both $X_1$ and $X_2$
are independent binary random variables with $p_{X_1}(1) = q_1$ and
$p_{X_2}(1) = q_2$. The joint distribution $\mathbf{p}_{\mathcal{K}}$
is parameterized by $(q_1,q_2)$. In this case,
$d(\mathbf{p}_{\mathcal{K}}) = d(q_1,q_2) = 2\left(1-\sqrt{q_1q_2} -
  \sqrt{(1-q_1)(1-q_2)}\right)$,
$\Delta_0(\mathbf{p}_{\mathcal{K}}) = \Delta_0(q_1,q_2) =
2H_2\left(\frac{q_1+q_2}{2}\right) -H_2(q_1)-H_2(q_2)$, and
$\alpha_*(\gamma) = 2H_2\left(
  \frac{\gamma}{2}(1-\frac{\gamma}{4})\right) - H_2\left(
  \gamma(1-\frac{\gamma}{4})\right)$, where $H_2(\cdot)$ is the binary
entropy function. For any $d_1>0$, brute-force searching for the
minimum values of
$\Delta_1(\mathbf{p}_{\mathcal{K}})=\Delta_1(q_1,q_2)$ over the
respective boundaries $\Delta_0(\mathbf{p}_{\mathcal{K}}) = \alpha$
and $d(\mathbf{p}_{\mathcal{K}}) = \gamma_*(\alpha)$ numerically
calculates $\beta^*(\alpha)$ and $\beta_*(\alpha)$ for
$0<\alpha \leq \alpha_*(d_1)$.  This calculation reveals that
$\beta^*(\alpha)>\beta_*(\alpha)$ for $0<\alpha < \alpha_*(d_1)$.

While suboptimal in the stronger sense of achieving
$(\alpha, \beta^*(\alpha))$, Theorem~\ref{thm:errexp}(ii) ensures that
the test~\eqref{equ:hypo} is able to achieve a positive error exponent
pair whenever the Hoeffding test can. In addition,
Theorem~\ref{thm:errexp}(v) also asserts that the
test~\eqref{equ:hypo} is optimal in the weaker sense that it can
achieve the best error exponent of the first kind. The fact that the
test~\eqref{equ:hypo} using only marginal types is sufficient to
achieve the weaker optimality can be regarded as an inherent property
of the generalized $K$-sample problem as the result holds for any
diameter measure. The main advantage of using the
test~\eqref{equ:hypo} is that a simple privacy-preserving protocol can
be developed to support calculating the decision statistic at the
fusion center when the Hellinger diameter is used in the test.  The
details of the protocol will be discussed in
Section~\ref{sec:protocol}. In summary, Theorem~\ref{thm:errexp}
implies that the additional requirement of privacy protection does not
fundamentally require any tradeoff in the weaker optimal detection
performance in the generalized $K$-sample problem.

\section{Privacy Preserving Protocol} \label{sec:protocol} 

Henceforth, we consider the use of the Hellinger diameter~\eqref{equ:d} in the test~\eqref{equ:hypo} for privacy protection.  To
perform the test~\eqref{equ:hypo}, the sensors must send their
respective types to the fusion center in the form of messages over the
sensor networks. As discussed before, the measurement distributions of
each sensor may contain private information about that sensor. Our
privacy goal is to protect this private information from adversaries
both internal and external to the sensor network. The type
$\tilde{Q}_k$ that the $k$th sensor sends to the fusion center in the
test~\eqref{equ:hypo} is an estimate of the measurement distribution
$p_{0,k}$ or $p_{1,k}$ of the sensor. Thus, we must protect the types
$\tilde{Q}_{k}$ from any adversaries. That means no entities other
than the $k$th sensor should have access to $\tilde{Q}_{k}$. We will
provide a more precise and quantitative specification of this notion of
privacy protection later in Section~\ref{sec:analysis}.  In this
section, we specify the privacy threat model and describe a simple
protocol based on public key cryptography to protect
$\mathbf{\tilde{Q}}_{\mathcal{K}}$ from any adversary under the threat
model with no loss in weak optimal detection performance as discussed
in Section~\ref{sec:errexp}.

\subsection{Threat Model}
\label{sec:threat_model}

We consider a threat model in which potential adversaries may be an
outside eavesdropper, the fusion center, and/or a subset of the $K$
sensors.  We restrict these adversaries to be ``honest but curious.''
That means any adversary, while attempting to obtain information about
the measurement distributions of the sensors, will not act in any way
that may disrupt proper execution of the hypothesis
test~\eqref{equ:hypo} by the fusion center. For example, no adversary
may inject messages containing false information (or no information)
about the set of types $\mathbf{\tilde{Q}}_{\mathcal{K}}$ that may
cause the test~\eqref{equ:hypo} to fail.

As discussed in Section~\ref{sec:formulation}, we assume all messages
passed between the fusion center and the sensors are available to all
entities under this thread model. No raw measurements, i.e.,
$\mathbf{X}_{\mathcal{K}}$, are sent to the fusion center, which
performs the test~\eqref{equ:hypo} based solely on the messages
that it receives from the sensors. In addition to observing the
messages in the network, an adversarial sensor obviously has access to
its own measurements. 

We allow the adversaries to collude in that they may share all network
messages and sensor measurements among themselves. In this sense, it
is more convenient to consider all colluding adversaries as a single
adversarial entity (\emph{the attacker}) that has access to all the
messages and sensor measurements available to the set. For the rest of
the paper, we will denote the set of adversarial sensors by the index
subset $\mathcal{L} \subseteq \mathcal{K}$. Hence, the attacker has
access to all network messages and the measurement collection
$\mathbf{X}_{\mathcal{L}}$. Also, we assume that $\mathcal{L}$ is
known to the attacker but not to any nonadversarial sensors in
$\mathcal{K} \setminus \mathcal{L}$. We will describe the exact
contents of the network messages and a more precise model of how the
attacker may act in Section~\ref{sec:analysis} after the detailed
protocol steps are laid out below.

\subsection{Privacy-preserving
  Protocol} \label{sec:protocol_description}

The basic idea of the proposed privacy-preserving protocol is to let
the sensors use secret random numbers in $\mathcal{N}_m$ to obfuscate
the messages that report their observed types to the fusion center. To
facilitate the obfuscation operation, the type information is also
quantized to $\mathcal{N}_m$. The modulo-sum of the random numbers is
zero, and hence the obfuscation cancels when the fusion center
combines the messages to perform the hypothesis
test~\eqref{equ:hypo}. Since all network messages are public, the
secret random numbers need to be protected from the attacker via
public key cryptography.

There are three phases in the proposed protocol. In the first phase,
each sensor generates its key pair, and sends the public key to the
other sensors.  In the second phase, each sensor generates a set of
random numbers and encrypt them into ciphertexts, which are then sent
to other sensors in the network. Each sensor then decrypts the
ciphertexts to recover the secret random numbers designated to it.  In
the last phase, each sensor uses the set of secret random numbers
obtained in the first phase to obfuscate its observed type, and then
sends the obfuscated messages to the fusion center. The fusion center
employs the whole collection of messages received from all the sensors
to calculate the decision statistic to perform the hypothesis
test~\eqref{equ:hypo}.  The pseudo code shown in
Algorithm~\ref{alg:pride} summarizes the following detailed steps in
the three phases of the proposed protocol:
\begin{algorithm}
\caption{Privacy-preserving protocol} \label{alg:pride}
\begin{algorithmic}[1]
\Require The public-key cryptographic scheme $\Pi=(S,E,D)$ and the
  set of sensor measurements types $\mathbf{X}_{\mathcal{K}}$

\Ensure The decision statistic $\tilde{d}(\mathbf{G}_{\mathcal{K}})$
  required for performing the hypothesis test~\eqref{equ:hypo} at
  the fusion center

\hspace*{-40pt} \emph{Phase 1)}
\For {each $k \in\mathcal{K}$}
\State \hspace*{-10pt} The $k$th sensor:
\State runs $S(1^n)$ to generate the key pair $(\Phi_k^n, \Psi_k^n)$
\State sends the public key  $\Phi_k^n$ to all other sensors
\State calculates its quantized square-root type $Q_k$ from $\mathbf{X}_k$
\EndFor

\hspace*{-40pt} \emph{Phase 2)}
\For {each $x\in\mathcal{X}$}
\For {each $k \in\mathcal{K}$}
\State \hspace*{-10pt} The $k$th sensor:
\State generates $u(\mathcal{N}_m)$-i.i.d. 
    $[R_{k,l}(x)]_{l \in \mathcal{K}\setminus k}$
\State calculates $R_{k,k}(x)$ according to~\eqref{equ:Rkk}
\For {each $l\in\mathcal{K}\setminus k$}
\State encrypts $C_{k,l}(x) = E(R_{k,l}(x); \Phi_{l}^n)$
\State sends the ciphertext $C_{k,l}(x)$ to the $l$th sensor
\EndFor
\EndFor
\For {each $k \in\mathcal{K}$}
\State \hspace*{-10pt} The $k$th sensor:
\For {each $l\in\mathcal{K}\setminus k$}
\State receives $C_{l,k}(x)$  from the $l$th sensor
\State decrypts $R_{l,k}(x) = D(C_{l,k}(x);\Psi_{k}^n)$ 
\EndFor
\State  calculates $\Sigma_{\mathbf{R}_{\mathcal{K},k}}(x)$
\EndFor
\EndFor

\hspace*{-40pt} \emph{Phase 3)}
\For {each $k \in\mathcal{K}$}
\State \hspace*{-10pt} The $k$th sensor:
\For {each $x\in\mathcal{X}$}
\State calculates $G_k(x)$ according to~\eqref{equ:Gk}
\State sends  the obfuscated message $G_k(x)$ to the fusion center
\EndFor
\EndFor
\State The fusion center:
\State \hspace*{10pt}receives the whole set of obfuscated messages
$\mathbf{G}_{\mathcal{K}}$ from all $K$ sensors
\State \hspace*{10pt}calculates $\tilde{d}(\mathbf{G}_{\mathcal{K}})$
according to~\eqref{d}
\State \Return $\tilde{d}(\mathbf{G}_{\mathcal{K}})$
\end{algorithmic}
\end{algorithm}

\paragraph*{Phase 1} For each $k\in\mathcal{K}$, the $k$th sensor runs
$S(1^n)$ to generate the key pair $(\Phi_k^n, \Psi_k^n)$, sends the
public key $\Phi_k^n$ to all other sensors. Then, the $k$th sensor
calculates $\tilde{Q}_k = p_{\mathbf{X}_k}$ from the $t$-length
observation vector $\mathbf{X}_k$, and for each $x\in\mathcal{X}$, it
quantizes $\sqrt{\tilde{Q}_{k}(x)}$ to $\mathcal{N}_m$ to obtain the
quantized value $Q_{k}(x)$.

\paragraph*{Phase 2} For each $x\in\mathcal{X}$ and $k\in\mathcal{K}$,
the $k$th sensor generates a collection of
$u(\mathcal{N}_m)$-i.i.d. random numbers
$[R_{k,l}(x)]_{l \in \mathcal{K}\setminus k}$, and it
calculates
\begin{equation}\label{equ:Rkk}
  R_{k,k}(x)=\ominus\bigoplus_{l\in\mathcal{K}\setminus k} R_{k,l}(x).
\end{equation}
For each $x\in\mathcal{X}$ and $l\in\mathcal{K}\setminus k$, the $k$th sensor generates the
ciphertext $C_{k,l}(x)=E(R_{k,l}(x);\Phi_{l}^n)$
by encrypting $R_{k,l}(x)$ using the public key $\Phi_{l}^n$,
and sends the ciphertext $C_{k,l}(x)$ to the $l$th
sensor\footnote{The sole purpose of public-key encryption here is to
  make sure that no entity other than the $l$th sensor is able to
  obtain $[\mathbf{R}_{k,l}]_{k \in \mathcal{K}\setminus l}$.}. After
the above round of transmission of ciphertexts, the $k$th sensor
receives the ciphertext collection
$[C_{l,k}(x)]_{l\in\mathcal{K}\setminus k,x\in\mathcal{X}}$ from the other
sensors, and it recovers each
$R_{l,k}(x)=D(C_{l,k}(x);\Psi_{k}^n)$ by
decrypting $C_{l,k}(x)$ using its own private key
$\Psi_{k}^n$. Then, the $k$th sensor computes the secret random number
collection
$\mathbf{\Sigma}_{\mathbf{R}_{\mathcal{K},k}}$.

\paragraph*{Phase 3} For each $x\in\mathcal{X}$ and $k\in\mathcal{K}$, the $k$th sensor
constructs the obfuscated message by
\begin{equation}\label{equ:Gk}
G_{k}(x)={Q}_{k}(x)\oplus
\Sigma_{\mathbf{R}_{\mathcal{K},k}}(x).
\end{equation}
It then sends the collection $\mathbf{G}_k$ to the fusion center.
After receiving the whole set of obfuscated messages
$\mathbf{G}_{\mathcal{K}}= [\mathbf{G}_k]_{k\in\mathcal{K}}$ from all
$K$ sensors, the fusion center calculates 
\begin{equation}\label{d}
\tilde{d}(\mathbf{G}_{\mathcal{K}}) =  K^{2}- \sum_{x\in\mathcal{X}}
\left(\bigoplus_{k=1}^{K}G_{k}(x) \right)^{2},
\end{equation}
which will be used in place of the decision statistic
$d(\mathbf{\tilde Q}_{\mathcal{K}})$ in~\eqref{equ:hypo}.

In the description of the proposed protocol above, we have implicitly
assumed that all sensors, adversarial or not, faithfully follow
the protocol steps. Nevertheless, it is possible for an adversarial
sensor to behave deviantly while still satisfying the requirement in
Section~\ref{sec:threat_model} above not disrupting proper execution
of the test~\eqref{equ:hypo}, so long as~\eqref{equ:Rkk}
and~\eqref{equ:Gk} are both followed. Based on this assumption, we
establish below the ``correctness'' of the proposed protocol by
investigating the detection performance of the hypothesis
test~\eqref{equ:hypo} with $\tilde{d}(\mathbf{G}_{\mathcal{K}})$ as
the decision statistic, while a precise specification of the steps
allowed to be taken by the adversarial sensors under the threat model
described in Section~\ref{sec:threat_model} will be provided in
Section~\ref{sec:analysis}.

Choose $N > K$. From~\eqref{equ:Rkk},
\begin{equation} \label{e:sum0} \bigoplus_{k=1}^{K}
  \Sigma_{\mathbf{R}_{\mathcal{K},k}}(x)=\bigoplus_{k=1}^{K}\bigoplus^{K}_{l=1}R_{l,k}(x)=0,
\end{equation} 
for each $x\in\mathcal{X}$.
Combining this with (\ref{equ:Gk}) and (\ref{d}) gives
\begin{align}
\tilde{d}(\mathbf{G}_{\mathcal{K}})
&=K^{2}-\sum_{x\in\mathcal{X}}\left(\bigoplus_{k=1}^{K}Q_{k}(x)\right)^{2}
\notag \\
& =K^{2}-\sum_{x\in\mathcal{X}} \left(\sum_{k=1}^{K}Q_{k}(x)\right)^{2},
\label{equ:diameter}
\end{align}
where the last equality results since
$\sum_{k=1}^{K}Q_{k}(x)\leq K < N$. From~\eqref{equ:diameter}, it is
easy to see that
\begin{equation}\label{equ:dGerr}
| d(\mathbf{\tilde Q}_{\mathcal{K}}) -
\tilde{d}(\mathbf{G}_{\mathcal{K}})| \leq 2^{-m} K^2 |\mathcal{X}|.
\end{equation}
Thus, using $\tilde{d}(\mathbf{G}_{\mathcal{K}})$ instead as the
decision statistic in~\eqref{equ:hypo} is equivalent to perturbing the
decision threshold $\gamma$. Recall from Theorem~\ref{thm:errexp} that
the decision threshold parameterizes the boundary of region of all
error exponent pairs achievable by the test~\eqref{equ:hypo}. Hence,
as long as $m$ is chosen large enough so that the perturbation bound
above is small (i.e., $m=\mathcal{O}(\log_2 K^2 |\mathcal{X}|)$),
using $\tilde{d}(\mathbf{G}_{\mathcal{K}})$ will cause only a small
shift from the target error exponent pair along that boundary.



\subsection{Simulation Example} \label{sec:simulation}

In this section, we present a simulation example to demonstrate the
detection performance of the privacy-preserving event detection
protocol described in Section~\ref{sec:protocol_description}. The
application scenario considered in this example also helps to motivate
the abstract formulation of the detection problem given in
Section~\ref{sec:formulation}.

\subsubsection{Simulation Scenario}
We consider a simple crowd spectrum sensing application scenario in
which smart phones act as spectrum sensors trying to detect whether a
specific frequency band is occupied in their vicinity. Each phone uses
its radio to make received power measurements at the frequency band of
interest, calculates the quantized square-root type of the power
measurements, and sends messages to a fusion center following the
protocol in Section~\ref{sec:protocol_description}.

In the simulation, as shown in Figure~\ref{fig:setup},
\begin{figure}
\centering
\includegraphics[width=0.4\textwidth]{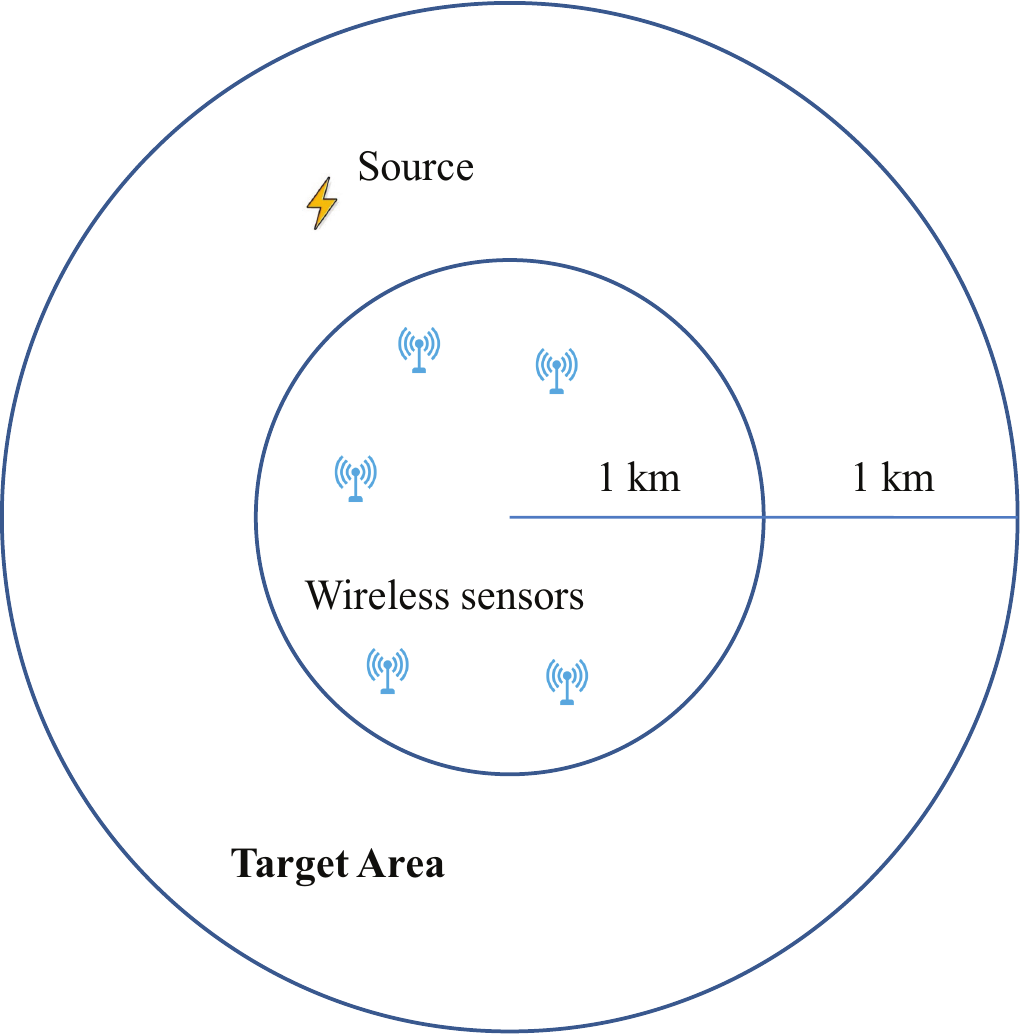}
\caption{The source and sensor regions in the crowd spectrum sensing example.}
\label{fig:setup}
\end{figure}
we consider a circular area with a radius of $2$~km, within which
there is a signal source at an unknown location that may transmit at
the frequency band of interest with an unknown power. If the source
does not transmit (i.e., the frequency band is not occupied),
$\theta=0$; otherwise, $\theta=1$. There are $K$ spectrum sensors,
uniformly distributed in a concentric circle with a radius of $1$~km,
for detecting whether the source transmits or not. The propagation
loss from the source to the sensors is modeled by the CBRS channel
model given in~\cite[pp. 12-13]{<1km} for distances less than $1$~km
and by the Hata model given in~\cite[Eqn. (A-3)]{>1km} for
distances greater than $1$~km. In both cases, the carrier frequency is
fixed at $3625$~MHz, the height of the source antenna is chosen to be
$20$~m, and the antenna height of each sensor is chosen to be $1.5$~m.

We assume that the radio receivers in the spectrum sensors suffer only
from i.i.d. thermal noise, whose effects on the received power level
is modeled by an additive Chi-square distributed component with two
degrees of freedom. The source power and noise power are set to
$25$~dBm and $-103$~dBm, respectively. The measured power in the
decibel scale at each sensor is uniformly quantized to $128$ levels in
the range from $-130$~dBm to $-60$~dBm.  We also set $m=13$.  No
information about the locations of the source and sensors, the channel
model, or the thermal noise described above is made available to the
sensors or the fusion center.  In this case, $|\mathcal{X}|=128$ and
$d_0=0$.

We note that the case of $d_0 = d_1$ is excluded
(see~\eqref{equ:restriction}) in the generalized $K$-sample problem
formulation in Section~\ref{sec:detection}. For the simulation example
described above, the case of $d_0=d_1$ corresponds to the scenario in
which all sensors have the same power measurement distribution when
the source transmits. As the
channel model assumed is isotropic, this can happen only if the
sensors either are co-located, or are at the same distance from the
signal source. With the sensors uniformly distributed in the circular
area shown in Figure~\ref{fig:setup}, it is highly unlikely that we
encounter such a contrived case. As a matter of fact, in the
simulation results shown below, we do not have a single instance of
occurrence of this contrived case in 900 different location
configurations that are randomly generated. In practice, the occurrence
of the case $d_0=d_1$ will be even rarer because of anisotropic
channel conditions, sensor movement, and other channel variations. In
all, the generalized $K$-sample problem formulation with 
$d_1 > d_0$ is a practically robust approach to tackle the crowd
spectrum sensing problem.

\subsubsection{Simulation Results}
We consider two simulation experiments. In the first experiment,
we set the number of sensors $K=8$. We select the length of the
measurement sequences $t$ from $360$ to $600$ at an increment interval
of $20$. We select $30$ groups of random sensor locations and $30$
random source locations uniformly distributed in their respective
areas. This set of random locations form $900$ different
configurations, from which we obtain the worst-case error
probabilities of the first and second types.  For each value of $t$
and each configuration of locations, we conduct the detection
simulation $7.2\times 10^6$ times. For each value of $t$, we find the
largest testing threshold $\gamma$ that makes $\lambda_t$ no more than
$5 \times 10^{-4}$, $5\times 10^{-5}$, and $5\times 10^{-6}$
respectively, and record the corresponding values of
$-\frac{1}{t} \log_2 \mu_t$. These values serve as estimates of the
error exponent of the first type. The results are plotted in
Figure~\ref{fig:exponent}. 
\begin{figure}
\centering
\includegraphics[width=0.49\textwidth]
{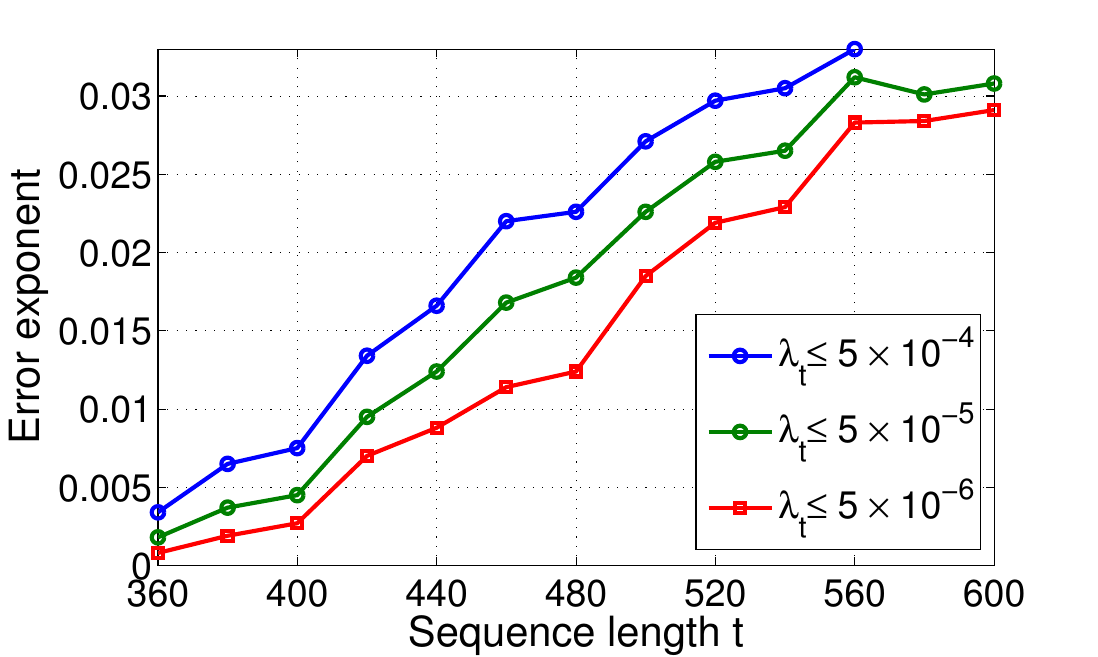}
\caption{Plots of $-\frac{1}{t} \log_2\mu_t$ vs. $t$ for different
  bounds on $\lambda_t$.}
\label{fig:exponent}
\end{figure}
It can be seen that the value of $-\frac{1}{t} \log_2 \mu_t$ increases
as the sequence length $t$ grows, and it levels off as $t$ becomes
large. The results indicate that a positive error exponent of the
first type is achieved, and thus the condition that $d_1 > d_0$ is
valid among the $900$ configurations.

In the second experiment, we fix $t=500$ and consider two different
numbers of sensors, $K=7$ and $8$. For both cases, we select
$900$ random location configurations as in the first experiment to
obtain the worst-cases error probabilities. The receiver operation
characteristic (ROC) curves for the cases of $K=7$ and $K=8$ are
plotted in Figure~\ref{fig:K}.
\begin{figure}
\centering
\includegraphics[width=0.45\textwidth]
{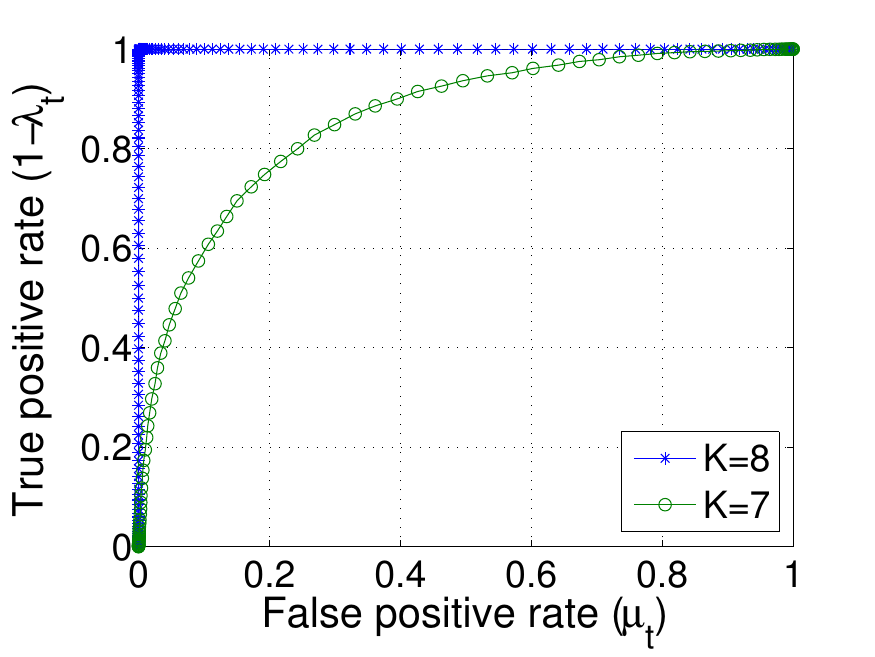}
\caption{The ROC curves for the networks with 7 and 8 sensors.}
\label{fig:K}
\end{figure}
For each case, the ROC curve is obtained from $2\times 10^5$
simulation runs. It can be seen from the figure that at $t=500$
dropping a single sensor from $K=8$ to $7$ significantly degrade the
worst-case detection performance. These results indicate that the
error exponents achieved by the test~\eqref{equ:hypo} with $K=7$
sensors, albeit may still be positive, seem to be smaller than those
achieved by the test~\eqref{equ:hypo} with $K=8$ sensors.



\section{Privacy Analysis}\label{sec:analysis}

In this section, we present a privacy-preserving performance analysis
on the protocol proposed in Section~\ref{sec:protocol} based on the
attacker-challenger formalism described in
Section~\ref{sec:crypto_assump}. The attacker $\mathscr{A}$ is the
combined entity consisting of an external eavesdropper, the fusion
center, and the set of adversarial sensors as discussed in
Section~\ref{sec:threat_model}.  Since no knowledge about how many or
the identity of the set of adversarial sensors is required for the
protocol to operate, we may set
$\mathcal{L} = \{K-L+1, K-L+2, \ldots, K\}$ without any loss of
generality in the analysis below, where $L$ denotes the number of
adversarial sensors. The challenger $\mathscr{C}$, on the other hand,
can be thought of as a fictitious entity that maintains the operation
of the proposed protocol in that it provides all the available inputs
to the attacker in accordance to the protocol. From
Section~\ref{sec:protocol}, these inputs include the public keys
$\mathbf{\Phi}^n_{\mathcal{K} \setminus \mathcal{L}}$, the ciphertexts
$\mathbf{\bar C}_{\mathcal{K}\setminus \mathcal{L},\mathcal{K}}$, and
the obfuscated messages
$\mathbf{G}_{\mathcal{K} \setminus \mathcal{L}}$ sent by the
non-adversarial sensors.

In addition to the above inputs provided by the challenger, the
attacker obviously has access to the quantized square-root types
$\mathbf{Q}_{\mathcal{L}}$ observed by the adversarial sensors as well
as the key pairs
$(\mathbf{\Phi}^n_{\mathcal{L}}, \mathbf{\Psi}^n_{\mathcal{L}})$, the
secret random numbers $\mathbf{R}_{\mathcal{L},\mathcal{K}}$, the
ciphertexts $\mathbf{\bar C}_{\mathcal{L},\mathcal{K}}$, and the
obfuscated messages $\mathbf{G}_{\mathcal{L}}$ generated by
themselves. The goal of the attacker is to produce an estimate of the
square-root types $\mathbf{Q}_{\mathcal{K}\setminus\mathcal{L}}$
observed by the non-adversarial sensors from all available information
described above. The attacker does not need to follow the exact steps in the
protocol proposed in Section~\ref{sec:protocol} as long as any
deviations must not disrupt proper execution of the test
in~\eqref{equ:hypo}. Specifically, we allow the adversarial sensors in
Phase:
\begin{enumerate} 
\item to wait until receiving the public keys
  $\mathbf{\Phi}^n_{\mathcal{K} \setminus \mathcal{L}}$ from the
  non-adversarial sensors before generating their key pairs as any
  general PPT functions of
  $\mathbf{\Phi}^n_{\mathcal{K} \setminus \mathcal{L}}$, i.e.,
  \begin{equation} \label{equ:key_fn}
  (\mathbf{\Phi}^n_{\mathcal{L}}, \mathbf{\Psi}^n_{\mathcal{L}})
  = (\mathbf{\Phi}^n_{\mathcal{L}}(\mathbf{\Phi}^n_{\mathcal{K}
    \setminus \mathcal{L}}),
  \mathbf{\Psi}^n_{\mathcal{L}}(\mathbf{\Phi}^n_{\mathcal{K} \setminus
    \mathcal{L}})),
\end{equation}

\item to wait until receiving the ciphertexts
  $\mathbf{\bar C}_{\mathcal{K}\setminus \mathcal{L},\mathcal{K}}$
  from the non-adversarial sensors and decrypting to obtain
  $\mathbf{R}_{\mathcal{K}\setminus\mathcal{L},\mathcal{L}}$ before
  generating their random numbers as any general PPT functions of the
  information they possess up to that point, i.e.,
 \begin{equation} \label{equ:RLK_fn}
  \mathbf{R}_{\mathcal{L},\mathcal{K}} =
  \mathbf{R}_{\mathcal{L},\mathcal{K}}(\mathbf{Q}_{\mathcal{L}},
  \mathbf{\Phi}^n_{\mathcal{K}},\mathbf{\Psi}^n_{\mathcal{L}},\mathbf{\bar
    C}_{\mathcal{K}\setminus \mathcal{L},\mathcal{K}},
  \mathbf{R}_{\mathcal{K}\setminus\mathcal{L},\mathcal{L}})
 \end{equation} 
 with the restriction that~\eqref{equ:Rkk} must be satisfied for all
 $R_{l,l}$ where $l \in \mathcal{L}$, and

\item to use all information available at the end of the
  protocol in any general PPT estimator for
  $\mathbf{Q}_{\mathcal{K}\setminus\mathcal{L}}$, i.e.,
\begin{align} 
& \hspace{-5pt}
    \mathbf{\hat Q}_{\mathcal{K}\setminus \mathcal{L}} = \notag \\
&\mathbf{\hat Q}_{\mathcal{K}\setminus\mathcal{L}} (\mathbf{Q}_{\mathcal{L}}, 
   \mathbf{\Phi}^n_{\mathcal{K}}, \mathbf{\Psi}^n_{\mathcal{L}},
   \mathbf{\bar C}_{\mathcal{K},\mathcal{K}}, \mathbf{G}_{\mathcal{K}},
   \mathbf{R}_{\mathcal{L},\mathcal{K}}, 
   \mathbf{R}_{\mathcal{K}\setminus  \mathcal{L},\mathcal{L}})
\label{equ:Qhat_fn}
\end{align} 
with the restriction that~\eqref{equ:Gk} must be satisfied for all
$G_l$ where $l \in \mathcal{L}$.  
\end{enumerate}

\subsection{Main Result} \label{sec:tea}

Consider the following type estimation attack (TEA) experiment, as
shown in Figure~\ref{fig:tea}, between the attacker $\mathscr{A}$ and
her challenger $\mathscr{C}$ as specified below:
\begin{figure}
  \centering
  \includegraphics[width=0.5\textwidth]{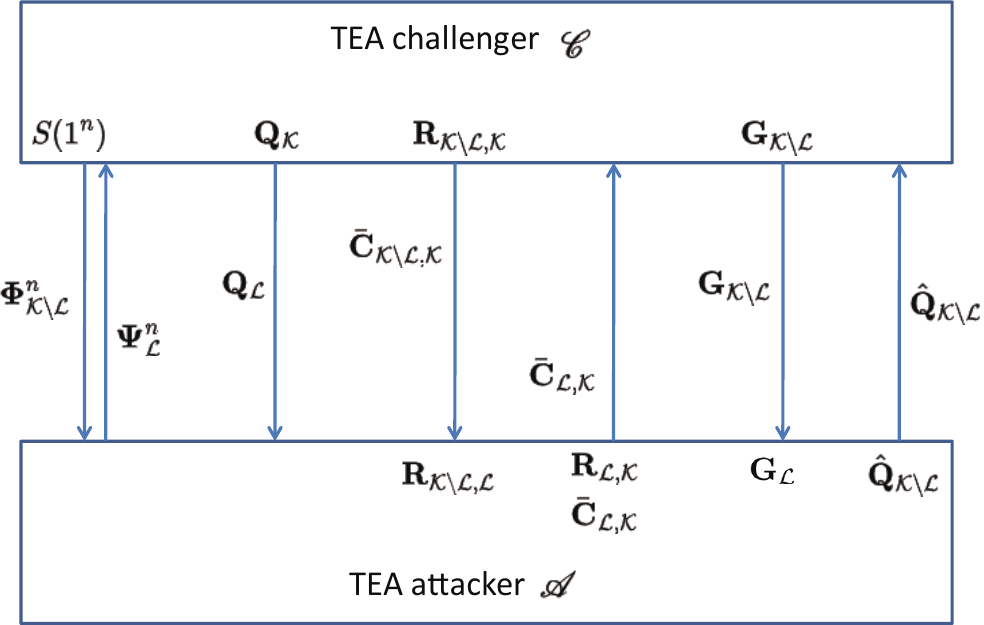}
  \caption{The TEA experiment.}
\label{fig:tea}
\end{figure}
\begin{enumerate}
\item For each $k\in\mathcal{K}\setminus\mathcal{L}$, $\mathscr{C}$
  runs $S(1^{n})$ to get the pair of public key
  $\Phi_{k}^n\in\mathcal{E}_n$ and private key
  $\Psi_{k}^n\in\mathcal{D}_n$, and gives
  $\boldsymbol{\Phi}_{\mathcal{K}\setminus\mathcal{L}}^n$ to
  $\mathscr{A}$. Then, $\mathscr{A}$ generates the set of key pairs
  $(\boldsymbol{\Phi}_{\mathcal{L}}^n,
  \boldsymbol{\Psi}_{\mathcal{L}}^n)$ according to~\eqref{equ:key_fn},
  and gives $\boldsymbol{\Phi}_{\mathcal{L}}^n$ to $\mathscr{C}$.

\item $\mathscr{C}$ draws a collection of quantized square-root types  $\mathbf{Q}_{\mathcal{K}}\in\mathcal{Q}^K_t(\mathcal{X})$ according
  to the distribution $p_{\mathbf{Q}_{\mathcal{K}}}(\cdot;\theta)$,
  and gives $\mathbf{Q}_{\mathcal{L}}$ to $\mathscr{A}$. We assume
  that $\mathscr{A}$ also knows the value of the system state
  $\theta \in \{0,1\}$ and the distribution
  $p_{\mathbf{Q}_{\mathcal{K}}}(\cdot;\theta)$.  Hence, we will simply
  write $p_{\mathbf{Q}_{\mathcal{K}}}(\cdot)$ in place of
  $p_{\mathbf{Q}_{\mathcal{K}}}(\cdot;\theta)$ in the discussion below
  to simplify notation.

\item $\mathscr{C}$ generates $\mathbf{R}_{\mathcal{K}\setminus\mathcal{L},\mathcal{K}}=[\mathbf{R}_{k,l}]_{k\in\mathcal{K}\setminus\mathcal{L},l\in\mathcal{K}}$ with
  $\mathbf{R}_{k,l}$
  i.i.d. $\sim {u}(\mathcal{N}_m^{|\mathcal{X}|})$ for
  $k\neq l$, and  $\mathbf{R}_{k,k}=\ominus\bigoplus_{l\in\mathcal{K}\setminus
    k}\mathbf{R}_{k,l}$ according to~\eqref{equ:Rkk}. Then,
  $\mathscr{C}$ computes
  ${\mathbf{\bar C}}_{\mathcal{K}\setminus\mathcal{L},\mathcal{K}}
  ={E}({\mathbf{\bar R}}_{\mathcal{K}\setminus
    \mathcal{L},\mathcal{K}};\mathbf{\Phi}^n_{\mathcal{K}})$, and
  gives it to $\mathscr{A}$.

\item After receiving ${\mathbf{\bar C}}_{\mathcal{K}\setminus \mathcal{L},\mathcal{K}}$,
$\mathscr{A}$ decrypts to get $\mathbf{R}_{\mathcal{K}\setminus\mathcal{L},\mathcal{L}}
=D(\mathbf{C}_{\mathcal{K}\setminus \mathcal{L},\mathcal{L}};\mathbf{\Psi}^n_{\mathcal{L}})$,
and generates $\mathbf{R}_{\mathcal{L},\mathcal{K}}$ according
to~\eqref{equ:RLK_fn}. Then, $\mathscr{A}$ computes
${\mathbf{\bar C}}_{\mathcal{L},\mathcal{K}} =E({\mathbf{\bar
    R}}_{\mathcal{L},\mathcal{K}};\mathbf{\Phi}^n_{\mathcal{K}})$, and
gives
$\mathbf{\bar C}_{\mathcal{L},\mathcal{K}}$ to $\mathscr{C}$.

\item $\mathscr{C}$ calculates
  $\mathbf{R}_{\mathcal{L},\mathcal{K}\setminus\mathcal{L}} =
  D(\mathbf{C}_{\mathcal{L},\mathcal{K}\setminus\mathcal{L}} ;
  \mathbf{\Psi}^n_{\mathcal{K}\setminus\mathcal{L}})$, computes
  $\mathbf{G}_{\mathcal{K}\backslash \mathcal{L}}
  =\mathbf{Q}_{\mathcal{K}\setminus
    \mathcal{L}}\oplus\mathbf{\Sigma}_{\mathbf{R}_{\mathcal{K},\mathcal{K}\setminus\mathcal{L}}}$
  according to~\eqref{equ:Gk}, and gives
  $\mathbf{G}_{\mathcal{K}\setminus \mathcal{L}}$ to $\mathscr{A}$.

\item $\mathscr{A}$ computes
  $\mathbf{G}_{\mathcal{L}}=\mathbf{Q}_{\mathcal{L}}
  \oplus\mathbf{\Sigma}_{\mathbf{R}_{\mathcal{K},\mathcal{L}}}$
  according to~\eqref{equ:Gk}. Then, according to~\eqref{equ:Qhat_fn},
  $\mathscr{A}$ generates, and reports to $\mathscr{C}$,
  ${\mathbf{\hat Q}}_{\mathcal{K}\setminus \mathcal{L}}$ as her
  estimate of $\mathbf{Q}_{\mathcal{K}\setminus \mathcal{L}}$.
\end{enumerate}
Note that in step 2) above the distribution
$p_{\mathbf{Q}_{\mathcal{K}}}(\cdot;\theta)$ models two different
physical mechanisms that give rise to the randomness of
$\mathbf{Q}_{\mathcal{K}}$. The first mechanism is the choice of
$\mathbf{p}_{\theta, \mathcal{K}}$, which is a random instantiation
from some underlying random model that characterizes attributes, such
as the locations as in the example of Section~\ref{sec:simulation}, of
the sensors. A more conservative deterministic approach is adopted in
the formulation of the event detection problem in
Section~\ref{sec:detection} by treating
$\mathbf{p}_{\theta, \mathcal{K}}$ as deterministic and considering the
worst-case detection errors. It is more convenient to consider a
random model for privacy analysis here. The second mechanism is the
random instantiation of $\mathbf{X}_{\mathcal{K}}$, of which
$\mathbf{Q}_{\mathcal{K}}$ is a function, from
$\mathbf{p}_{\theta, \mathcal{K}}$. This mechanism is modeled in
exactly the same way in the detection problem.

For any radius $\tau\geq 0$ and
$\mathbf{q}_{\mathcal{K} \setminus \mathcal{L}} \in
\mathcal{Q}^{K-L}_t(\mathcal{X})$, define a neighborhood of quantized
square root types around
$\mathbf{q}_{\mathcal{K} \setminus \mathcal{L}}$:
\begin{align*}
 \mathcal{N}_{\tau}(\mathbf{q}_{\mathcal{K} \setminus \mathcal{L}})
 &= \big\{ \mathbf{q}'_{\mathcal{K} \setminus \mathcal{L}} \in
\mathcal{Q}^{K-L}_t(\mathcal{X}):
d_H(\mathbf{q}'^2_{\mathcal{K}\setminus \mathcal{L}},
\mathbf{q}^2_{\mathcal{K}\setminus \mathcal{L}}) \leq \tau, 
\\
& \hspace{20pt}
\boldsymbol{\Sigma}_{\mathbf{q}'_{\mathcal{K}\setminus \mathcal{L}}}
=\boldsymbol{\Sigma}_{\mathbf{q}_{\mathcal{K}\setminus \mathcal{L}}} \big\},
\end{align*}
where $\mathbf{q}^2_{\mathcal{K}\setminus \mathcal{L}}$ denotes
elementwise squaring of the vector
$\mathbf{q}_{\mathcal{K}\setminus \mathcal{L}}$. Then, we say
$\mathscr{A}$ wins the TEA experiment if
${\mathbf{\hat Q}}_{\mathcal{K}\setminus \mathcal{L}}$ is within a
small neighborhood around
$\mathbf{Q}_{\mathcal{K}\setminus \mathcal{L}}$, i.e.,
${\mathbf{\hat Q}}_{\mathcal{K}\setminus \mathcal{L}} \in
\mathcal{N}_{\tau}(\mathbf{Q}_{\mathcal{K}\setminus \mathcal{L}})$.

\begin{theorem}\label{thm:tea}
  Let $\mathbf{\hat Q}_{\mathcal{K}\setminus\mathcal{L}}$ be the
  estimator for ${\mathbf{Q}}_{\mathcal{K}\setminus\mathcal{L}}$ of
  any PPT attacker $\mathscr{A}$ in the TEA experiment above.  Given
  $[\boldsymbol{\Sigma}_{\mathbf{Q}_{\mathcal{K}\setminus\mathcal{L}}},
  \mathbf{Q}_{\mathcal{L}}]$, let
  ${\mathbf{\hat Q}}^{\prime}_{\mathcal{K}\setminus\mathcal{L}}$ be
  another estimator that has the same conditional distribution as
  ${\mathbf{\hat Q}}_{\mathcal{K}\setminus\mathcal{L}}$ but is
  conditionally \textbf{independent} of
  ${\mathbf{Q}}_{\mathcal{K}\setminus\mathcal{L}}$. If $L\leq K-2$,
  then for any $\tau \geq 0$,
  $\boldsymbol{\sigma}\in\mathcal{N}_m^{|\mathcal{X}|}$,
  $\mathbf{q}_{\mathcal{L}}\in\mathcal{Q}_t^L(\mathcal{X})$,
\begin{align}
& \hspace{-20pt}
\Pr(\mathbf{\hat{Q}}_{\mathcal{K}\setminus\mathcal{L}}\in
\mathcal{N}_{\tau}(\mathbf{Q}_{\mathcal{K}
\backslash\mathcal{L}})
\mid \boldsymbol{\Sigma}_{\mathbf{Q}_{\mathcal{K}\backslash\mathcal{L}}}=\boldsymbol{\sigma},
\mathbf{Q}_{\mathcal{L}}=\mathbf{q}_{\mathcal{L}}) \notag \\
\leq &
\Pr(\mathbf{\hat{Q}}^{\prime}_{\mathcal{K}
\setminus\mathcal{L}}\in
\mathcal{N}_{\tau}(\mathbf{Q}_{\mathcal{K}\backslash\mathcal{L}})
\mid \boldsymbol{\Sigma}_{\mathbf{Q}_{\mathcal{K}\backslash\mathcal{L}}}=\boldsymbol{\sigma},
\mathbf{Q}_{\mathcal{L}}=\mathbf{q}_{\mathcal{L}}) \notag \\
& + 8(K-L-1)|\mathcal{X}|\cdot F_{\mathrm{CPA}}(n)
\label{equ:edra}
\end{align}
where $F_{\mathrm{CPA}}(n)$, given in~\eqref{equ:adv}, is the
probability advantage of an attacker in the CPA experiment against the
public-key cryptographic scheme $\Pi$ with security parameter $n$.
\end{theorem}
The theorem guarantees that as long as the public-key cryptographic
scheme $\Pi$ employed is CPA-secure, any PPT attacker cannot do much
better than independently guessing the value of
${\mathbf{Q}}_{\mathcal{K}\setminus\mathcal{L}}$ given her own
information $\mathbf{Q}_{\mathcal{L}}$ from the adversarial sensors
and the information
$\boldsymbol{\Sigma}_{\mathbf{Q}_{\mathcal{K}\backslash\mathcal{L}}}$
``leaked'' to her via the proposed protocol. 
One may further quantify
the notion of ``not much better'' above, by noting that since
$\mathcal{Q}^{K-L}_t(\mathcal{X})$ has at most
$(t+1)^{(K-L)|\mathcal{X}|}$ elements, we must have
\begin{align*}
& \hspace{-5pt}
\max_{\mathbf{\hat{Q}}^{\prime}_{\mathcal{K} \setminus\mathcal{L}}}
\Pr(\mathbf{\hat{Q}}^{\prime}_{\mathcal{K} \setminus\mathcal{L}} \in
\mathcal{N}_{\tau}(\mathbf{Q}_{\mathcal{K}\backslash\mathcal{L}})
\mid \boldsymbol{\Sigma}_{\mathbf{Q}_{\mathcal{K}\backslash\mathcal{L}}}=\boldsymbol{\sigma},
\mathbf{Q}_{\mathcal{L}}=\mathbf{q}_{\mathcal{L}}) \\
&\geq (t+1)^{-(K-L)|\mathcal{X}|}.
\end{align*}
If $\Pi$ is CPA-secure, then it suffices to choose
$n = \mathcal{O}(t^{\rho (K-L)|\mathcal{X}|})$, for any
$\rho > 0$, to make the bound in~\eqref{equ:edra} non-trivial.  We
emphasize that the direct bound on the successful estimation
probability achieved by the attacker given by Theorem~\ref{thm:tea}
provides a much stronger privacy guarantee than what the \emph{view}
approach can.

The main idea of the proof of Theorem~\ref{thm:tea} is to first reduce
the TEA experiment to a type discrimination attack (TDA) experiment in
which the attacker aims to distinguish between a pair of quantized
square-root types instead. The TDA experiment is then further
decomposed into two CPA experiments and a third one involving only the
secret random numbers. The bound on the correct estimation probability
achieved by the attacker in~\eqref{equ:edra} is obtained from the
advantages of the experiments in the chain of reduction steps
mentioned.  Based on this roadmap, we will construct the proof of
Theorem~\ref{thm:tea} step by step in the rest of this section.

\subsection{Useful Lemmas}

Before proceeding to construct the proof of Theorem~\ref{thm:tea}, we
state here a few lemmas that help to simplify later discussions. As
the proofs of these lemmas are either trivial or technical rather than
illustrative, they are provided for completeness in
Appendix~\ref{sec:prooflemmas}.

\begin{lem}\label{lem:useful2}
  Suppose $L\leq K-2$. Let
  $\mathcal{I}=\{1,2\} \subseteq \mathcal{K} \setminus \mathcal{L}$,
  and
  $\mathbf{R}_{\mathcal{I},\mathcal{K}}=[\mathbf{R}_{k,l}]_{k\in\mathcal{I},l\in\mathcal{K}}$
  be a collection of random variables satisfying $\mathbf{R}_{k,l}$
  i.i.d. $\sim {u}(\mathcal{N}_m^{|\mathcal{X}|})$ for $k\neq l$, and
  $\mathbf{R}_{k,k}=\ominus\bigoplus_{j\in\mathcal{K}\setminus
    k}\mathbf{R}_{k,j}$ according to~\eqref{equ:Rkk}. Then, for any
  $\mathbf{r}_{\mathcal{I},\mathcal{L}}\in\mathcal{N}_m^{2L|\mathcal{X}|}$
  and
  $\boldsymbol{\sigma}_{\mathcal{K}\setminus\mathcal{L}} \in
  \mathcal{N}_m^{(K-L)|\mathcal{X}|}$,
\begin{align}
& \hspace{-10pt}
P_{\boldsymbol{\Sigma}_{\mathbf{R}_{\mathcal{I},\mathcal{K}\setminus\mathcal{L}}}\mid
\mathbf{R}_{\mathcal{I},\mathcal{L}}}
(\boldsymbol{\sigma}_{\mathcal{K}\setminus\mathcal{L}}\mid\mathbf{r}_{\mathcal{I},\mathcal{L}})
  \notag \\
&=
2^{-m(K-L-1)|\mathcal{X}|}\cdot
\delta\left(\boldsymbol{\Sigma}_{\boldsymbol{\sigma}_{\mathcal{K}\setminus\mathcal{L}}}
\oplus\bigoplus_{l\in\mathcal{L}}\boldsymbol{\Sigma}_{\mathbf{r}_{\mathcal{I},l}}\right).
\label{equ:lem2}
\end{align}
\end{lem}

\begin{lem}\label{lem:useful3}
  Let $B\in\{0,1\}$, $Y\in\mathcal{Y}$, $U\in\mathcal{U}$,
  $V\in\mathcal{V}$, and $W\in\mathcal{W}$ be discrete random
  variables. Let $\hat{B}(Y,V,W)$ be a PPT estimator of $B$.
  If $W$ is conditionally independent of $B$ given
  $[Y,U,V]$ and $W=W(Y,U,V)$ can be generated by a PPT algorithm, then the
  estimator $\hat{B}_0(Y,U,V)=\hat{B}(Y,V,W(Y,U,V))$ is PPT, and for any
  $(u,v)\in\mathcal{U}\times\mathcal{V}$, 
\begin{align}
& \hspace{-10pt}
\Pr(\hat{B}_0 (Y,U,V)=B\mid U=u,V=v) \notag \\
& =\Pr(\hat{B}(Y,V,W)=B\mid U=u,V=v).
\label{equ:useful3}
\end{align}
\end{lem}

\begin{lem}\label{lem:useful4}
  Let $Z_1\in\mathcal{Z}_1$, $Z_2\in\mathcal{Z}_2$,
  $Z_3\in\mathcal{Z}_3$, and $W=W(Z_1,Z_2)\in\mathcal{W}$ be discrete
  random variables. If $Z_1$ is conditionally independent of $Z_3$
  given $Z_2$, then $W$ is conditionally independent of $Z_3$ given
  $Z_2$.
\end{lem}

\subsection{Multi-Encryption CPA Experiment}

Recall that the privacy-preserving protocol in
Section~\ref{sec:protocol} requires each of the $K$ sensors to send
multiple ciphertexts to other sensors. Thus, to prove
Theorem~\ref{thm:tea}, we need to extend the CPA experiment described
in Section~\ref{sec:crypto_assump} to the multi-sensor, multi-message
setting of $K$ sensors, each encrypting $M_{k}$ messages (plaintexts),
and $M = \sum_{k=1}^{K}M_{k}$:
\begin{enumerate}
\item The challenger runs $S(1^{n})$ to generate the pair of public
  key $\Phi_{k}^n\in\mathcal{E}_n$ and private key
  $\Psi_{k}^n\in\mathcal{D}_n$, for each $k\in\mathcal{K}$. The
  challenger gives the set of public keys
  $\boldsymbol{\Phi}_{\mathcal{K}}^n$ to the attacker.

\item The attacker generates two collections of challenge messages
  $\mathbf{R}^0_{\mathcal{K}}=[[R_{k,i}^{0}]_{i=1}^{M_{k}}]_{k=1}^{K}$
  and
  $\mathbf{R}^1_{\mathcal{K}}=[[R_{k,i}^{1}]_{i=1}^{M_{k}}]_{k=1}^{K}$,
  where $R_{k,i}^{0}$ and $R_{k,i}^{1}$ are i.i.d.
  $\sim u(\mathcal{N}_m)$ for all $k\in\mathcal{K}$ and
  $i=1,2,\ldots,M_{k}$. The attacker gives
  $\mathbf{R}^0_{\mathcal{K}}$ and $\mathbf{R}^1_{\mathcal{K}}$ to the
  challenger.

\item The challenger generates an independent random bit $B=\{0,1\}$
  with equal probabilities, computes the ciphertext collection
  $\mathbf{C}^B_{\mathcal{K}}=[[{E}(R_{k,i}^{B};\Phi_k^n)]_{i=1}^{M_{k}}]_{k=1}^{K}\in\mathcal{C}^M$,
  and gives it to the attacker.

\item The attacker uses the estimator
  $\hat{B}=\hat{B}(\mathbf{C}^B_{\mathcal{K}},
  \mathbf{R}^0_{\mathcal{K}},
  \mathbf{R}^1_{\mathcal{K}},\mathbf{\Phi}^n_{\mathcal{K}})$ to output
  her estimate of $B$, and reports $\hat{B}$ to the challenger.
\end{enumerate}
If $\hat{B}=B$, then the attacker wins the multi-encryption CPA
experiment. The following lemma expresses the winning probability
advantage of the multi-encryption CPA attacker in terms of that of a
CPA attacker:
\begin{lem}\label{lem:mcpa}
  For any PPT attacker in the multi-encryption CPA experiment
  described above, 
\begin{equation}
    \Pr(\hat{B}=B)-\frac{1}{2}\leq M\cdot F_{\mathrm{CPA}}(n).
\end{equation}
\end{lem}
\begin{proof}
  Based on Assumption~\ref{a:crypto}, the reduction approach
  in~\cite{kmcpa} can be directly used here to establish the lemma.
\end{proof}

\subsection{Type Discrimination Attack (TDA)} \label{sec:tda}

The proof of Theorem~\ref{thm:tea} relies on a simpler version of the
TEA experiment in which the attacker tries to distinguish between a
pair of quantized square-root types instead. We refer to this simpler
experiment as the type discrimination attack (TDA) experiment. The steps
of the TDA experiment between the attacker $\mathscr{A}'$ and her challenger
$\mathscr{C}'$, as shown in Figure~\ref{fig:tda}, are as follows:
\begin{figure}
  \centering
  \includegraphics[width=0.46\textwidth]{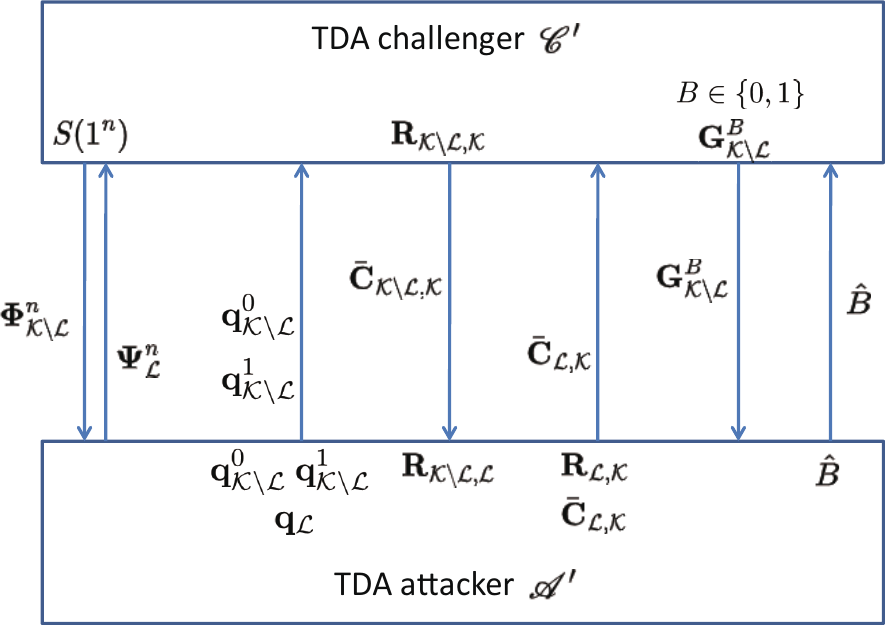}
  \caption{The TDA experiment.}
\label{fig:tda}
\end{figure}
\begin{enumerate}
\item Same as step 1) of the TEA experiment with $\mathscr{A}'$ and
  $\mathscr{C}'$ taking the roles of $\mathscr{A}$ and $\mathscr{C}$,
  respectively.

\item $\mathscr{A}'$ selects three collections of quantized square-root
  types: $\mathbf{q}_{\mathcal{L}}\in\mathcal{Q}_t^{L} (\mathcal{X})$,
  $\mathbf{q}^{0}_{\mathcal{K}\setminus\mathcal{L}}
  \in\mathcal{Q}^{K-L}_t (\mathcal{X})$, and
  $\mathbf{q}^{1}_{\mathcal{K}\setminus\mathcal{L}}
  \in\mathcal{Q}^{K-L}_t (\mathcal{X})$ satisfying
  $\mathbf{\Sigma}_{\mathbf{q}^{0}_{\mathcal{K}\setminus\mathcal{L}}}
  =\mathbf{\Sigma}_{\mathbf{q}^{1}_{\mathcal{K}\setminus\mathcal{L}}}$.
  $\mathscr{A}'$~gives
  $[\mathbf{q}^{0}_{\mathcal{K}\setminus\mathcal{L}},
  \mathbf{q}^{1}_{\mathcal{K}\setminus\mathcal{L}}]$ to $\mathscr{C}'$.

\item Same as step 3) of the TEA experiment with $\mathscr{A}'$ and
  $\mathscr{C}'$ taking the roles of $\mathscr{A}$ and $\mathscr{C}$,
  respectively. 

\item Same as step 4) of the TEA experiment with $\mathscr{A}'$ and
  $\mathscr{C}'$ taking the roles of $\mathscr{A}$ and $\mathscr{C}$,
  respectively.

\item $\mathscr{C}'$ calculates
  $\mathbf{R}_{\mathcal{L},\mathcal{K}\setminus\mathcal{L}} =
  D(\mathbf{C}_{\mathcal{L},\mathcal{K}\setminus\mathcal{L}} ;
  \mathbf{\Psi}^n_{\mathcal{K}\setminus\mathcal{L}})$, generates an
  independent random bit $B\in\{0,1\}$ with equal probabilities,
  computes
  $\mathbf{G}^B_{\mathcal{K}\setminus
    \mathcal{L}}=\mathbf{q}^{B}_{\mathcal{K}\backslash
    \mathcal{L}}\oplus\mathbf{\Sigma}_{\mathbf{R}_{\mathcal{K},
      \mathcal{K}\setminus\mathcal{L}}}$, and gives
  $\mathbf{G}^B_{\mathcal{K}\setminus \mathcal{L}}$ to $\mathscr{A}'$.

\item $\mathscr{A}'$ estimates $B$ using the estimator
  $\hat{B} = \hat{B} ( \mathbf{q}^0_{\mathcal{K}\setminus\mathcal{L}},
  \mathbf{q}^1_{\mathcal{K}\setminus \mathcal{L}},
  \mathbf{q}_{\mathcal{L}}, 
  {\mathbf{\bar C}}_{\mathcal{K},\mathcal{K}},
  \mathbf{G}^B_{\mathcal{K}\setminus \mathcal{L}},
  \mathbf{R}_{\mathcal{L},\mathcal{K}},
  \mathbf{R}_{\mathcal{K}\setminus \mathcal{L},\mathcal{L}},
  \mathbf{\Phi}^n_{\mathcal{K}},\mathbf{\Psi}^n_{\mathcal{L}} )$, and
  reports $\hat{B}$ to $\mathscr{C}'$.
\end{enumerate}
If $\hat{B}=B$, it is said that $\mathscr{A}'$ wins the TDA experiment. The following
lemma expresses the winning probability advantage of the TDA attacker
in terms of that of a CPA attacker:
\begin{lem}\label{lem:tda}
  Suppose $L\leq K-2$.  For any PPT attacker in the TDA experiment
  described above, 
  $\mathbf{q}_{\mathcal{L}} \in \mathcal{Q}^{L}_t (\mathcal{X})$, and 
  $\mathbf{q}^{0}_{\mathcal{K}\setminus\mathcal{L}},
  \mathbf{q}^{1}_{\mathcal{K}\setminus\mathcal{L}} \in \mathcal{Q}^{K-L}_t (\mathcal{X})$
  satisfying
  $\mathbf{\Sigma}_{\mathbf{q}^0_{\mathcal{K}\setminus\mathcal{L}}}
  =\mathbf{\Sigma}_{\mathbf{q}^1_{\mathcal{K}\setminus\mathcal{L}}}$,
\begin{align}
& \hspace{-5pt}
\Pr(\hat{B}=B \mid \mathbf{Q}^{0}_{\mathcal{K}\setminus\mathcal{L}}
=\mathbf{q}^{0}_{\mathcal{K}\setminus\mathcal{L}},
\mathbf{Q}^{1}_{\mathcal{K}\setminus\mathcal{L}}=\mathbf{q}^{1}_{\mathcal{K}\setminus\mathcal{L}},
\mathbf{Q}_{\mathcal{L}}=\mathbf{q}_{\mathcal{L}})
\notag \\
&\leq \frac{1}{2}  + 4(K-L-1)|\mathcal{X}|\cdot F_{\mathrm{CPA}}(n).
\label{equ:tda}
\end{align}
\end{lem}


\begin{proof}
  The main idea of the proof is to use the TDA attacker $\mathscr{A}'$
  to construct three attackers in three new experiments. The first two
  experiments can be reduced to multi-encryption CPA experiments by way
  of Lemma~\ref{lem:useful3}. Thus, Lemma~\ref{lem:mcpa} gives the
  probability advantages of the attackers in these two experiments. On
  the other hand, the probability advantage of the attacker winning
  the third experiment can be analyzed using
  Lemma~\ref{lem:useful2}. Then, the probability advantage of
  $\mathscr{A}'$ winning her TDA experiment can be derived from the
  probability advantages of the new attackers winning their respective
  experiments.  For convenience, we write $\mathcal{I}=\{1,2\}$ and
  $\mathcal{J}=\mathcal{K}\setminus(\mathcal{I} \cup \mathcal{L})$
throughout the rest of the proof.

  As shown in Figure~\ref{fig:tdaa1}, we construct the first
  experiment with attacker $\mathscr{A}_{1}$ and challenger
  $\mathscr{C}_{1}$ as follows:
\begin{figure}
  \centering
  \includegraphics[width=0.47\textwidth]{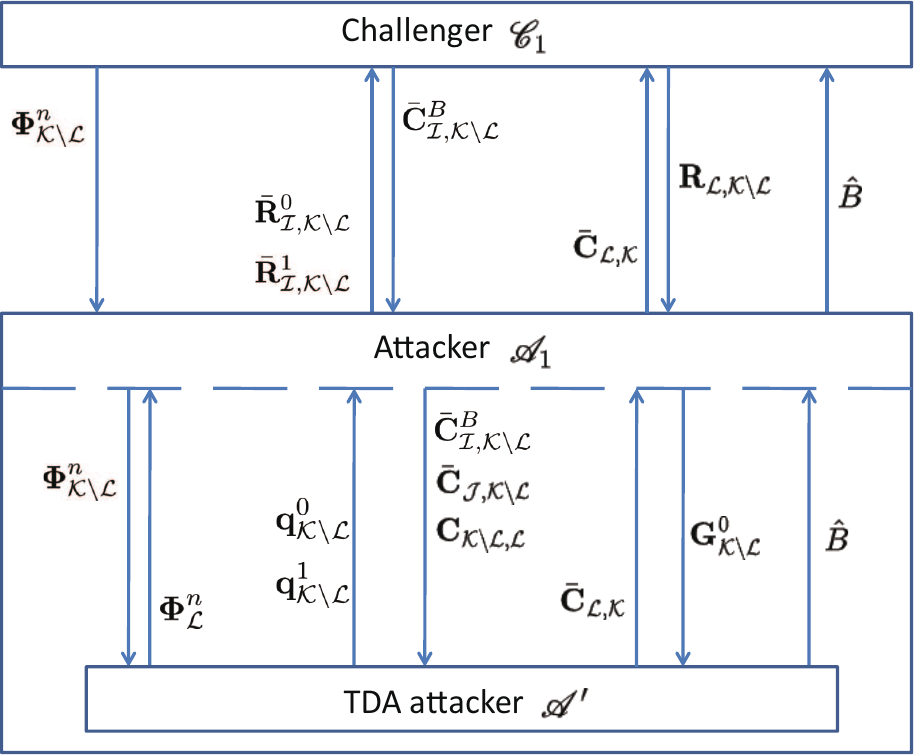}\\
\caption{The first constructed experiment for proving Lemma \ref{lem:tda}.}\label{fig:tdaa1}
\end{figure}

\begin{enumerate}

\item $\mathscr{C}_{1}$ runs $S(1^{n})$ to get the key pair collection
  $(\boldsymbol{\Phi}^n_{\mathcal{K} \setminus \mathcal{L}},
  \boldsymbol{\Psi}^n_{\mathcal{K} \setminus \mathcal{L}})$, and gives
  $\boldsymbol{\Phi}^n_{\mathcal{K} \setminus \mathcal{L}}$ to
  $\mathscr{A}_1$, who passes it on to $\mathscr{A}'$. Then,
  $\mathscr{A}'$ generates the set of key pairs
  $(\boldsymbol{\Phi}_{\mathcal{L}}^n,
  \boldsymbol{\Psi}_{\mathcal{L}}^n)$ according to~\eqref{equ:key_fn},
  and gives $\boldsymbol{\Phi}_{\mathcal{L}}^n$ to $\mathscr{A}_1$. 

\item $\mathscr{A}'$ selects $\mathbf{q}_{\mathcal{L}}$,
  $\mathbf{q}^{0}_{\mathcal{K}\setminus \mathcal{L}}$, and
  $\mathbf{q}^{1}_{\mathcal{K}\setminus \mathcal{L}}$ as in step 2) of
  the TDA experiment, and then passes
  $[\mathbf{q}^{0}_{\mathcal{K}\setminus \mathcal{L}},
  \mathbf{q}^{1}_{\mathcal{K}\setminus \mathcal{L}}]$ to
  $\mathscr{A}_1$.

\item $\mathscr{A}_{1}$ generates
  $\mathbf{R}_{\mathcal{K}\setminus\mathcal{L},\mathcal{K}}$ with
  $\mathbf{R}_{k,l}$
  i.i.d. $\sim u(\mathcal{N}_m^{|\mathcal{X}|})$ for
  $k\neq l$, and
  $\mathbf{R}_{k,k}=\ominus\bigoplus_{l\in\mathcal{K}\setminus
    k}\mathbf{R}_{k,l}$ according to~\eqref{equ:Rkk}. Then,
  $\mathscr{A}_{1}$ sets
  ${\mathbf{\bar R}}^0_{\mathcal{I},\mathcal{K}\setminus\mathcal{L}}
  ={\mathbf{\bar R}}_{\mathcal{I},\mathcal{K}\setminus\mathcal{L}}$
  and
  ${\mathbf{\bar R}}^1_{\mathcal{I},\mathcal{K}\setminus\mathcal{L}} =
  0^{2(K-L-1)|\mathcal{X}|}$, and gives these two collections to
  $\mathscr{C}_1$.

\item $\mathscr{C}_{1}$ selects an independent bit $B\in\{0,1\}$ with
  equal probabilities, computes
  ${\mathbf{\bar C}}^B_{\mathcal{I},\mathcal{K}\setminus\mathcal{L}}
  =E({\mathbf{\bar R}}^B_{\mathcal{I},\mathcal{K}\setminus\mathcal{L}}
  ; \boldsymbol{\Phi}^n_{\mathcal{K}\setminus\mathcal{L}})$, and gives
  ${\mathbf{\bar C}}^B_{\mathcal{I},\mathcal{K}\setminus\mathcal{L}}$
  to $\mathscr{A}_{1}$.

\item $\mathscr{A}_{1}$ computes
  ${\mathbf{\bar C}}_{\mathcal{J},\mathcal{K}\setminus\mathcal{L}}
  =E({\mathbf{\bar R}}_{\mathcal{J},\mathcal{K}\setminus\mathcal{L}};
  \boldsymbol{\Phi}^n_{\mathcal{K}\setminus\mathcal{L}})$,
  $\mathbf{C}_{\mathcal{K}\setminus\mathcal{L},\mathcal{L}}
  =E(\mathbf{R}_{\mathcal{K}\setminus\mathcal{L},\mathcal{L}} ;
  \boldsymbol{\Phi}^n_{\mathcal{L}})$, and gives
  ${\mathbf{\bar C}}^B_{\mathcal{I},\mathcal{K}\setminus\mathcal{L}}$,
  ${\mathbf{\bar C}}_{\mathcal{J},\mathcal{K}\setminus\mathcal{L}}$,
  and $\mathbf{C}_{\mathcal{K}\setminus\mathcal{L},\mathcal{L}}$ to
  $\mathscr{A}'$.

\item $\mathscr{A}'$ follows step 4) of the TDA experiment to decrypt
  $\mathbf{C}_{\mathcal{K}\setminus\mathcal{L},\mathcal{L}}$, generate
  $\mathbf{R}_{\mathcal{L}, \mathcal{K}}$, encrypt to obtain
  $\mathbf{\bar C}_{\mathcal{L}, \mathcal{K}}$, and send
  $\mathbf{\bar C}_{\mathcal{L}, \mathcal{K}}$ to $\mathscr{A}_1$, who
  then passes on
  $\mathbf{C}_{\mathcal{L}, \mathcal{K} \setminus\mathcal{L}}$ to
  $\mathscr{C}_{1}$.

\item $\mathscr{C}_{1}$ calculates
  $\mathbf{R}_{\mathcal{L}, \mathcal{K} \setminus\mathcal{L}} =
  D(\mathbf{C}_{\mathcal{L}, \mathcal{K} \setminus\mathcal{L}} ;
  \mathbf{\Psi}^n_{\mathcal{K} \setminus\mathcal{L}})$ and sends
  $\mathbf{R}_{\mathcal{L}, \mathcal{K} \setminus\mathcal{L}}$ to
  $\mathscr{A}_1$.

\item $\mathscr{A}_1$ computes
  $\mathbf{G}^0_{\mathcal{K}\setminus\mathcal{L}} =
  \mathbf{q}^{0}_{\mathcal{K}\setminus\mathcal{L}} \oplus
  \mathbf{\Sigma}_{\mathbf{R}_{\mathcal{K},\mathcal{K}\setminus
      \mathcal{L}}}$, and gives
  $\mathbf{G}^0_{\mathcal{K}\setminus\mathcal{L}}$ to $\mathscr{A}'$.

\item $\mathscr{A}'$ uses
  $\hat{B}=\hat{B}(\mathbf{q}^0_{\mathcal{K}\setminus\mathcal{L}},
  \mathbf{q}^1_{\mathcal{K}\setminus\mathcal{L}},
  \mathbf{q}_{\mathcal{L}}, [ \mathbf{\bar
    C}^B_{\mathcal{I},\mathcal{K}\setminus\mathcal{L}},
  \mathbf{\bar C}_{\mathcal{J},\mathcal{K}\setminus\mathcal{L}},\\ \mathbf{C}_{\mathcal{K}\setminus\mathcal{L},\mathcal{L}},
  \mathbf{\bar C}_{\mathcal{L},\mathcal{K}}],  \mathbf{G}^0_{\mathcal{K}\setminus\mathcal{L}},
  \mathbf{R}_{\mathcal{L},\mathcal{K}},\mathbf{R}_{\mathcal{K}\setminus
    \mathcal{L},\mathcal{L}},
  \boldsymbol{\Phi}^n_{\mathcal{K}},\boldsymbol{\Psi}^n_{\mathcal{L}})$
  in step 6) of the TDA experiment with the input arguments as
  specified to estimate $B$, and reports $\hat{B}$ to
  $\mathscr{A}_{1}$, who passes it on to $\mathscr{C}_{1}$.
\end{enumerate}

We will use Lemmas~\ref{lem:useful3} and~\ref{lem:useful4} below. To
match the notation in the lemmas, let
$\mathbf{Q}^0_{\mathcal{K}\setminus\mathcal{L}} \sim
\delta_{\mathbf{q}^0_{\mathcal{K}\setminus\mathcal{L}}}$,
$\mathbf{Q}^1_{\mathcal{K}\setminus\mathcal{L}} \sim
\delta_{\mathbf{q}^1_{\mathcal{K}\setminus\mathcal{L}}}$,
$\mathbf{Q}_{\mathcal{L}} \sim \delta_{\mathbf{q}_{\mathcal{L}}}$,
$Y = {\mathbf{\bar
    C}}^B_{\mathcal{I},\mathcal{K}\setminus\mathcal{L}}$,
$U = {\mathbf{\bar R}}_{\mathcal{I},\mathcal{K}\setminus\mathcal{L}}$,
$V=\boldsymbol{\Phi}^n_{\mathcal{K}\setminus\mathcal{L}}$,
$Z_0= [\mathbf{Q}^0_{\mathcal{K}\setminus\mathcal{L}},
\mathbf{Q}^1_{\mathcal{K}\setminus\mathcal{L}},
\mathbf{Q}_{\mathcal{L}}, {\mathbf{\bar
    C}}_{\mathcal{J},\mathcal{K}\setminus\mathcal{L}},
\mathbf{C}_{\mathcal{K}\setminus\mathcal{L},\mathcal{L}},
\mathbf{R}_{\mathcal{K}\setminus\mathcal{L},\mathcal{L}},
\boldsymbol{\Phi}^n_{\mathcal{L}},
\boldsymbol{\Psi}^n_{\mathcal{L}}]$,
$Z_1 = [Z_0, {\mathbf{\bar
    R}}_{\mathcal{J},\mathcal{K}\setminus\mathcal{L}}]$,
$Z_2 = [Y,U,V]$, and
$W = [{\mathbf{\bar C}}_{\mathcal{L},\mathcal{K}},
\mathbf{G}^0_{\mathcal{K}\setminus\mathcal{L}}, 
\mathbf{R}_{\mathcal{L},\mathcal{K}}, Z_0]$.

Note that
${\mathbf{\bar C}}_{\mathcal{L},\mathcal{K}} = E( {\mathbf{\bar
    R}}_{\mathcal{L},\mathcal{K}};
\boldsymbol{\Phi}^n_{\mathcal{K}})$,
$\mathbf{G}^0_{\mathcal{K}\setminus\mathcal{L}}$ is a function of
$[\mathbf{Q}^0_{\mathcal{K}\setminus\mathcal{L}}, {\mathbf{\bar
    R}}_{\mathcal{K}\setminus\mathcal{L},\mathcal{K}},
\mathbf{R}_{\mathcal{L},\mathcal{K}\setminus\mathcal{L}}]$, and
$\mathbf{R}_{\mathcal{L},\mathcal{K}}$ is a function of
$[ \mathbf{Q}_{\mathcal{L}},
\mathbf{\Phi}^n_{\mathcal{K}},\mathbf{\Psi}^n_{\mathcal{L}},
{\mathbf{\bar C}}^B_{\mathcal{K}\setminus \mathcal{L},\mathcal{K}},
\mathbf{R}_{\mathcal{K}\setminus \mathcal{L},\mathcal{L}}]$
(see~\eqref{equ:RLK_fn}). Hence, $W$ can be expressed as a function
of $[Z_1, Z_2]$. Since the functions $S$, $E$, $D$, and
$\mathbf{R}_{\mathcal{L},\mathcal{K}}$ are all PPT, the generation of
$W$ from $Z_1$ and $Z_2$ is also PPT. According to
Lemma~\ref{lem:useful4}, if $Z_1$ is conditionally independent of $B$
given $Z_2$, then $W$ will also be conditionally independent of $B$
given $Z_2$. The conditional independence between $Z_1$ and $B$ given
$Z_2$ is established by~\eqref{equ:a1-1},
where the first equality is due to Assumption~\ref{a:crypto}, and the
second equality results because
$[{\mathbf{\bar{C}}}^B_{\mathcal{I},\mathcal{K}\setminus\mathcal{L}},B]$ is
conditionally independent of
$[{\mathbf{\bar R}}_{\mathcal{J},\mathcal{K}\setminus\mathcal{L}},
\mathbf{R}_{\mathcal{K}\setminus\mathcal{L},\mathcal{L}},
\boldsymbol{\Phi}^n_{\mathcal{L}},\boldsymbol{\Psi}^n_{\mathcal{L}}]$
given
$[{\mathbf{\bar R}}_{\mathcal{I},\mathcal{K}\setminus\mathcal{L}},
\boldsymbol{\Phi}^n_{\mathcal{K}\setminus\mathcal{L}}]$.
\begin{figure*}[!hbtp]
\begin{align}
&\hspace{-5pt}
p_{\mathbf{Q}^0_{\mathcal{K}\setminus\mathcal{L}},
\mathbf{Q}^1_{\mathcal{K}\setminus\mathcal{L}},
\mathbf{Q}_{\mathcal{L}},
{\mathbf{\bar C}}_{\mathcal{J},\mathcal{K}\setminus\mathcal{L}},
\mathbf{C}_{\mathcal{K}\setminus\mathcal{L},\mathcal{L}},
{\mathbf{\bar R}}_{\mathcal{J},\mathcal{K}\setminus\mathcal{L}},
\mathbf{R}_{\mathcal{K}\setminus\mathcal{L},\mathcal{L}},
\boldsymbol{\Phi}^n_{\mathcal{L}},\boldsymbol{\Psi}^n_{\mathcal{L}}
\mid {\mathbf{\bar C}}^B_{\mathcal{I},\mathcal{K}\setminus\mathcal{L}},
{\mathbf{\bar R}}_{\mathcal{I},\mathcal{K}\setminus\mathcal{L}},
\boldsymbol{\Phi}^n_{\mathcal{K}\setminus\mathcal{L}},B}
(\mathbf{q}'_{\mathcal{K}\setminus\mathcal{L}},
\mathbf{q}''_{\mathcal{K}\setminus\mathcal{L}},
\mathbf{q}'_{\mathcal{L}},
\bar{\mathbf{c}}_{\mathcal{J},\mathcal{K}\setminus\mathcal{L}},
\mathbf{c}_{\mathcal{K}\setminus\mathcal{L},\mathcal{L}},
\bar{\mathbf{r}}_{\mathcal{J},\mathcal{K}\setminus\mathcal{L}},
\mathbf{r}_{\mathcal{K}\setminus\mathcal{L},\mathcal{L}},
\notag \\
& \hspace{20pt}
\boldsymbol{\phi}^n_{\mathcal{L}},
\boldsymbol{\psi}^n_{\mathcal{L}}
\mid\bar{\mathbf{c}}_{\mathcal{I},\mathcal{K}\setminus\mathcal{L}},
\bar{\mathbf{r}}_{\mathcal{I},\mathcal{K}\setminus\mathcal{L}},
\boldsymbol{\phi}^n_{\mathcal{K}\setminus\mathcal{L}},b)
\notag \\
&=
\delta_{\mathbf{q}^0_{\mathcal{K}\setminus\mathcal{L}}}
(\mathbf{q}'_{\mathcal{K}\setminus\mathcal{L}})
\cdot \delta_{\mathbf{q}^1_{\mathcal{K}\setminus\mathcal{L}}}
(\mathbf{q}''_{\mathcal{K}\setminus\mathcal{L}})
\cdot \delta_{\mathbf{q}_{\mathcal{L}}}(\mathbf{q}'_{\mathcal{L}})
\cdot p_{{\mathbf{\bar C}}_{\mathcal{J},\mathcal{K}\setminus\mathcal{L}}\mid
{\mathbf{\bar R}}_{\mathcal{J},\mathcal{K}\setminus\mathcal{L}},
\boldsymbol{\Phi}^n_{\mathcal{K}\setminus\mathcal{L}}}
({\mathbf{\bar c}}_{\mathcal{J},\mathcal{K}\setminus\mathcal{L}}\mid
{\mathbf{\bar r}}_{\mathcal{J},\mathcal{K}\setminus\mathcal{L}},
\boldsymbol{\phi}^n_{\mathcal{K}\setminus\mathcal{L}})
\cdot p_{\mathbf{C}_{\mathcal{K}\setminus\mathcal{L},\mathcal{L}}
\mid \mathbf{R}_{\mathcal{K}\setminus\mathcal{L},\mathcal{L}},\boldsymbol{\Phi}^n_{\mathcal{L}}}
(\mathbf{c}_{\mathcal{K}\setminus\mathcal{L},\mathcal{L}} \mid
\notag \\
& \hspace{20pt}
\mathbf{r}_{\mathcal{K}\setminus\mathcal{L},\mathcal{L}},
\boldsymbol{\phi}^n_{\mathcal{L}})
\cdot p_{\bar{\mathbf{R}}_{\mathcal{J},\mathcal{K}\setminus\mathcal{L}},
\mathbf{R}_{\mathcal{K}\setminus\mathcal{L},\mathcal{L}},
\boldsymbol{\Phi}^n_{\mathcal{L}},\boldsymbol{\Psi}^n_{\mathcal{L}}
\mid{\mathbf{\bar C}}^B_{\mathcal{I},\mathcal{K}\setminus\mathcal{L}},
{\mathbf{\bar R}}_{\mathcal{I},\mathcal{K}\setminus\mathcal{L}},
\boldsymbol{\Phi}^n_{\mathcal{K}\setminus\mathcal{L}},B}
({\mathbf{\bar r}}_{\mathcal{J},\mathcal{K}\setminus\mathcal{L}},
\mathbf{r}_{\mathcal{K}\setminus\mathcal{L},\mathcal{L}},
\boldsymbol{\phi}^n_{\mathcal{L}},\boldsymbol{\psi}^n_{\mathcal{L}}
\mid {\mathbf{\bar c}}_{\mathcal{I},\mathcal{K}\setminus\mathcal{L}},
{\mathbf{\bar r}}_{\mathcal{I},\mathcal{K}\setminus\mathcal{L}},
\boldsymbol{\phi}^n_{\mathcal{K}\setminus\mathcal{L}},b)
\notag \\
&=
\delta_{\mathbf{q}^0_{\mathcal{K}\setminus\mathcal{L}}}
(\mathbf{q}'_{\mathcal{K}\setminus\mathcal{L}})
\cdot \delta_{\mathbf{q}^1_{\mathcal{K}\setminus\mathcal{L}}}
(\mathbf{q}''_{\mathcal{K}\setminus\mathcal{L}})
\cdot \delta_{\mathbf{q}_{\mathcal{L}}}(\mathbf{q}'_{\mathcal{L}})
\cdot p_{{\mathbf{\bar C}}_{\mathcal{J},\mathcal{K}\setminus\mathcal{L}}\mid
{\mathbf{\bar R}}_{\mathcal{J},\mathcal{K}\setminus\mathcal{L}},
\boldsymbol{\Phi}^n_{\mathcal{K}\setminus\mathcal{L}}}
({\mathbf{\bar c}}_{\mathcal{J},\mathcal{K}\setminus\mathcal{L}}\mid
{\mathbf{\bar r}}_{\mathcal{J},\mathcal{K}\setminus\mathcal{L}},
\boldsymbol{\phi}^n_{\mathcal{K}\setminus\mathcal{L}})
\cdot p_{\mathbf{C}_{\mathcal{K}\setminus\mathcal{L},\mathcal{L}}
\mid \mathbf{R}_{\mathcal{K}\setminus\mathcal{L},\mathcal{L}}, \boldsymbol{\Phi}^n_{\mathcal{L}}}
(\mathbf{c}_{\mathcal{K}\setminus\mathcal{L},\mathcal{L}} \mid
\notag \\
& \hspace{20pt}
\mathbf{r}_{\mathcal{K}\setminus\mathcal{L},\mathcal{L}},
\boldsymbol{\phi}^n_{\mathcal{L}})
\cdot p_{{\mathbf{\bar R}}_{\mathcal{J},\mathcal{K}\setminus\mathcal{L}},
\mathbf{R}_{\mathcal{K}\setminus\mathcal{L},\mathcal{L}},
\boldsymbol{\Phi}^n_{\mathcal{L}},\boldsymbol{\Psi}^n_{\mathcal{L}}
\mid {\mathbf{\bar R}}_{\mathcal{I},\mathcal{K}\setminus\mathcal{L}},
\boldsymbol{\Phi}^n_{\mathcal{K}\setminus\mathcal{L}}}
({\mathbf{\bar r}}_{\mathcal{J},\mathcal{K}\setminus\mathcal{L}},
\mathbf{r}_{\mathcal{K}\setminus\mathcal{L},\mathcal{L}},
\boldsymbol{\phi}^n_{\mathcal{L}},\boldsymbol{\psi}^n_{\mathcal{L}}
\mid {\mathbf{\bar r}}_{\mathcal{I},\mathcal{K}\setminus\mathcal{L}},
\boldsymbol{\phi}^n_{\mathcal{K}\setminus\mathcal{L}}).
\label{equ:a1-1}
\end{align}
\end{figure*}

Now, we can apply Lemma~\ref{lem:useful3} with $Y$, $U$, $V$, $B$, and
$W$ as specified above to get a reduced PPT estimator
$\hat{B}_0(Y,U,V)$ satisfying
\begin{align}
& \hspace{-5pt}
\Pr(\hat{B}(Y,V,W)=B\mid {\mathbf{\bar R}}_{\mathcal{I},\mathcal{K}\setminus\mathcal{L}}
=\mathbf{\bar r}_{\mathcal{I},{\mathcal{K}}\setminus\mathcal{L}},
\boldsymbol{\Phi}^n_{\mathcal{K}\setminus\mathcal{L}}=\boldsymbol{\phi}^n_{\mathcal{K}\setminus\mathcal{L}},
\notag \\
& \hspace{20pt}
\mathbf{Q}^0_{\mathcal{K}\setminus\mathcal{L}}=\mathbf{q}^0_{\mathcal{K}\setminus\mathcal{L}},
\mathbf{Q}^1_{\mathcal{K}\setminus\mathcal{L}}=\mathbf{q}^1_{\mathcal{K}\setminus\mathcal{L}},
\mathbf{Q}_{\mathcal{L}}=\mathbf{q}_{\mathcal{L}}) 
\notag \\
&=
\Pr(\hat{B}_1(Y,U,V)=B \mid {\mathbf{\bar R}}^0_{\mathcal{I},\mathcal{K}
\setminus\mathcal{L}}
=\mathbf{\bar r}_{\mathcal{I},\mathcal{K}\setminus\mathcal{L}},
\notag \\
& \hspace{20pt}
{\mathbf{\bar R}}^1_{\mathcal{I},\mathcal{K}\setminus
\mathcal{L}}=0^{2(K-L-1)|\mathcal{X}|},
\boldsymbol{\Phi}^n_{\mathcal{K}\setminus\mathcal{L}}=\boldsymbol{\phi}^n_{\mathcal{K}\setminus\mathcal{L}}),
\label{equ:reduce1}
\end{align}
where the additional conditioning on
$\mathbf{Q}^0_{\mathcal{K}\setminus\mathcal{L}}$,
$\mathbf{Q}^1_{\mathcal{K}\setminus\mathcal{L}}$,
$\mathbf{Q}_{\mathcal{L}}$, and
${\mathbf{\bar R}}^1_{\mathcal{I},\mathcal{K}\setminus\mathcal{L}}$
applies because of the triviality of those random variables.  Let
$\mathbb{Q}(\mathbf{q}^0_{\mathcal{K}\setminus\mathcal{L}},
\mathbf{q}^1_{\mathcal{K}\setminus\mathcal{L}},\mathbf{q}_{\mathcal{L}})$
be the shorthand notation for the event
$\{\mathbf{Q}^0_{\mathcal{K}\setminus\mathcal{L}}
=\mathbf{q}^0_{\mathcal{K}\setminus\mathcal{L}},
\mathbf{Q}^1_{\mathcal{K}\setminus\mathcal{L}}
=\mathbf{q}^1_{\mathcal{K}\setminus\mathcal{L}},
\mathbf{Q}_{\mathcal{L}}=\mathbf{q}_{\mathcal{L}} \}$.  Clearly,
\eqref{equ:reduce1} further implies
\begin{align}
& \hspace{-5pt}
\Pr(\hat{B}(Y,V,W)=B \mid \mathbb{Q}(\mathbf{q}^0_{\mathcal{K}\setminus\mathcal{L}},
\mathbf{q}^1_{\mathcal{K}\setminus\mathcal{L}},\mathbf{q}_{\mathcal{L}}))
\notag \\
&=
\Pr(\hat{B}_1(Y,U,V)=B) 
\leq \frac{1}{2} + 2(K-L-1)|\mathcal{X}|\cdot F_{\mathrm{CPA}}(n),
\label{equ:a1-6}
\end{align}
where the equality results from the fact that we set
${\mathbf{\bar R}}^0_{\mathcal{I},\mathcal{K} \setminus\mathcal{L}} =
{\mathbf{\bar R}}_{\mathcal{I},\mathcal{K} \setminus\mathcal{L}}$ and
${\mathbf{\bar R}}^1_{\mathcal{I},\mathcal{K} \setminus\mathcal{L}}$
is trivially distributed, and the inequality is due to
Lemma~\ref{lem:mcpa} as the reduced estimator given by
Lemma~\ref{lem:useful3}
$\hat{B}_1(Y,U,V) = \hat{B}_1({\mathbf{\bar
    C}}^B_{\mathcal{I},\mathcal{K}\setminus\mathcal{L}}, {\mathbf{\bar
    R}}^0_{\mathcal{I},\mathcal{K}\setminus\mathcal{L}}, {\mathbf{\bar
    R}}^1_{\mathcal{I},\mathcal{K}\setminus\mathcal{L}},
\boldsymbol{\Phi}^n_{\mathcal{K}\setminus\mathcal{L}})$ is in the form
of the estimator in the multi-encryption CPA experiment with
$\mathscr{A}_1$ and $\mathscr{C}_1$ respectively as the CPA attacker
and challenger.

For cleaner notation in what follows, we write
$\boldsymbol{\Gamma}_0=E
(\mathbf{\bar{R}}_{\mathcal{I},\mathcal{K}\setminus\mathcal{L}};
\boldsymbol{\Phi}^n_{\mathcal{K}\setminus\mathcal{L}})$,
$\boldsymbol{\Gamma}_1=E(0^{2(K-L-1)|\mathcal{X}|};
\boldsymbol{\Phi}^n_{\mathcal{K}\setminus\mathcal{L}})$, and
$\boldsymbol{\Xi} = [\mathbf{Q}^0_{\mathcal{K}\setminus\mathcal{L}},
\mathbf{Q}^1_{\mathcal{K}\setminus \mathcal{L}},
\mathbf{Q}_{\mathcal{L}}, \mathbf{\bar
  C}_{\mathcal{J},\mathcal{K}\setminus\mathcal{L}},
\mathbf{C}_{\mathcal{K}\setminus\mathcal{L},\mathcal{L}}, \mathbf{\bar
  C}_{\mathcal{L},\mathcal{K}},
\mathbf{R}_{\mathcal{L},\mathcal{K}}, \mathbf{R}_{\mathcal{K}\setminus
  \mathcal{L},\mathcal{L}}, \\
\mathbf{\Phi}^n_{\mathcal{K}},\mathbf{\Psi}^n_{\mathcal{L}}]$.
Then, it is
simple to check in~\eqref{equ:a1-6} that
$\hat{B}(Y,V,W) = \hat{B}(\boldsymbol{\Xi}, \boldsymbol{\Gamma}_0,
\mathbf{G}^0_{\mathcal{K}\setminus\mathcal{L}})$ given $B=0$, and
$\hat{B}(Y,V,W) = \hat{B}(\boldsymbol{\Xi}, \boldsymbol{\Gamma}_1,
\mathbf{G}^0_{\mathcal{K}\setminus\mathcal{L}})$ given $B=1$.
Moreover, notice that both
$\hat{B} (\boldsymbol{\Xi}, \boldsymbol{\Gamma}_0,
\mathbf{G}^0_{\mathcal{K}\setminus\mathcal{L}})$ and
$\hat{B} (\boldsymbol{\Xi}, \boldsymbol{\Gamma}_1,
\mathbf{G}^0_{\mathcal{K}\setminus\mathcal{L}})$ are conditionally
independent of $B$ given
$[
\mathbf{Q}^0_{\mathcal{K}\setminus\mathcal{L}},
\mathbf{Q}^1_{\mathcal{K}\setminus\mathcal{L}},\mathbf{Q}_{\mathcal{L}}]$.
Hence, \eqref{equ:a1-6} implies
\begin{align}
& \hspace{-5pt}
\frac{1}{2}\Pr(\hat{B}
(\boldsymbol{\Xi}, \boldsymbol{\Gamma}_0,
\mathbf{G}^0_{\mathcal{K}\setminus\mathcal{L}})=0 \mid
\mathbb{Q}(\mathbf{q}^0_{\mathcal{K}\setminus\mathcal{L}},
\mathbf{q}^1_{\mathcal{K}\setminus\mathcal{L}},\mathbf{q}_{\mathcal{L}}))
\notag \\
&
+\frac{1}{2}\Pr(\hat{B}
(\boldsymbol{\Xi},\boldsymbol{\Gamma}_1,\mathbf{G}^0_{\mathcal{K}\setminus\mathcal{L}})=1
\mid \mathbb{Q}(\mathbf{q}^0_{\mathcal{K}\setminus\mathcal{L}},
\mathbf{q}^1_{\mathcal{K}\setminus\mathcal{L}},\mathbf{q}_{\mathcal{L}})) 
\notag \\
& \leq \frac{1}{2} + 2(K-L-1)|\mathcal{X}| \cdot F_{\mathrm{CPA}}(n).
\label{equ:a1}
\end{align}

Next, we construct the second experiment with attacker $\mathscr{A}_2$
and challenger $\mathscr{C}_2$ in the same way as in the previous
experiment, except that $\mathscr{A}_2$ assigns
${\mathbf{\bar
    R}}^0_{\mathcal{I},\mathcal{K}\setminus\mathcal{L}}=0^{2(K-L-1)|\mathcal{X}|}$,
${\mathbf{\bar R}}^1_{\mathcal{I},\mathcal{K}\setminus\mathcal{L}}
={\mathbf{\bar R}}_{\mathcal{I},\mathcal{K}\setminus\mathcal{L}}$ in
step 3) and computes
$\mathbf{G}^1_{\mathcal{K}\setminus\mathcal{L}}=\mathbf{q}^{1}_{\mathcal{K}\setminus\mathcal{L}}
\oplus\mathbf{\Sigma}_{\mathbf{R}_{\mathcal{K},\mathcal{K}\setminus
    \mathcal{L}}}$ in step 8).  Following a similar analysis, we 
get for this experiment,
\begin{align}
& \hspace{-5pt}
\frac{1}{2}\Pr(\hat{B}
(\boldsymbol{\Xi}, \boldsymbol{\Gamma}_1,
\mathbf{G}^1_{\mathcal{K}\setminus\mathcal{L}})=0 \mid
\mathbb{Q}(\mathbf{q}^0_{\mathcal{K}\setminus\mathcal{L}},
\mathbf{q}^1_{\mathcal{K}\setminus\mathcal{L}},\mathbf{q}_{\mathcal{L}}))
\notag \\
&
+\frac{1}{2}\Pr(\hat{B}
(\boldsymbol{\Xi},\boldsymbol{\Gamma}_0,\mathbf{G}^1_{\mathcal{K}\setminus\mathcal{L}})=1
\mid \mathbb{Q}(\mathbf{q}^0_{\mathcal{K}\setminus\mathcal{L}},
\mathbf{q}^1_{\mathcal{K}\setminus\mathcal{L}},\mathbf{q}_{\mathcal{L}})) 
\notag \\
& \leq \frac{1}{2} + 2(K-L-1)|\mathcal{X}| \cdot F_{\mathrm{CPA}}(n).
\label{equ:a2}
\end{align}

As shown in Figure \ref{fig:tdaa3}, we construct the third
  experiment with attacker $\mathscr{A}_{3}$ and challenger
  $\mathscr{C}_{3}$ as follows:
\begin{figure}
  \centering
\includegraphics[width=0.49\textwidth]{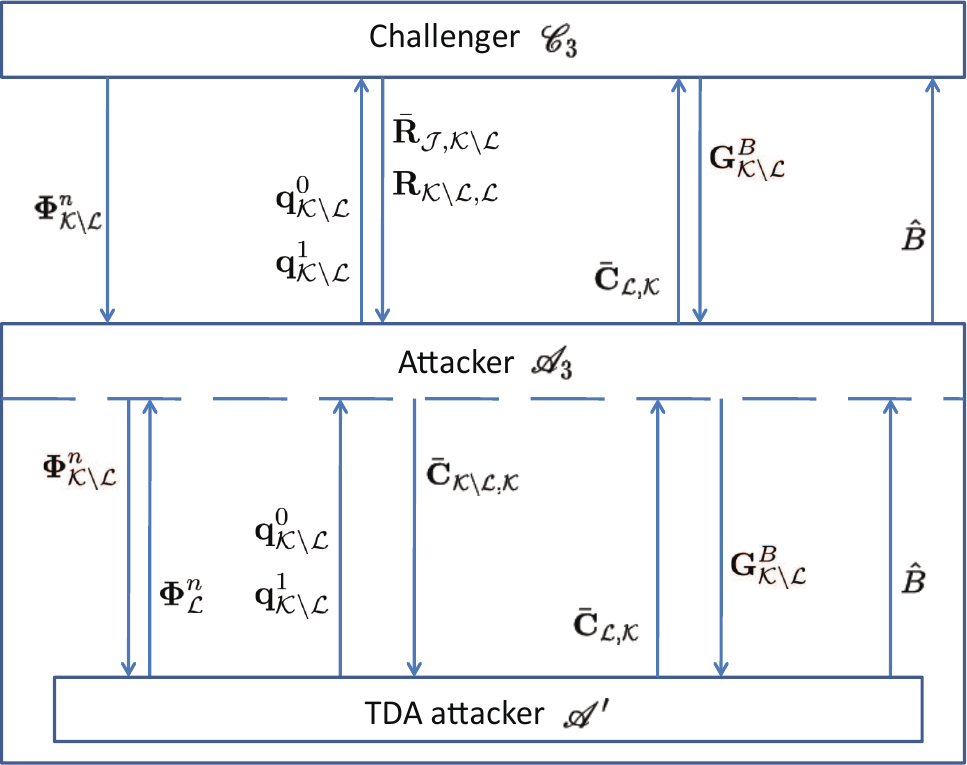}\\
\caption{The third constructed experiment for proving Lemma \ref{lem:tda}.}\label{fig:tdaa3}
\end{figure}

\begin{enumerate}

\item Same as step 1) in the first experiment with $\mathscr{A}_3$ and
  $\mathscr{C}_{3}$ taking the roles of $\mathscr{A}_1$ and
  $\mathscr{C}_{1}$, respectively.
  
\item $\mathscr{A}'$ selects $\mathbf{q}_{\mathcal{L}}$,
  $\mathbf{q}^{0}_{\mathcal{K}\setminus \mathcal{L}}$, and
  $\mathbf{q}^{1}_{\mathcal{K}\setminus \mathcal{L}}$ as in step 2) of
  the TDA experiment, and then passes
  $[\mathbf{q}^{0}_{\mathcal{K}\setminus \mathcal{L}},
  \mathbf{q}^{1}_{\mathcal{K}\setminus \mathcal{L}}]$ to
  $\mathscr{A}_3$, who then passes them on to $\mathscr{C}_3$.

\item $\mathscr{C}_{3}$ generates
  $\mathbf{R}_{\mathcal{K}\setminus\mathcal{L},\mathcal{K}}$ with
  $\mathbf{R}_{k,l}$
  i.i.d. $\sim u(\mathcal{N}_m^{|\mathcal{X}|})$ for
  $k\neq l$, and
  $\mathbf{R}_{k,k}=\ominus\bigoplus_{l\in\mathcal{K}\setminus
    k}\mathbf{R}_{k,l}$ according to~\eqref{equ:Rkk}, and then gives
  $[{\mathbf{\bar R}}_{\mathcal{J},\mathcal{K}\setminus\mathcal{L}},
  \mathbf{R}_{\mathcal{K}\setminus\mathcal{L},\mathcal{L}}]$ to
  $\mathscr{A}_3$.

\item $\mathscr{A}_{3}$ computes
  ${\mathbf{\bar C}}_{\mathcal{I},\mathcal{K}
    \setminus\mathcal{L}}=E(0^{2(K-L-1)|\mathcal{X}|};
  \boldsymbol{\Phi}^n_{\mathcal{K}\setminus\mathcal{L}})$,
  ${\mathbf{\bar C}}_{\mathcal{J},\mathcal{K}\setminus\mathcal{L}}
  =E(\bar{\mathbf{R}}_{\mathcal{J},\mathcal{K}\setminus\mathcal{L}};
  \boldsymbol{\Phi}^n_{\mathcal{K}\setminus\mathcal{L}})$,
  $\mathbf{C}_{\mathcal{K}\setminus\mathcal{L},\mathcal{L}}
  =E(\mathbf{R}_{\mathcal{K}\setminus\mathcal{L},\mathcal{L}};
  \boldsymbol{\Phi}^n_{\mathcal{L}})$, and gives
  ${\mathbf{\bar C}}_{\mathcal{K}\setminus\mathcal{L},\mathcal{K}}$ to
  $\mathscr{A}'$.

\item $\mathscr{A}'$ follows step 4) of the TDA experiment to decrypt
  $\mathbf{C}_{\mathcal{K}\setminus\mathcal{L},\mathcal{L}}$, generate
  $\mathbf{R}_{\mathcal{L}, \mathcal{K}}$, encrypt to obtain
  $\mathbf{\bar C}_{\mathcal{L}, \mathcal{K}}$, and send
  $\mathbf{\bar C}_{\mathcal{L}, \mathcal{K}}$ to $\mathscr{A}_3$, who
  then passes on
  $\mathbf{C}_{\mathcal{L},\mathcal{K} \setminus\mathcal{L}}$ to
  $\mathscr{C}_{3}$.

\item $\mathscr{C}_{3}$ calculates
  $\mathbf{R}_{\mathcal{L}, \mathcal{K} \setminus\mathcal{L}} =
  D(\mathbf{C}_{\mathcal{L}, \mathcal{K} \setminus\mathcal{L}} ;
  \mathbf{\Psi}^n_{\mathcal{K} \setminus\mathcal{L}})$. Then,
  $\mathscr{C}_{3}$ selects an independent bit $B\in\{0,1\}$ with
  equal probabilities, computes
  $\mathbf{G}^B_{\mathcal{K}\setminus\mathcal{L}}
  =\mathbf{q}^{B}_{\mathcal{K}\setminus\mathcal{L}} \oplus
  \mathbf{\Sigma}_{\mathbf{R}_{\mathcal{K},\mathcal{K}\setminus
      \mathcal{L}}}$, and gives
  $\mathbf{G}^B_{\mathcal{K}\setminus\mathcal{L}}$ to $\mathscr{A}_3$,
  who passes it on to $\mathscr{A}'$.

\item $\mathscr{A}'$ uses
  $\hat{B}=\hat{B}(\mathbf{q}^0_{\mathcal{K}\setminus
    \mathcal{L}},\mathbf{q}^1_{\mathcal{K}\setminus\mathcal{L}},
  \mathbf{q}_{\mathcal{L}}, {\mathbf{\bar
      C}}_{\mathcal{K},\mathcal{K}},
  \mathbf{G}^B_{\mathcal{K}\setminus\mathcal{L}},\mathbf{R}_{\mathcal{L},\mathcal{K}}, \\
  \mathbf{R}_{\mathcal{K}\setminus
    \mathcal{L},\mathcal{L}},\boldsymbol{\Phi}^n_{\mathcal{K}},
  \boldsymbol{\Psi}^n_{\mathcal{L}})$ in step 6) of the TDA experiment
  with the input arguments as specified to estimate $B$, and reports
  $\hat{B}$ to $\mathscr{A}_{3}$, who then passes it on to
  $\mathscr{C}_{3}$.
\end{enumerate}

We will again use Lemmas~\ref{lem:useful3} and~\ref{lem:useful4} by
letting $Y=\mathbf{G}^B_{\mathcal{K}\setminus\mathcal{L}}$,
$U={\mathbf{\bar R}}_{\mathcal{J},\mathcal{K}\setminus\mathcal{L}}$,
$V=[\mathbf{Q}^0_{\mathcal{K}\setminus\mathcal{L}},
\mathbf{Q}^1_{\mathcal{K}\setminus\mathcal{L}},
\mathbf{R}_{\mathcal{L},\mathcal{K}\setminus\mathcal{L}},
\mathbf{R}_{\mathcal{K}\setminus\mathcal{L},\mathcal{L}},
\boldsymbol{\Phi}^n_{\mathcal{K}\setminus\mathcal{L}}]$,
$Z_1=[\mathbf{Q}_{\mathcal{L}}, {\mathbf{\bar
    C}}_{\mathcal{K}\setminus\mathcal{L},\mathcal{K}},
\boldsymbol{\Phi}^n_{\mathcal{L}},\boldsymbol{\Psi}^n_{\mathcal{L}}]$,
$Z_2=[Y,U,V]$, and
$W=[{\mathbf{\bar C}}_{\mathcal{L},\mathcal{K}}, {\mathbf{R}}_{\mathcal{L},\mathcal{L}}, Z_1]$ this time. 
Like before, it is easy to check that we again have $W$ as a PPT
function of $[Z_1,Z_2]$ in this case. Thus, we may apply
Lemma~\ref{lem:useful4} again to obtain that $W$ is conditionally
independent of $B$ given $Z_2$ as long as $Z_1$ and $B$ are
conditionally independent given $Z_2$. This latter fact is established
by~\eqref{equ:a3-0},
where the equality is based on Assumption~\ref{a:crypto},
\eqref{equ:key_fn}, and the fact that
${\mathbf{\bar C}}_{\mathcal{I},\mathcal{K}\setminus\mathcal{L}} =
E(0^{2(K-L-1)|\mathcal{X}|};
\boldsymbol{\Phi}^n_{\mathcal{K}\setminus\mathcal{L}})$. 
\begin{figure*}
\begin{align}
&p_{\mathbf{Q}_{\mathcal{L}}, {\mathbf{\bar C}}_{\mathcal{K}\setminus\mathcal{L},\mathcal{K}},
\boldsymbol{\Phi}^n_{\mathcal{L}}, \boldsymbol{\Psi}^n_{\mathcal{L}}
\mid 
\mathbf{G}^B_{\mathcal{K}\setminus\mathcal{L}},
{\mathbf{\bar R}}_{{\mathcal{J}},\mathcal{K}\setminus\mathcal{L}},
\mathbf{Q}^0_{\mathcal{K}\setminus\mathcal{L}}, \mathbf{Q}^1_{\mathcal{K}\setminus\mathcal{L}},
{\mathbf{\bar R}}_{\mathcal{L},\mathcal{K}\setminus\mathcal{L}},
\mathbf{R}_{\mathcal{K}\setminus\mathcal{L},\mathcal{L}},
\boldsymbol{\Phi}^n_{\mathcal{K}\setminus\mathcal{L}},B}
(\mathbf{q}^\prime_{\mathcal{L}}, {\mathbf{\bar c}}_{\mathcal{K}\setminus\mathcal{L},\mathcal{K}},
\boldsymbol{\phi}^n_{\mathcal{L}}, \boldsymbol{\psi}^n_{\mathcal{L}}
\mid 
\mathbf{g}_{\mathcal{K}\setminus\mathcal{L}},
{\mathbf{\bar r}}_{{\mathcal{J}},\mathcal{K}\setminus\mathcal{L}},
\mathbf{q}^{\prime}_{\mathcal{K}\setminus\mathcal{L}}, 
\mathbf{q}^{\prime\prime}_{\mathcal{K}\setminus\mathcal{L}},
{\mathbf{\bar r}}_{\mathcal{L},\mathcal{K}\setminus\mathcal{L}},
\notag \\
& \hspace{20pt}
\mathbf{r}_{\mathcal{K}\setminus\mathcal{L},\mathcal{L}},
\boldsymbol{\phi}^n_{\mathcal{K}\setminus\mathcal{L}},b)
\notag \\
&=
\delta_{\mathbf{q}_{\mathcal{L}}}(\mathbf{q}'_{\mathcal{L}}) 
\cdot
p_{{\mathbf{\bar C}}_{\mathcal{J},\mathcal{K}\setminus\mathcal{L}}
\mid {\mathbf{\bar R}}_{\mathcal{J},\mathcal{K}\setminus\mathcal{L}}, 
\boldsymbol{\Phi}^n_{\mathcal{K}\setminus\mathcal{L}}} 
(\mathbf{\bar c}_{\mathcal{J},\mathcal{K}\setminus\mathcal{L}}
\mid {\mathbf{\bar r}}_{\mathcal{J},\mathcal{K}\setminus\mathcal{L}}, 
\boldsymbol{\phi}^n_{\mathcal{K}\setminus\mathcal{L}})
\cdot
p_{{\mathbf{C}}_{\mathcal{K}\setminus\mathcal{L}, \mathcal{L}} 
\mid \mathbf{R}_{\mathcal{K}\setminus\mathcal{L},\mathcal{L}}, 
\boldsymbol{\Phi}^n_{\mathcal{L}}}
( {\mathbf{c}}_{\mathcal{K}\setminus\mathcal{L}, \mathcal{L}} 
\mid \mathbf{r}_{\mathcal{K}\setminus\mathcal{L},\mathcal{L}}, 
\boldsymbol{\phi}^n_{\mathcal{L}} )
\notag \\
& \hspace{20pt}
\cdot 
p_{{\mathbf{\bar C}}_{\mathcal{I},\mathcal{K}\setminus\mathcal{L}}
\mid \boldsymbol{\Phi}^n_{\mathcal{K}\setminus\mathcal{L}}} 
(\mathbf{\bar c}_{\mathcal{I},\mathcal{K}\setminus\mathcal{L}} \mid
\boldsymbol{\phi}^n_{\mathcal{K}\setminus\mathcal{L}})
\cdot p_{\boldsymbol{\Phi}^n_{\mathcal{L}},\boldsymbol{\Psi}^n_{\mathcal{L}}
\mid\boldsymbol{\Phi}^n_{\mathcal{K}\setminus\mathcal{L}}}
(\boldsymbol{\phi}^n_{\mathcal{L}},\boldsymbol{\psi}^n_{\mathcal{L}}
\mid\boldsymbol{\phi}^n_{\mathcal{K}\setminus\mathcal{L}})
\label{equ:a3-0}
\end{align}
\end{figure*}

Further, expressed in the previous notation
$\hat{B}(Y,V,W) = \hat{B}(\boldsymbol{\Xi}, \boldsymbol{\Gamma}_1,
\mathbf{G}^B_{\mathcal{K}\setminus\mathcal{L}})$. Thus,
by applying Lemma~\ref{lem:useful3} with $Y$, $U$, $V$, and $W$ as
specified, we get a reduced PPT estimator
$\hat{B}_3(\mathbf{G}^B_{\mathcal{K}\setminus\mathcal{L}},
\mathbf{Q}^0_{\mathcal{K}\setminus\mathcal{L}},
\mathbf{Q}^1_{\mathcal{K}\setminus\mathcal{L}}, {\mathbf{\bar
  R}}_{\mathcal{K}\setminus\mathcal{I},\mathcal{K}\setminus\mathcal{L}},
\mathbf{R}_{\mathcal{K}\setminus\mathcal{L},\mathcal{L}},
\boldsymbol{\Phi}^n_{\mathcal{K}\setminus\mathcal{L}})$ that satisfies
\begin{align}
& 
  \Pr ( \hat{B}(\boldsymbol{\Xi}, \boldsymbol{\Gamma}_1,
  \mathbf{G}^B_{\mathcal{K}\setminus\mathcal{L}}) = B \mid
  \mathbb{Q}(\mathbf{q}^0_{\mathcal{K}\setminus\mathcal{L}},
  \mathbf{q}^1_{\mathcal{K}\setminus\mathcal{L}},
  \mathbf{q}_{\mathcal{L}}), 
\notag \\
& \hspace{10pt}
  \mathbf{\bar R}_{\mathcal{K}\setminus\mathcal{I}, \mathcal{K}\setminus\mathcal{L}} 
 = \mathbf{\bar r}_{\mathcal{K}\setminus\mathcal{I}, \mathcal{K}\setminus\mathcal{L}}, 
  \mathbf{R}_{\mathcal{K}\setminus\mathcal{L},\mathcal{L}} 
 = \mathbf{r}_{\mathcal{K}\setminus\mathcal{L},\mathcal{L}},
  \boldsymbol{\Phi}^n_{\mathcal{K}\setminus\mathcal{L}}
 = \boldsymbol{\phi}^n_{\mathcal{K}\setminus\mathcal{L}})
\notag \\
&= \Pr ( \hat{B}_3 = B \mid 
  \mathbb{Q}(\mathbf{q}^0_{\mathcal{K}\setminus\mathcal{L}},
  \mathbf{q}^1_{\mathcal{K}\setminus\mathcal{L}},
  \mathbf{q}_{\mathcal{L}}), 
  \mathbf{\bar R}_{\mathcal{K}\setminus\mathcal{I}, \mathcal{K}\setminus\mathcal{L}} 
 = \mathbf{\bar r}_{\mathcal{K}\setminus\mathcal{I}, \mathcal{K}\setminus\mathcal{L}}, 
\notag \\
& \hspace{30pt}
  \mathbf{R}_{\mathcal{K}\setminus\mathcal{L},\mathcal{L}} 
 = \mathbf{r}_{\mathcal{K}\setminus\mathcal{L},\mathcal{L}},
  \boldsymbol{\Phi}^n_{\mathcal{K}\setminus\mathcal{L}}
 = \boldsymbol{\phi}^n_{\mathcal{K}\setminus\mathcal{L}})
\notag \\
& = \frac{1}{2},
\label{equ:A3}
\end{align}
where we have used the triviality of the distribution of
$\mathbf{Q}_{\mathcal{L}}$ in the first equality, and the last
equality can be obtained based on Lemma~\ref{lem:useful2} as shown in
Appendix~\ref{app:A3}.

Since both
$\hat{B} (\boldsymbol{\Xi}, \boldsymbol{\Gamma}_1,
\mathbf{G}^0_{\mathcal{K}\setminus\mathcal{L}})$ and
$\hat{B} (\boldsymbol{\Xi}, \boldsymbol{\Gamma}_1,
\mathbf{G}^1_{\mathcal{K}\setminus\mathcal{L}})$ are conditionally
independent of $B$ given
$[\mathbf{Q}^0_{\mathcal{K}\setminus\mathcal{L}},
\mathbf{Q}^1_{\mathcal{K}\setminus\mathcal{L}},
\mathbf{Q}_{\mathcal{L}}]$, (\ref{equ:A3}) implies
\begin{align}
& \hspace{-5pt}
\frac{1}{2} \Pr( \hat{B}
(\boldsymbol{\Xi},\boldsymbol{\Gamma}_1,\mathbf{G}^0_{\mathcal{K}\setminus\mathcal{L}})=0\mid
\mathbb{Q}(\mathbf{q}^0_{\mathcal{K}\setminus\mathcal{L}},
\mathbf{q}^1_{\mathcal{K}\setminus\mathcal{L}},\mathbf{q}_{\mathcal{L}}))
\notag \\
&
+\frac{1}{2} \Pr( \hat{B}
(\boldsymbol{\Xi},\boldsymbol{\Gamma}_1,\mathbf{G}^1_{\mathcal{K}\setminus\mathcal{L}})=1
\mid \mathbb{Q}(\mathbf{q}^0_{\mathcal{K}\setminus\mathcal{L}},
\mathbf{q}^1_{\mathcal{K}\setminus\mathcal{L}},\mathbf{q}_{\mathcal{L}}))
\notag \\
&=\frac{1}{2}.
\label{equ:a3}
\end{align}

Finally, adding up \eqref{equ:a1}, \eqref{equ:a2}, and \eqref{equ:a3}
from the three experiments constructed above gives
\begin{align}
& \hspace{-5pt}
\Pr( \hat{B}
(\boldsymbol{\Xi}, \boldsymbol{\Gamma}_0,
\mathbf{G}^B_{\mathcal{K}\setminus\mathcal{L}}) = B \mid
\mathbb{Q}(\mathbf{q}^0_{\mathcal{K}\setminus\mathcal{L}},
\mathbf{q}^1_{\mathcal{K}\setminus\mathcal{L}},\mathbf{q}_{\mathcal{L}}))
\notag \\
&= \frac{1}{2} \Pr( \hat{B}
(\boldsymbol{\Xi}, \boldsymbol{\Gamma}_0,
\mathbf{G}^0_{\mathcal{K}\setminus\mathcal{L}}) =0 \mid
\mathbb{Q}(\mathbf{q}^0_{\mathcal{K}\setminus\mathcal{L}},
\mathbf{q}^1_{\mathcal{K}\setminus\mathcal{L}},\mathbf{q}_{\mathcal{L}}))
\notag \\
&
+\frac{1}{2}\Pr(\hat{B}
(\boldsymbol{\Xi}, \boldsymbol{\Gamma}_0,
\mathbf{G}^1_{\mathcal{K}\setminus\mathcal{L}}) = 1
\mid \mathbb{Q}(\mathbf{q}^0_{\mathcal{K}\setminus\mathcal{L}},
\mathbf{q}^1_{\mathcal{K}\setminus\mathcal{L}},\mathbf{q}_{\mathcal{L}})) 
\notag \\
& \leq \frac{1}{2} + 4(K-L-1)|\mathcal{X}| \cdot F_{\mathrm{CPA}}(n).
\label{equ:a1a2a3}
\end{align}
Note that
$\hat{B}(\boldsymbol{\Xi}, \boldsymbol{\Gamma}_0,
\mathbf{G}^B_{\mathcal{K}\setminus\mathcal{L}})$ is exactly the
estimator $\hat{B}$ used by $\mathscr{A}'$ in step 6) of the TDA
experiment. As a result, \eqref{equ:a1a2a3}
establishes~\eqref{equ:tda}.

\end{proof}



\subsection{Proof of Theorem \ref{thm:tea}}

As discussed before, we will reduce the TEA experiment to a TDA
experiment by constructing a TDA attacker $\mathscr{A}'$ and her
estimator $\hat{B}$ (see step 6) of the TDA experiment from the TEA
attacker $\mathscr{A}$ and her estimator
$\mathbf{\hat{Q}}_{\mathcal{K}\setminus \mathcal{L}}$
(see~\eqref{equ:Qhat_fn}). This reduction allows us to express the
winning probability of $\mathscr{A}$ as that of $\mathscr{A}'$, thus
proving~\eqref{equ:edra} using Lemma~\ref{lem:tda}. The steps of the constructed
TDA experiment, shown in Figure~\ref{fig:teaproof}, are
as follows:
\begin{figure}
  \centering \includegraphics[width=0.48\textwidth]{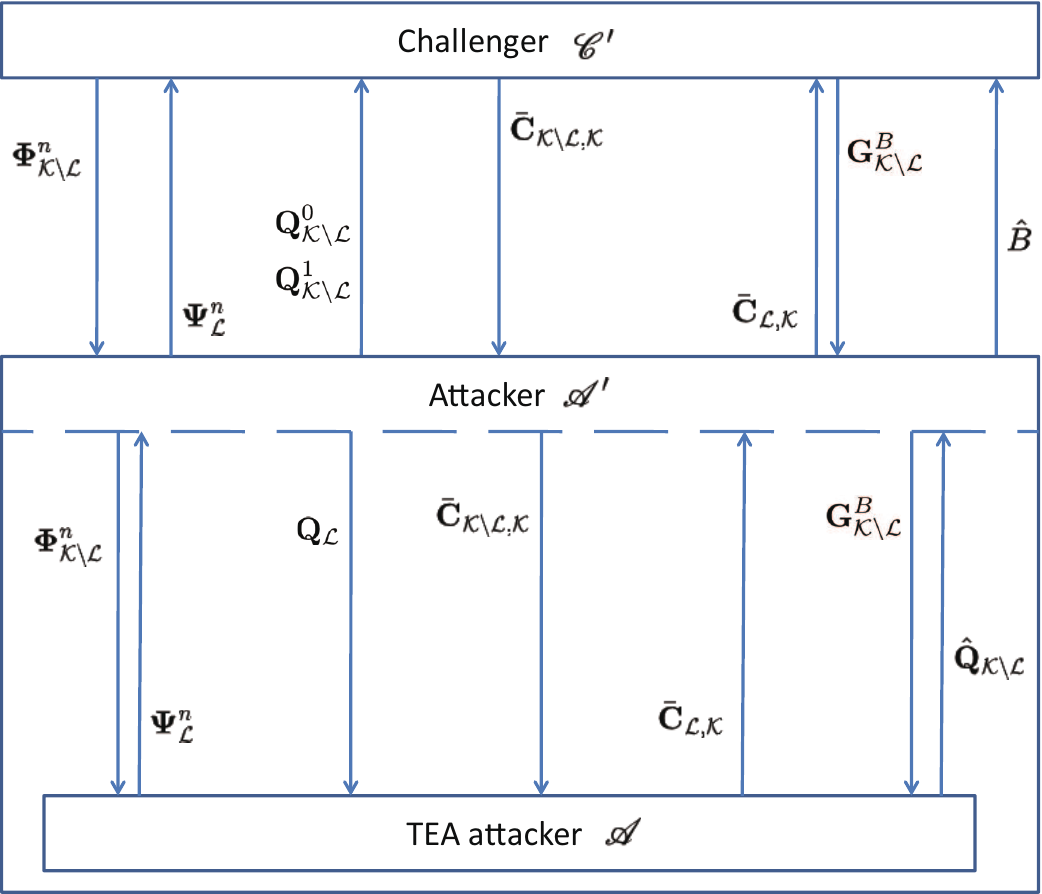}
  \caption{The constructed experiment for proving Theorem \ref{thm:tea}.}
\label{fig:teaproof}
\end{figure}
\begin{enumerate}
\item $\mathscr{C}'$ runs $S(1^{n})$ to get the key pair collection
  $(\boldsymbol{\Phi}^n_{\mathcal{K} \setminus \mathcal{L}},
  \boldsymbol{\Psi}^n_{\mathcal{K} \setminus \mathcal{L}})$ and gives
  $\boldsymbol{\Phi}^n_{\mathcal{K} \setminus \mathcal{L}}$ to
  $\mathscr{A}'$, who passes it on to $\mathscr{A}$. Then,
  $\mathscr{A}$ generates the key pair collection
  $(\boldsymbol{\Phi}_{\mathcal{L}}^n,
  \boldsymbol{\Psi}_{\mathcal{L}}^n)$ according to~\eqref{equ:key_fn}
  and gives $\boldsymbol{\Phi}_{\mathcal{L}}^n$ to $\mathscr{A}'$, who
  then passes it on to $\mathscr{C}'$.

\item $\mathscr{A}'$ draws
  $[\mathbf{Q}^0_{\mathcal{K}\setminus\mathcal{L}},
  \mathbf{Q}_{\mathcal{L}}] \in \mathcal{Q}_t^K(\mathcal{X})$
  according to $p_{\mathbf{Q}_{\mathcal{K}}}$ and
  $\mathbf{Q}^1_{\mathcal{K}\setminus\mathcal{L}} \in
  \mathcal{Q}_t^{K-L}(\mathcal{X})$ according to
  $p_{\mathbf{Q}_{\mathcal{K}\setminus\mathcal{L}} \mid
    \boldsymbol{\Sigma}_{\mathbf{Q}_{\mathcal{K}\setminus\mathcal{L}}},
    \mathbf{Q}_{\mathcal{L}}} (\cdot\mid
  \boldsymbol{\Sigma}_{\mathbf{Q}^0_{\mathcal{K}\setminus\mathcal{L}}},
  \mathbf{Q}_{\mathcal{L}})$.  Then, $\mathscr{A}'$ gives
  $[\mathbf{Q}^0_{\mathcal{K}\setminus\mathcal{L}},
  \mathbf{Q}^1_{\mathcal{K}\setminus\mathcal{L}}]$ to $\mathscr{C}'$
  and $\mathbf{Q}_{\mathcal{L}}$ to $\mathscr{A}$.

\item Same as step 3) of the TDA experiment. Then $\mathscr{A}'$
  passes
  $\mathbf{\bar C}_{\mathcal{K}\setminus \mathcal{L},\mathcal{K}}$ to
  $\mathscr{A}$.

\item $\mathscr{A}$ follows step 4) of the TEA experiment to calculate
  $\mathbf{R}_{\mathcal{K} \setminus\mathcal{L},\mathcal{L}}$,
  $\mathbf{R}_{\mathcal{L},\mathcal{K}}$, and
  ${\mathbf{\bar C}}_{\mathcal{L}, \mathcal{K}}$. Then, she gives
  ${\mathbf{\bar C}}_{\mathcal{L}, \mathcal{K}}$ to $\mathscr{A}'$,
  who passes it on to $\mathscr{C}'$.

\item Same as step 5) of the TDA experiment. Then, $\mathscr{A}'$
  passes $\mathbf{G}^B_{\mathcal{K}\setminus \mathcal{L}}$ on to
  $\mathscr{A}$.

\item $\mathscr{A}$ follows step 6) of the TEA experiment to compute
  $\mathbf{G}_{\mathcal{L}}$ and reports her estimate
  ${\mathbf{\hat Q}}_{\mathcal{K}\setminus \mathcal{L}} = \mathbf{\hat
    Q}_{\mathcal{K}\setminus\mathcal{L}} (\mathbf{Q}_{\mathcal{L}}, 
  \mathbf{\Phi}^n_{\mathcal{K}}, 
  \mathbf{\Psi}^n_{\mathcal{L}},\\ 
  \mathbf{\bar C}_{\mathcal{K},\mathcal{K}},  [\mathbf{G}^B_{\mathcal{K}\setminus \mathcal{L}},
  \mathbf{G}_{\mathcal{L}}],
  \mathbf{R}_{\mathcal{L},\mathcal{K}},
  \mathbf{R}_{\mathcal{K}\setminus \mathcal{L},\mathcal{L}})$ to
  $\mathscr{A}'$.

\item Given $\tau\geq 0$, $\mathscr{A}'$ estimates $B$ by setting
  $\hat{B}=0$ if
  ${\mathbf{\hat Q}}_{\mathcal{K}\backslash \mathcal{L}} \in
  \mathcal{N}_{\tau}(\mathbf{Q}^0_{\mathcal{K}\backslash\mathcal{L}})$
  and $\hat{B}=1$ otherwise. Then, $\mathscr{A}'$ reports $\hat{B}$ to
  $\mathscr{C}'$. 
\end{enumerate}

Note that the estimator
$\hat B = \hat B(\mathbf{Q}^0_{\mathcal{K}\setminus\mathcal{L}},
\mathbf{Q}^1_{\mathcal{K}\setminus\mathcal{L}},
\mathbf{Q}_{\mathcal{L}}, {\mathbf{\bar C}}_{\mathcal{K},\mathcal{K}},\\
\mathbf{G}_{\mathcal{L}},
\mathbf{G}^B_{\mathcal{K}\setminus\mathcal{L}},
\mathbf{R}_{\mathcal{L},\mathcal{K}},
\mathbf{R}_{\mathcal{K}\setminus\mathcal{L},\mathcal{L}},
\boldsymbol{\Phi}^n_{\mathcal{K}}, \boldsymbol{\Psi}^n_{\mathcal{L}})$
because of the functional form of
${\mathbf{\hat Q}}_{\mathcal{K}\backslash \mathcal{L}}$
(see~\eqref{equ:Qhat_fn}).
We will use Lemma~\ref{lem:useful3} by setting
$Y = [{\mathbf{\bar C}}_{\mathcal{K},\mathcal{K}},
\mathbf{G}^B_{\mathcal{K}\setminus\mathcal{L}},
\mathbf{R}_{\mathcal{L},\mathcal{K}},
\mathbf{R}_{\mathcal{K}\setminus\mathcal{L},\mathcal{L}},
\boldsymbol{\Phi}^n_{\mathcal{K}},
\boldsymbol{\Psi}^n_{\mathcal{L}}]$, $U=\emptyset$,
$V=[\mathbf{Q}^0_{\mathcal{K}\setminus\mathcal{L}},
\mathbf{Q}^1_{\mathcal{K}\setminus\mathcal{L}},
\mathbf{Q}_{\mathcal{L}}]$, and $W=\mathbf{G}_{\mathcal{L}}$. Since
$\mathbf{G}_{\mathcal{L}} = \mathbf{Q}_{\mathcal{L}} \oplus
\boldsymbol{\Sigma}_{\mathbf{R}_{\mathcal{K},\mathcal{L}}}$ and
$\boldsymbol{\Sigma}_{\mathbf{R}_{\mathcal{K},\mathcal{L}}}$ is a
deterministic function of
$[\mathbf{R}_{\mathcal{L},\mathcal{K}},
\mathbf{R}_{\mathcal{K}\setminus\mathcal{L},\mathcal{L}}]$, it is
clear that $W$ and $B$ are conditionally independent given $[Y,U,V]$
and the generation of $W$ from $[Y,U,V]$ is PPT.  Thus,
Lemma~\ref{lem:useful3} applies in this case to give a PPT estimator
$\hat B_0(Y,U,V)$ satisfying
\begin{align}
&\hspace{-5pt}
\Pr(\hat{B}(Y,V,W)=B \mid
\mathbb{Q}(\mathbf{q}^0_{\mathcal{K}\setminus\mathcal{L}},
\mathbf{q}^1_{\mathcal{K}\setminus\mathcal{L}},\mathbf{q}_{\mathcal{L}}))
\notag \\
&=
\Pr(\hat{B}_0(Y,U,V)=B \mid
\mathbb{Q}(\mathbf{q}^0_{\mathcal{K}\setminus\mathcal{L}},
\mathbf{q}^1_{\mathcal{K}\setminus\mathcal{L}},\mathbf{q}_{\mathcal{L}}))
\notag \\
& \leq \frac{1}{2}
+4(K-L-1)|\mathcal{X}|\cdot F_{\mathrm{CPA}}(n),
\label{equ:tea-1}
\end{align}
where the last inequality is due to Lemma~\ref{lem:tda} because the
estimator
$\hat{B}_0(Y,U,V) = \hat{B}_0(
\mathbf{Q}^0_{\mathcal{K}\setminus\mathcal{L}},
\mathbf{Q}^1_{\mathcal{K}\setminus\mathcal{L}},
\mathbf{Q}_{\mathcal{L}}, 
{\mathbf{\bar C}}_{\mathcal{K},\mathcal{K}},
\mathbf{G}^B_{\mathcal{K}\setminus\mathcal{L}}, \\
\mathbf{R}_{\mathcal{L},\mathcal{K}}, 
\mathbf{R}_{\mathcal{K}\setminus\mathcal{L},\mathcal{L}},
\boldsymbol{\Phi}^n_{\mathcal{K}}, \boldsymbol{\Psi}^n_{\mathcal{L}})$
is exactly in the form of the estimator in step 6) of the TDA
experiment.

By the definition of $\hat B$ in step 7) of the constructed TDA
experiment, we have
\begin{align}
&
\Pr(\hat{B}(Y,V,W)=B \mid
\mathbb{Q}(\mathbf{q}^0_{\mathcal{K}\setminus\mathcal{L}},
\mathbf{q}^1_{\mathcal{K}\setminus\mathcal{L}},\mathbf{q}_{\mathcal{L}}))
\notag \\
&= \frac{1}{2} 
\Pr({\mathbf{\hat Q}}_{\mathcal{K}
\setminus\mathcal{L}}
\in\mathcal{N}_{\tau}
(\mathbf{Q}^0_{\mathcal{K}\setminus\mathcal{L}})
\mid \mathbb{Q}(\mathbf{q}^0_{\mathcal{K}\setminus\mathcal{L}},
\mathbf{q}^1_{\mathcal{K}\setminus\mathcal{L}},
\mathbf{q}_{\mathcal{L}}),B=0) 
\notag \\
& \hspace{10pt} 
+ \frac{1}{2} \Big\{ 1-
\Pr({\mathbf{\hat Q}}_{\mathcal{K}
\setminus\mathcal{L}}
\in\mathcal{N}_{\tau}
(\mathbf{Q}^0_{\mathcal{K}\setminus\mathcal{L}})
\mid \mathbb{Q}(\mathbf{q}^0_{\mathcal{K}\setminus\mathcal{L}},
\mathbf{q}^1_{\mathcal{K}\setminus\mathcal{L}},
\mathbf{q}_{\mathcal{L}}), 
\notag \\
& \hspace{75pt}  
 B=1) \Big\},
\label{equ:PBh=B}
\end{align}
where we have used the fact that $B$ and $[\mathbf{Q}^0_{\mathcal{K}\setminus\mathcal{L}},
\mathbf{Q}^1_{\mathcal{K}\setminus\mathcal{L}},\mathbf{Q}_{\mathcal{L}}]$
are independent. Putting~\eqref{equ:PBh=B} into 
\eqref{equ:tea-1} gives
\begin{align}
& 
\Pr({\mathbf{\hat Q}}_{\mathcal{K}
\setminus\mathcal{L}}
\in \mathcal{N}_{\tau}
(\mathbf{Q}^0_{\mathcal{K}\setminus\mathcal{L}})
\mid \mathbb{Q}(\mathbf{q}^0_{\mathcal{K}\setminus\mathcal{L}},
\mathbf{q}^1_{\mathcal{K}\setminus\mathcal{L}},
\mathbf{q}_{\mathcal{L}}),B=0)
\notag \\
& \leq
\Pr({\mathbf{\hat Q}}_{\mathcal{K} \setminus\mathcal{L}}
\in\mathcal{N}_{\tau}
(\mathbf{Q}^0_{\mathcal{K}\setminus\mathcal{L}})
\mid \mathbb{Q}(\mathbf{q}^0_{\mathcal{K}\setminus\mathcal{L}},
\mathbf{q}^1_{\mathcal{K}\setminus\mathcal{L}},
\mathbf{q}_{\mathcal{L}}),B=1)
\notag \\
& \hspace{10pt} + 
8(K-L-1)|\mathcal{X}|\cdot F_{\mathrm{CPA}}(n).
\label{equ:tea-2}
\end{align}

To simply notation, let
$\boldsymbol{\Upsilon} = [\mathbf{\Phi}^n_{\mathcal{K}},
\mathbf{\Psi}^n_{\mathcal{L}}, \mathbf{\bar
  C}_{\mathcal{K},\mathcal{K}}, \mathbf{G}_{\mathcal{L}}$,
$\mathbf{R}_{\mathcal{L},\mathcal{K}},
\mathbf{R}_{\mathcal{K}\setminus \mathcal{L},\mathcal{L}}]$.
Conditioned on $B=0$,
$\mathbf{\hat Q}_{\mathcal{K} \setminus\mathcal{L}}$ is a function of
$\boldsymbol{\Upsilon}$, $\mathbf{Q}_{\mathcal{L}}$, and
$\mathbf{G}^0_{\mathcal{K}\setminus \mathcal{L}}$. 
For any
$\mathbf{q}^0_{\mathcal{K}\setminus\mathcal{L}}$,
$\mathbf{q}^1_{\mathcal{K}\setminus\mathcal{L}}$, and
$\boldsymbol{\sigma}$ satisfying
$\boldsymbol{\Sigma}_{\mathbf{q}^0_{\mathcal{K}\setminus\mathcal{L}}}
= \boldsymbol{\Sigma}_{\mathbf{q}^1_{\mathcal{K}\setminus\mathcal{L}}}
= \boldsymbol{\sigma}$, we thus have
\begin{align}
& 
\Pr({\mathbf{\hat Q}}_{\mathcal{K} \setminus\mathcal{L}}
\in \mathcal{N}_{\tau}
(\mathbf{Q}^0_{\mathcal{K}\setminus\mathcal{L}})
\mid \mathbb{Q}(\mathbf{q}^0_{\mathcal{K}\setminus\mathcal{L}},
\mathbf{q}^1_{\mathcal{K}\setminus\mathcal{L}},
\mathbf{q}_{\mathcal{L}}), B=0)
\notag \\
& = 
\sum_{\boldsymbol{\upsilon}, \mathbf{g}_{\mathcal{K}\setminus\mathcal{L}}}
\sum_{\mathbf{q}'_{\mathcal{K}\setminus\mathcal{L}} \in  \mathcal{N}_{\tau}
(\mathbf{q}^0_{\mathcal{K}\setminus\mathcal{L}})}
\hspace{-15pt} p_{\mathbf{\hat Q}_{\mathcal{K} \setminus\mathcal{L}} \mid
\boldsymbol{\Upsilon}, \mathbf{Q}_{\mathcal{L}},
\mathbf{G}^0_{\mathcal{K}\setminus \mathcal{L}}}
(\mathbf{q}'_{\mathcal{K}\setminus\mathcal{L}} \mid
\boldsymbol{\upsilon},
\mathbf{q}_{\mathcal{L}},
\mathbf{g}_{\mathcal{K}\setminus\mathcal{L}}) \cdot 
\notag \\
& \hspace{30pt}
p_{\boldsymbol{\Upsilon}, 
  \mathbf{G}^0_{\mathcal{K}\setminus \mathcal{L}} \mid 
   \mathbf{Q}^0_{\mathcal{K}\setminus\mathcal{L}},
\mathbf{Q}^1_{\mathcal{K}\setminus\mathcal{L}},
\mathbf{Q}_{\mathcal{L}}} (\boldsymbol{\upsilon}, 
  \mathbf{g}_{\mathcal{K}\setminus \mathcal{L}} \mid 
   \mathbf{q}^0_{\mathcal{K}\setminus\mathcal{L}},
\mathbf{q}^1_{\mathcal{K}\setminus\mathcal{L}},
\mathbf{q}_{\mathcal{L}})
\notag \\
& = 
\sum_{\boldsymbol{\upsilon}, \mathbf{g}_{\mathcal{K}\setminus\mathcal{L}}}
\sum_{\mathbf{q}'_{\mathcal{K}\setminus\mathcal{L}} \in  \mathcal{N}_{\tau}
(\mathbf{q}^0_{\mathcal{K}\setminus\mathcal{L}})}
\hspace{-15pt} p_{\mathbf{\hat Q}_{\mathcal{K} \setminus\mathcal{L}} \mid
\boldsymbol{\Upsilon}, \mathbf{Q}_{\mathcal{L}},
\mathbf{G}^0_{\mathcal{K}\setminus \mathcal{L}}}
(\mathbf{q}'_{\mathcal{K}\setminus\mathcal{L}} \mid
\boldsymbol{\upsilon},
\mathbf{q}_{\mathcal{L}},
\mathbf{g}_{\mathcal{K}\setminus\mathcal{L}}) \cdot 
\notag \\
& \hspace{0pt}
p_{\boldsymbol{\Upsilon}, 
  \mathbf{G}^0_{\mathcal{K}\setminus \mathcal{L}} \mid 
   \mathbf{Q}^0_{\mathcal{K}\setminus\mathcal{L}},
   \mathbf{Q}^1_{\mathcal{K}\setminus\mathcal{L}},
\boldsymbol{\Sigma}_{\mathbf{Q}^0_{\mathcal{K}\setminus\mathcal{L}}},
\mathbf{Q}_{\mathcal{L}}} (\boldsymbol{\upsilon}, 
  \mathbf{g}_{\mathcal{K}\setminus \mathcal{L}} \mid 
   \mathbf{q}^0_{\mathcal{K}\setminus\mathcal{L}}, 
   \mathbf{q}^1_{\mathcal{K}\setminus\mathcal{L}}, 
\boldsymbol{\sigma},
\mathbf{q}_{\mathcal{L}})
\notag \\
& = 
\Pr({\mathbf{\hat Q}}_{\mathcal{K} \setminus\mathcal{L}}
\in \mathcal{N}_{\tau}
(\mathbf{q}^0_{\mathcal{K}\setminus\mathcal{L}})
\mid \mathbf{Q}_{\mathcal{K}\setminus\mathcal{L}} 
= \mathbf{q}^0_{\mathcal{K}\setminus\mathcal{L}},
\mathbf{Q}^1_{\mathcal{K}\setminus\mathcal{L}} 
= \mathbf{q}^1_{\mathcal{K}\setminus\mathcal{L}},
\notag \\
& \hspace{40pt}
\boldsymbol{\Sigma}_{\mathbf{Q}_{\mathcal{K}\setminus\mathcal{L}}} 
= \boldsymbol{\sigma}, 
\mathbf{Q}_{\mathcal{L}} = \mathbf{q}_{\mathcal{L}}),
\label{equ:QhatB0}
\end{align}
where the first equality results because $B$ is independent of
$[\boldsymbol{\Upsilon}, \mathbf{G}^0_{\mathcal{K}\setminus
  \mathcal{L}}, \mathbf{Q}^0_{\mathcal{K}\setminus\mathcal{L}},
\mathbf{Q}^1_{\mathcal{K}\setminus\mathcal{L}},
\mathbf{Q}_{\mathcal{L}}]$, the second equality is due to the fact
that
$\boldsymbol{\Sigma}_{\mathbf{Q}^0_{\mathcal{K}\setminus\mathcal{L}}}$
is a deterministic function of
$\mathbf{Q}^0_{\mathcal{K}\setminus\mathcal{L}}$, and the last
equality is simply re-identifying
$\mathbf{G}^0_{\mathcal{K}\setminus\mathcal{L}}$ as
$\mathbf{G}_{\mathcal{K}\setminus\mathcal{L}}$ and
$\mathbf{Q}^0_{\mathcal{K}\setminus\mathcal{L}}$ as
$\mathbf{Q}_{\mathcal{K}\setminus\mathcal{L}}$ to fit the description
in the TEA experiment because
$\mathbf{Q}^0_{\mathcal{K}\setminus\mathcal{L}}$
(resp. $\mathbf{G}^0_{\mathcal{K}\setminus\mathcal{L}}$) has the same
conditional distribution as that of
$\mathbf{Q}_{\mathcal{K}\setminus\mathcal{L}}$
(resp. $\mathbf{G}_{\mathcal{K}\setminus\mathcal{L}}$) given
$[\boldsymbol{\Sigma}_{\mathbf{Q}_{\mathcal{K}\setminus\mathcal{L}}},
\mathbf{Q}_{\mathcal{L}}]$. We will write
$\mathbf{Q}_{\mathcal{K}\setminus\mathcal{L}}$
(resp. $\mathbf{G}_{\mathcal{K}\setminus\mathcal{L}}$) instead of
$\mathbf{Q}^0_{\mathcal{K}\setminus\mathcal{L}}$
(resp. $\mathbf{G}^0_{\mathcal{K}\setminus\mathcal{L}}$) below for
matching the notation in Theorem~\ref{thm:tea}.

Conditioned on $B=1$,
$\mathbf{\hat Q}_{\mathcal{K} \setminus\mathcal{L}}$ is a function of
$\boldsymbol{\Upsilon}$, $\mathbf{Q}_{\mathcal{L}}$, and
$\mathbf{G}^1_{\mathcal{K}\setminus \mathcal{L}}$ instead. To
distinguish from
$\mathbf{\hat Q}_{\mathcal{K} \setminus\mathcal{L}} = \mathbf{\hat
  Q}_{\mathcal{K} \setminus\mathcal{L}}(\boldsymbol{\Upsilon},
\mathbf{Q}_{\mathcal{L}}, \mathbf{G}^0_{\mathcal{K}\setminus
  \mathcal{L}})$, let
$\mathbf{\hat Q}'_{\mathcal{K} \setminus\mathcal{L}} = \mathbf{\hat
  Q}_{\mathcal{K} \setminus\mathcal{L}}(\boldsymbol{\Upsilon},
\mathbf{Q}_{\mathcal{L}}, \mathbf{G}^1_{\mathcal{K}\setminus
  \mathcal{L}})$ in this case. A similar argument as above follows to
show that
\begin{align}
& 
\Pr({\mathbf{\hat Q}}_{\mathcal{K} \setminus\mathcal{L}}
\in \mathcal{N}_{\tau}
(\mathbf{Q}^0_{\mathcal{K}\setminus\mathcal{L}})
\mid \mathbb{Q}(\mathbf{q}^0_{\mathcal{K}\setminus\mathcal{L}},
\mathbf{q}^1_{\mathcal{K}\setminus\mathcal{L}},
\mathbf{q}_{\mathcal{L}}), B=1)
\notag \\
& = 
\Pr(\mathbf{\hat Q}'_{\mathcal{K} \setminus\mathcal{L}}
\in \mathcal{N}_{\tau}
(\mathbf{q}^0_{\mathcal{K}\setminus\mathcal{L}})
\mid \mathbf{Q}_{\mathcal{K}\setminus\mathcal{L}} 
= \mathbf{q}^0_{\mathcal{K}\setminus\mathcal{L}},
\mathbf{Q}^1_{\mathcal{K}\setminus\mathcal{L}} 
= \mathbf{q}^1_{\mathcal{K}\setminus\mathcal{L}},
\notag \\
& \hspace{40pt}
\boldsymbol{\Sigma}_{\mathbf{Q}_{\mathcal{K}\setminus\mathcal{L}}} 
= \boldsymbol{\sigma}, 
\mathbf{Q}_{\mathcal{L}} = \mathbf{q}_{\mathcal{L}}).
\label{equ:QhatB1}
\end{align}
Applying~\eqref{equ:QhatB0} and~\eqref{equ:QhatB1}
to~\eqref{equ:tea-2}, and then conditionally averaging with respect to
$[\mathbf{Q}_{\mathcal{K}\setminus\mathcal{L}},
\mathbf{Q}^1_{\mathcal{K}\setminus\mathcal{L}}]$
gives~\eqref{equ:edra}.

The conditional independence between
$\mathbf{\hat Q}'_{\mathcal{K} \setminus\mathcal{L}}$ and
$\mathbf{Q}_{\mathcal{K}\setminus\mathcal{L}}$ given
$[\boldsymbol{\Sigma}_{\mathbf{Q}_{\mathcal{K}\setminus\mathcal{L}}},
\mathbf{Q}_{\mathcal{L}}]$ follows from
\begin{align}
  & p_{\mathbf{\hat Q}'_{\mathcal{K} \setminus\mathcal{L}} \mid 
    \mathbf{Q}_{\mathcal{K}\setminus\mathcal{L}}, 
    \boldsymbol{\Sigma}_{\mathbf{Q}_{\mathcal{K}\setminus\mathcal{L}}},
    \mathbf{Q}_{\mathcal{L}}} (\mathbf{q}'_{\mathcal{K}\setminus\mathcal{L}} \mid
    \mathbf{q}_{\mathcal{K}\setminus\mathcal{L}}, \boldsymbol{\sigma}, 
    \mathbf{q}_{\mathcal{L}})
    \notag \\
  & = 
    \sum_{\boldsymbol{\upsilon}, \mathbf{g}_{\mathcal{K}\setminus\mathcal{L}}}
    \sum_{\mathbf{q}^1_{\mathcal{K}\setminus\mathcal{L}} :
    \boldsymbol{\Sigma}_{\mathbf{q}^1_{\mathcal{K}\setminus\mathcal{L}}}
    = \boldsymbol{\sigma}} \hspace{-15pt}
    p_{\mathbf{\hat Q}'_{\mathcal{K} \setminus\mathcal{L}} \mid
    \boldsymbol{\Upsilon}, \mathbf{Q}_{\mathcal{L}},
    \mathbf{G}^1_{\mathcal{K}\setminus \mathcal{L}}}
    (\mathbf{q}'_{\mathcal{K}\setminus\mathcal{L}} \mid
    \boldsymbol{\upsilon},
    \mathbf{q}_{\mathcal{L}},
    \mathbf{g}_{\mathcal{K}\setminus\mathcal{L}}) \cdot 
    \notag \\
  & \hspace{00pt}
    p_{\boldsymbol{\Upsilon}, 
    \mathbf{G}^1_{\mathcal{K}\setminus \mathcal{L}},
    \mathbf{Q}^1_{\mathcal{K}\setminus\mathcal{L}}\mid 
    \mathbf{Q}_{\mathcal{K}\setminus\mathcal{L}},
    \boldsymbol{\Sigma}_{\mathbf{Q}_{\mathcal{K}\setminus\mathcal{L}}},
    \mathbf{Q}_{\mathcal{L}}} (\boldsymbol{\upsilon}, 
    \mathbf{g}_{\mathcal{K}\setminus \mathcal{L}},
    \mathbf{q}^1_{\mathcal{K}\setminus\mathcal{L}}        \mid 
    \mathbf{q}_{\mathcal{K}\setminus\mathcal{L}}, 
    \boldsymbol{\sigma},
    \mathbf{q}_{\mathcal{L}})
    \notag \\
  & = 
    \sum_{\boldsymbol{\upsilon}, \mathbf{g}_{\mathcal{K}\setminus\mathcal{L}}}
    p_{\mathbf{\hat Q}'_{\mathcal{K} \setminus\mathcal{L}} \mid
    \boldsymbol{\Upsilon}, \mathbf{Q}_{\mathcal{L}},
    \mathbf{G}^1_{\mathcal{K}\setminus \mathcal{L}}}
    (\mathbf{q}'_{\mathcal{K}\setminus\mathcal{L}} \mid
    \boldsymbol{\upsilon},
    \mathbf{q}_{\mathcal{L}},
    \mathbf{g}_{\mathcal{K}\setminus\mathcal{L}}) \cdot 
    \notag \\
  & \hspace{0pt}
    \sum_{\mathbf{q}^1_{\mathcal{K}\setminus\mathcal{L}} :
    \boldsymbol{\Sigma}_{\mathbf{q}^1_{\mathcal{K}\setminus\mathcal{L}}}
    = \boldsymbol{\sigma}}
    \hspace{-15pt}  p_{\boldsymbol{\Upsilon}, 
    \mathbf{G}^1_{\mathcal{K}\setminus \mathcal{L}} \mid 
    \mathbf{Q}^1_{\mathcal{K}\setminus\mathcal{L}},
    \boldsymbol{\Sigma}_{\mathbf{Q}_{\mathcal{K}\setminus\mathcal{L}}},
    \mathbf{Q}_{\mathcal{L}}} (\boldsymbol{\upsilon}, 
    \mathbf{g}_{\mathcal{K}\setminus \mathcal{L}} \mid 
    \mathbf{q}^1_{\mathcal{K}\setminus\mathcal{L}}, \boldsymbol{\sigma},
    \mathbf{q}_{\mathcal{L}}) \cdot
    \notag \\
  & \hspace{30pt}
    p_{\mathbf{Q}_{\mathcal{K}\setminus\mathcal{L}} \mid
    \boldsymbol{\Sigma}_{\mathbf{Q}_{\mathcal{K}\setminus\mathcal{L}}},
    \mathbf{Q}_{\mathcal{L}}} (
    \mathbf{q}^1_{\mathcal{K}\setminus\mathcal{L}} \mid
    \boldsymbol{\sigma}, \mathbf{q}_{\mathcal{L}}) 
    \notag \\
  & =
    p_{\mathbf{\hat Q}'_{\mathcal{K} \setminus\mathcal{L}} \mid 
    \boldsymbol{\Sigma}_{\mathbf{Q}_{\mathcal{K}\setminus\mathcal{L}}},
    \mathbf{Q}_{\mathcal{L}}} (\mathbf{q}'_{\mathcal{K}\setminus\mathcal{L}} \mid
    \boldsymbol{\sigma}, 
    \mathbf{q}_{\mathcal{L}}),
\label{equ:Qhat'-indept-Q}
\end{align}
where the second equality results because
$\mathbf{Q}_{\mathcal{K}\setminus\mathcal{L}}$ and
$\mathbf{Q}^1_{\mathcal{K}\setminus\mathcal{L}}$ have thee same
conditional distribution and are conditionally
independent given
$[\boldsymbol{\Sigma}_{\mathbf{Q}_{\mathcal{K}\setminus\mathcal{L}}},
\mathbf{Q}_{\mathcal{L}}]$ (see step 2) of the constructed
experiment), and $\mathbf{Q}_{\mathcal{K}\setminus\mathcal{L}}$ and
$[\boldsymbol{\Upsilon},
\mathbf{G}^1_{\mathcal{K}\setminus\mathcal{L}}]$ are conditionally
independent given
$[\mathbf{Q}^1_{\mathcal{K}\setminus\mathcal{L}},
\mathbf{Q}_{\mathcal{K}\setminus\mathcal{L}},
\mathbf{Q}_{\mathcal{L}}]$, which in turn is due to that
$\mathbf{G}^1_{\mathcal{K}\setminus\mathcal{L}}
=\mathbf{Q}^1_{\mathcal{K}\setminus\mathcal{L}} \oplus
\boldsymbol{\Sigma}_{\mathbf{R}_{\mathcal{K},\mathcal{K}\setminus\mathcal{L}}}$
and that $\mathbf{R}_{\mathcal{L},\mathcal{K}}$ is a function in the
form of~\eqref{equ:RLK_fn}. 

Finally, note that
\begin{align}
  & p_{\mathbf{\hat Q}_{\mathcal{K} \setminus\mathcal{L}} \mid 
    \boldsymbol{\Sigma}_{\mathbf{Q}_{\mathcal{K}\setminus\mathcal{L}}},
    \mathbf{Q}_{\mathcal{L}}} (\mathbf{q}'_{\mathcal{K}\setminus\mathcal{L}} \mid
    \boldsymbol{\sigma}, 
    \mathbf{q}_{\mathcal{L}})
    \notag \\
  & = 
    \sum_{\boldsymbol{\upsilon}, \mathbf{g}_{\mathcal{K}\setminus\mathcal{L}}}
    p_{\mathbf{\hat Q}_{\mathcal{K} \setminus\mathcal{L}} \mid
    \boldsymbol{\Upsilon}, \mathbf{Q}_{\mathcal{L}},
    \mathbf{G}_{\mathcal{K}\setminus \mathcal{L}}}
    (\mathbf{q}'_{\mathcal{K}\setminus\mathcal{L}} \mid
    \boldsymbol{\upsilon},
    \mathbf{q}_{\mathcal{L}},
    \mathbf{g}_{\mathcal{K}\setminus\mathcal{L}}) \cdot 
    \notag \\
  & \hspace{0pt}
    \sum_{\mathbf{q}_{\mathcal{K}\setminus\mathcal{L}} :
    \boldsymbol{\Sigma}_{\mathbf{q}_{\mathcal{K}\setminus\mathcal{L}}}
    = \boldsymbol{\sigma}}
    \hspace{-15pt}  p_{\boldsymbol{\Upsilon}, 
    \mathbf{G}_{\mathcal{K}\setminus \mathcal{L}} \mid 
    \mathbf{Q}_{\mathcal{K}\setminus\mathcal{L}},
    \boldsymbol{\Sigma}_{\mathbf{Q}_{\mathcal{K}\setminus\mathcal{L}}},
    \mathbf{Q}_{\mathcal{L}}} (\boldsymbol{\upsilon}, 
    \mathbf{g}_{\mathcal{K}\setminus \mathcal{L}} \mid 
    \mathbf{q}_{\mathcal{K}\setminus\mathcal{L}}, \boldsymbol{\sigma},
    \mathbf{q}_{\mathcal{L}}) \cdot
    \notag \\
  & \hspace{30pt}
    p_{\mathbf{Q}_{\mathcal{K}\setminus\mathcal{L}} \mid
    \boldsymbol{\Sigma}_{\mathbf{Q}_{\mathcal{K}\setminus\mathcal{L}}},
    \mathbf{Q}_{\mathcal{L}}} (
    \mathbf{q}_{\mathcal{K}\setminus\mathcal{L}} \mid
    \boldsymbol{\sigma}, \mathbf{q}_{\mathcal{L}}) 
    \notag \\
  & =
    p_{\mathbf{\hat Q}'_{\mathcal{K} \setminus\mathcal{L}} \mid 
    \boldsymbol{\Sigma}_{\mathbf{Q}_{\mathcal{K}\setminus\mathcal{L}}},
    \mathbf{Q}_{\mathcal{L}}} (\mathbf{q}'_{\mathcal{K}\setminus\mathcal{L}} \mid
    \boldsymbol{\sigma}, 
    \mathbf{q}_{\mathcal{L}}),
\label{equ:pQhat=}
\end{align}
where the second equality results by comparing the expression in the
line above is the same as that in second equality line
of~\eqref{equ:Qhat'-indept-Q} due to the fact that
$p_{\boldsymbol{\Upsilon}, \mathbf{G}_{\mathcal{K}\setminus
    \mathcal{L}} \mid \mathbf{Q}_{\mathcal{K}\setminus\mathcal{L}},
  \boldsymbol{\Sigma}_{\mathbf{Q}_{\mathcal{K}\setminus\mathcal{L}}},
  \mathbf{Q}_{\mathcal{L}}} = p_{\boldsymbol{\Upsilon},
  \mathbf{G}^1_{\mathcal{K}\setminus \mathcal{L}} \mid
  \mathbf{Q}^1_{\mathcal{K}\setminus\mathcal{L}},
  \boldsymbol{\Sigma}_{\mathbf{Q}_{\mathcal{K}\setminus\mathcal{L}}},
  \mathbf{Q}_{\mathcal{L}}}$ and
$p_{\mathbf{\hat Q}_{\mathcal{K} \setminus\mathcal{L}} \mid
  \boldsymbol{\Upsilon}, \mathbf{Q}_{\mathcal{L}},
  \mathbf{G}_{\mathcal{K}\setminus \mathcal{L}}} = p_{\mathbf{\hat
    Q}'_{\mathcal{K} \setminus\mathcal{L}} \mid \boldsymbol{\Upsilon},
  \mathbf{Q}_{\mathcal{L}}, \mathbf{G}^1_{\mathcal{K}\setminus
    \mathcal{L}}}$ as $\mathbf{G}_{\mathcal{K}\setminus \mathcal{L}}$
(resp. $\mathbf{\hat Q}_{\mathcal{K} \setminus\mathcal{L}}$) and
$\mathbf{G}^1_{\mathcal{K}\setminus \mathcal{L}}$ (resp.
$\mathbf{\hat Q}'_{\mathcal{K} \setminus\mathcal{L}}$) are obtained
from the same function with
$\mathbf{Q}_{\mathcal{K}\setminus\mathcal{L}}$
(resp. $\mathbf{G}_{\mathcal{K}\setminus\mathcal{L}}$) and
$\mathbf{Q}^1_{\mathcal{K}\setminus\mathcal{L}}$
(resp. $\mathbf{G}^1_{\mathcal{K}\setminus\mathcal{L}}$) as the
respective input arguments.



\section{Conclusion}
In this paper, we develop a privacy-preserving event detection scheme
for the generalized $K$-sample problem. In the proposed scheme, the
marginal types of sensors' measurements are first obfuscated with ZMS
random numbers, and then sent to the fusion center for the calculation
of a decision statistic based on the Hellinger diameter measure, so
that the privacy of individual sensors' data can be protected. We present
analysis to show that the proposed detection scheme 1) is optimal in
the sense that it achieves the best type-I error exponent when the
type-II error rate is required to be negligible, and 2) is secure
against any PPT attacker in the sense that the probability advantage
of the attacker successfully estimating the sensors' measured type
over independent guessing is negligible. The combination of these two
results implies that the additional requirement of privacy protection
does not fundamentally require any tradeoff in achieving the optimal
type-I error exponent in the generalized $K$-sample problem.

\appendices

\section{Proof of Theorem~\ref{thm:errexp}} \label{app:EXPproof}

In the proof below, we assume that the diameter measure $d(\cdot)$ is
bounded by a positive constant $d_{\max}$. Thus, we have
$0 \leq d_0 < d_1 \leq d_{\max}$. Note that this assumption is not
restrictive because there are simply more edge conditions to check
when $d(\cdot)$ is bounded. If $d(\cdot)$ is not bounded, one may
simply regard $d_{\max} = \infty$ and $\alpha_*(d_{\max}) = \infty$,
and make appropriate changes to the respective edge conditions below.

\subsubsection*{Useful properties}
We start by noting a number of properties of the functions involved in
Theorem~\ref{thm:errexp}. We will use these properties in proving the
various parts of the theorem below.

First, both $\Delta_0(\cdot)$ and $\Delta_1(\cdot)$ are bounded
continuous functions in $\mathcal{P}(\mathcal{X}^K)$ due to the
continuity of the KL divergence. In addition, both have positive
maximum values since $0 \leq d_0 < d_1$.  Next, both
$\alpha_*(\cdot)$ and $\gamma_*(\cdot)$ are clearly
non-decreasing. It is easy to see that $\alpha_*(\gamma) = 0$
for $0 \leq \gamma \leq d_0$, and that
$0< \alpha_*(d_1) \leq \alpha_*(d_{\max})$ as
$0 \leq d_0 < d_1 \leq d_{\max}$.
Because of the continuity of the functions $\Delta_0(\cdot)$ and
$d(\cdot)$ in $\mathcal{P}(\mathcal{X}^K)$, it is also easy to check that
$\alpha_*(\cdot)$ is right-continuous on $[0,d_{\max})$ and is
left-continuous on $(0,d_{\max}]$.
Similarly, $\beta^*(\cdot)$ is non-increasing and is continuous on
$(0, \infty)$.

We note that $\alpha_*(\gamma_*(\alpha)) \geq \alpha$ for
$\alpha \in [0,\alpha_* (d_{\max})]$. Indeed, as a consequence of its
right-continuity on $[0,d_{\max})$,
$\alpha_*(\gamma_*(\alpha)) \geq \alpha$ as long as
$\gamma_*(\alpha) < d_{\max}$. On the other hand, if
$\gamma_*(\alpha) = d_{\max}$, we have
$\alpha_*(\gamma_*(\alpha)) = \alpha_*(d_{\max}) \geq \alpha$
trivially.

In addition, for any $\gamma \in [0,d_{\max}]$, we have
$\gamma_*(\alpha) \geq \gamma$ if $\alpha > \alpha_*(\gamma)$.
Indeed, if $\alpha \in (\alpha_*(\gamma), \alpha_*(d_{\max})]$, we
have $\alpha_*(\gamma_*(\alpha)) \geq \alpha > \alpha_*(\gamma)$,
which in turn gives $\gamma_*(\alpha) \geq \gamma$ as
$\alpha_*(\cdot)$ is non-decreasing.  On the other hand, if
$\alpha > \alpha_*(d_{\max})$, then
$\gamma_*(\alpha) = d_{\max} \geq \gamma$ trivially.

\subsubsection*{Proof of (i)}
For $\alpha \in [0,\alpha_*(d_{\max})]$,
$\alpha_*(\gamma_*(\alpha)) \geq \alpha$ implies that
$\Delta_0(\mathbf{p}_{\mathcal{K}}) \geq \alpha$ if
$d(\mathbf{p}_{\mathcal{K}}) \geq \gamma_*(\alpha)$, which is
equivalent to that $d(\mathbf{p}_{\mathcal{K}}) < \gamma_*(\alpha)$ if
$\Delta_0(\mathbf{p}_{\mathcal{K}}) < \alpha$. This latter assertion
gives $\beta_*(\alpha) \leq \beta^*(\alpha)$.  On the other hand, for
$\alpha > \alpha_*(d_{\max})$, $\gamma_*(\alpha) = d_{\max}$, which
implies $\beta_*(\alpha)=0$.  Hence, we have
$\beta_*(\alpha) \leq \beta^*(\alpha)$ trivially.

\subsubsection*{Proof of (ii)}
We have $\beta^*(\alpha) > 0$ if $\beta_*(\alpha) >0$ from (i). We
need to show the other direction of implication. Suppose
$\beta^*(\alpha) > 0$. Then, there must exist an $\eta > 0$ such that
$d(\mathbf{p}_{\mathcal{K}}) \leq d_1 - \eta$ whenever
$\Delta_0(\mathbf{p}_{\mathcal{K}}) < \alpha$. This implies
$\gamma_*(\alpha) \leq d_1 - \eta$; otherwise there would be a
$\gamma \in (d_1-\eta, \gamma_*(\alpha) )$ and a
$\mathbf{p}_{\mathcal{K}}$ satisfying
$\Delta_0(\mathbf{p}_{\mathcal{K}}) < \alpha$ and
$d(\mathbf{p}_{\mathcal{K}}) \geq \gamma > d_1-\eta$. But
$\gamma_*(\alpha) \leq d_1 - \eta$ then forces $\beta_*(\alpha)>0$.

\subsubsection*{Proof of (iii)}
First, note that $\beta^*(0)=\infty$. Moreover, $\beta_*(\alpha) = 0$
if $\gamma_*(\alpha) \geq d_1$. From (ii), the continuity of
$\beta^*(\cdot)$, and the fact that $\gamma_*(\alpha) \geq d_1$ if
$\alpha > \alpha_*(d_1)$, we then have $\beta^*(\alpha) = 0$ if
$\alpha \geq \alpha_*(d_1)$.

Consider the Hoeffding test:
\begin{equation}\label{equ:besthypo} 
\begin{aligned} 
H_{0} &: \theta=0\quad \quad \text{if} \
  \Delta_0({\mathbf{\tilde Q}}_{\mathcal{K}}) < \gamma \\
H_{1} &: \theta=1\quad \quad \text{if} \
  \Delta_0({\mathbf{\tilde Q}}_{\mathcal{K}}) \geq \gamma 
\end{aligned} 
\end{equation} 
where
$\mathbf{\tilde Q}_{\mathcal{K}} \in
\tilde{\mathcal{Q}}_t(\mathcal{X}^K)$ is the joint type of all sensor
measurements, and $\gamma \geq 0$ is a detection threshold.  Base on
this test, we prove the achievability of $(\alpha, \beta^*(\alpha))$.
By setting the threshold $\gamma$ in the test~\eqref{equ:besthypo} to
$0$, we have $\mathcal{R}_t = \emptyset$, which gives $\mu_t=1$ and
$\lambda_t = 0$. Hence, $(0, \infty)$, i.e., $(0, \beta^*(0))$, is
achievable. On the other hand, by setting
$\gamma > \max_{\mathbf{p}_{\mathcal{K}}}
\Delta_0(\mathbf{p}_{\mathcal{K}})$ in~\eqref{equ:besthypo}, we have
$\mathcal{R}_t = \mathcal{X}^{Kt}$, which gives $\mu_t=0$ and
$\lambda_t = 1$. As a result, $(\infty, 0)$, and hence
$(\alpha, \beta^*(\alpha))$ for all $\alpha \geq \alpha_*(d_1)$, are
also achievable.

It remains to show the achievability of $(\alpha, \beta^*(\alpha))$
for $\alpha \in (0,\alpha_*(d_1))$.  To that end, set the threshold
$\gamma = \alpha$ in the test~\eqref{equ:besthypo}. Then,
$\mathcal{R}_t =\{\mathbf{x}_{\mathcal{K}}\in \mathcal{X}^{Kt}:
\Delta_0(\tilde{\mathbf{q}}_{\mathbf{x}_{\mathcal{K}}}) < \alpha\}$ is
the acceptance region. By Sanov's theorem (see
\cite[Theorem~11.4.1]{info}), we have
  \begin{align*}
  & \mu_t
  =  \max_{\mathbf{p}_{0,\mathcal{K}}\in\mathcal{P}_{0, \mathcal{K}}} \mathbf{p}_{0,\mathcal{K}}
    (\mathcal{R}_t^c)
    \notag \\
  &\leq
    \max_{\mathbf{p}_{0,\mathcal{K}}\in\mathcal{P}_{0, \mathcal{K}}} (t+1)^{|\mathcal{X}|^K}
    2^{-t \min_{\mathbf{\tilde q}_{\mathcal{K}} \in
    \tilde{\mathcal{Q}}_t(\mathcal{X}^K):
    \Delta_0(\mathbf{\tilde q}_{\mathcal{K}}) \geq \alpha}
    D(\mathbf{\tilde q}_{\mathcal{K}}\| \mathbf{p}_{0,\mathcal{K}})}
    \notag \\
  &\leq
    (t+1)^{|\mathcal{X}|^K} \cdot2^{-t \alpha},
\end{align*}
which leads to
$\liminf_{t\rightarrow\infty}-\frac{1}{t}\log_2 \mu_t \geq \alpha$.

Similarly,
\begin{align*}
&\lambda_t
=
\max_{\mathbf{p}_{1,\mathcal{K}}\in\mathcal{P}_{1,\mathcal{K}}}
\mathbf{p}_{1,\mathcal{K}} (\mathcal{R}_t)
\notag \\
&\leq 
\max_{\mathbf{p}_{1,\mathcal{K}}\in\mathcal{P}_{1,\mathcal{K}}}
(t+1)^{|\mathcal{X}|^K}
2^{-t \min_{\mathbf{\tilde q}_{\mathcal{K}} \in
    \tilde{\mathcal{Q}}_t(\mathcal{X}^K):
    \Delta_0(\mathbf{\tilde q}_{\mathcal{K}}) < \alpha}
D(\mathbf{\tilde q}_{\mathcal{K}}\| \mathbf{p}_{1,\mathcal{K}})}
\notag \\
& \leq
(t+1)^{|\mathcal{X}|^K} \cdot2^{-t \beta^*(\alpha)},
\end{align*}
which leads to
$\liminf_{t\rightarrow\infty}-\frac{1}{t}\log_2 \lambda_t \geq
\beta^*(\alpha)$. 

Write $S(\alpha)=\sup\{\beta: (\alpha,\beta) \text{~is
  achievable}\}$. Then, the achievability of
$(\alpha, \beta^*(\alpha))$ implies that
$S(\alpha) \geq \beta^*(\alpha)$. We will show below
$S(\alpha) \leq \beta^*(\alpha)$. If $\beta^*(\alpha) = \infty$, there
is nothing to show. Hence, it suffices to consider the case of
$\beta^*(\alpha) < \infty$, which implies $\alpha>0$. Thus, we may
assume both these restrictions below.

Let $\mathcal{R}_t$ be the acceptance region giving
$\liminf_{t\rightarrow\infty}-\frac{1}{t}\log_2 \mu_t \geq \alpha$.
For any $\epsilon > 0$,
\begin{align}
2^{-t(\alpha-\epsilon)} 
&\geq \mu_t 
= \max_{\mathbf{p}_{0,\mathcal{K}}\in\mathcal{P}_{0,\mathcal{K}}} 
  \mathbf{p}_{0,\mathcal{K}}(\mathcal{R}^c_t)
\notag \\
&
= \max_{\mathbf{p}_{0,\mathcal{K}}\in\mathcal{P}_{0,\mathcal{K}}} 
\sum_{{\mathbf{\tilde q}}_{\mathcal{K}}\in\tilde{\mathcal{Q}}^K_t(\mathcal{X})}
\mathbf{p}_{0,\mathcal{K}}(\mathcal{R}^c_t\cap
\mathcal{T}({\mathbf{\tilde q}}_{\mathcal{K}}))
\notag \\
& \geq 
\max_{\mathbf{p}_{0,\mathcal{K}}\in\mathcal{P}_{0,\mathcal{K}}} 
|\mathcal{R}^c_t\cap\mathcal{T}({\mathbf{\tilde q}}_{\mathcal{K}})| \cdot
2^{-t(H({\mathbf{\tilde q}}_{\mathcal{K}})+ D(\tilde{\mathbf{q}}_{\mathcal{K}}\|\mathbf{p}_{0,\mathcal{K}}))}
\notag \\
& = 
|\mathcal{R}^c_t\cap\mathcal{T}({\mathbf{\tilde q}}_{\mathcal{K}})| \cdot
2^{-t(H({\mathbf{\tilde q}}_{\mathcal{K}})+
            \Delta_0(\tilde{\mathbf{q}}_{\mathcal{K}}))}
\label{equ:cardinality2}
\end{align}
for all
$\tilde{\mathbf{q}}_{\mathcal{K}} \in
\tilde{\mathcal{Q}}_t(\mathcal{X}^K)$, whenever $t$ is sufficiently
large. In~\eqref{equ:cardinality2},
$H({\mathbf{\tilde q}}_{\mathcal{K}})$ is the entropy of
${\mathbf{\tilde q}}_{\mathcal{K}}$, and the second inequality is due to
\cite[Theorem~11.1.2]{info}.

On the other hand, write
\[
  \beta^*_t (\alpha) = \min_{\mathbf{\tilde
      q}_{\mathcal{K}}\in\tilde{\mathcal{Q}}_t(\mathcal{X}^K):
    \Delta_0(\mathbf{\tilde q}_{\mathcal{K}})<\alpha}
  \Delta_1({\mathbf{\tilde q}}_{\mathcal{K}}).
\]
Then, for any small enough $\epsilon > 0$, we have
\begin{align}
\lambda_t
&=
\max_{\mathbf{p}_{1,\mathcal{K}}\in\mathcal{P}_{1,\mathcal{K}}}
\sum_{{\mathbf{\tilde q}}_{\mathcal{K}}\in\tilde{\mathcal{Q}}_t(\mathcal{X}^K)}
\mathbf{p}_{1,\mathcal{K}}(\mathcal{R}_t\cap
\mathcal{T}({\mathbf{\tilde q}}_{\mathcal{K}}))
\notag\\
&\geq \hspace{-20pt}
\max_{\substack{
\mathbf{p}_{1,\mathcal{K}}\in\mathcal{P}_{1,\mathcal{K}} \\
\mathbf{\tilde q}_{\mathcal{K}}\in\tilde{\mathcal{Q}}_t(\mathcal{X}^K):
\Delta_0(\mathbf{\tilde q}_{\mathcal{K}})<\alpha-2\epsilon}
} \hspace{-25pt}
|\mathcal{R}_t\cap\mathcal{T}({\mathbf{\tilde q}}_{\mathcal{K}})|
\cdot 2^{-t(H({\mathbf{\tilde q}}_{\mathcal{K}})+
D({\mathbf{\tilde q}}_{\mathcal{K}}\|\mathbf{p}_{1,\mathcal{K}}))}
\notag \\
&=  \hspace{-20pt}
\max_{
\mathbf{\tilde q}_{\mathcal{K}}\in\tilde{\mathcal{Q}}_t(\mathcal{X}^K):
\Delta_0(\mathbf{\tilde q}_{\mathcal{K}})<\alpha-2\epsilon}
\hspace{-0pt}
\begin{aligned}[t]
&\left(|\mathcal{T}({\mathbf{\tilde
  q}}_{\mathcal{K}})|  - |\mathcal{R}^c_t\cap\mathcal{T}({\mathbf{\tilde
  q}}_{\mathcal{K}})| \right)
\cdot 
\\
&\hspace{30pt}
2^{-t(H({\mathbf{\tilde q}}_{\mathcal{K}})+
\Delta_1({\mathbf{\tilde q}}_{\mathcal{K}}))}
\end{aligned}
\notag \\            
& \geq
\hspace{-20pt}
\max_{
\mathbf{\tilde q}_{\mathcal{K}}\in\tilde{\mathcal{Q}}_t(\mathcal{X}^K):
\Delta(\mathbf{\tilde q}_{\mathcal{K}})<\alpha-2\epsilon}
\hspace{-35pt}
\begin{aligned}[t]
&\big((t+1)^{-|\mathcal{X}|^K}
-2^{t(\Delta_0({\mathbf{\tilde
      q}}_{\mathcal{K}})-\alpha+\epsilon)}\big)
\cdot
2^{-t \Delta_1({\mathbf{\tilde q}}_{\mathcal{K}})} 
\end{aligned}
\notag \\
& \geq
\hspace{-20pt}
\max_{
\mathbf{\tilde q}_{\mathcal{K}}\in\tilde{\mathcal{Q}}_t(\mathcal{X}^K):
\Delta_0(\mathbf{\tilde q}_{\mathcal{K}})<\alpha-2\epsilon}
\hspace{-20pt}
\big((t+1)^{-|\mathcal{X}|^K} - 2^{-t\epsilon} \big) \cdot
2^{-t \Delta_1({\mathbf{\tilde q}}_{\mathcal{K}})} 
\notag \\
& \geq
\frac{1}{2} (t+1)^{-|\mathcal{X}|^K} \cdot
2^{-t \beta^*_t (\alpha - 2\epsilon)},
\label{equ:lambdabound}
\end{align}
whenever $t$ is sufficiently large, where the second inequality is due
to~\eqref{equ:cardinality2} and~\cite[Theorem~11.1.3]{info}. 

Because $\Delta_0(\cdot)$ and $\Delta_1(\cdot) $ are continuous in
$\mathcal{P}(\mathcal{X}^K)$,
$\bigcup_{t\geq 1}\tilde{\mathcal{Q}}_t(\mathcal{X}^K)$ is dense in
$\mathcal{P}(\mathcal{X}^K)$ and $\beta^*(\cdot)$ is continuous,
we have for every $\eta>0$,
\begin{equation}\label{equ:beta'_t2}
\beta_t^*(\alpha-2\epsilon)\leq \beta^*(\alpha-2\epsilon)+\eta \leq
\beta^*(\alpha) + 2\eta,
\end{equation}
whenever $t$ is sufficiently large and $\epsilon >0$ is sufficiently
small.  Putting~\eqref{equ:beta'_t2} back
into~\eqref{equ:lambdabound}, we get 
\begin{equation*}
-\frac{1}{t}\log \lambda_t\leq
\frac{1}{t}\left(|\mathcal{X}|^K \log(t+1)+1\right)+\beta^*(\alpha)+2\eta,
\end{equation*}
which implies $S(\alpha) \leq \beta^*(\alpha)$ by letting $\eta
\rightarrow 0$.

\subsubsection*{Proof of (iv)}
First, for any $\gamma < d_1$ and $\alpha \leq \alpha_*(\gamma)$, we
clearly have $\gamma_*(\alpha) \leq \gamma < d_1$ and thus
$\beta_*(\alpha) > 0$, which also implies $\beta^*(\alpha) > 0$ from
(ii).  From (iii), $(\alpha, \beta^*(\alpha))$ is an achievable pair
for every $\alpha \geq 0$. Hence, for any $\gamma < d_1$ and
$\alpha \leq \alpha_*(\gamma)$, $\alpha$ is also achievable. Write
$s_* = \sup\{ \alpha: \alpha \text{~is achievable}\}$. Then, the
continuity and non-decreasing nature of $\alpha_*(\cdot)$ give
$s_* \geq \alpha_*(d_1)$. It remains to prove $s_* \leq \alpha_*(d_1)$
below.

Let $\mathcal{R}_t$ be the acceptance region giving
$\lim_{t\rightarrow \infty} \lambda_t = 0$, which implies
$\lim_{t\rightarrow \infty} \mathbf{p}_{1,\mathcal{K}}(\mathcal{R}_t)
= 0$ for every
$\mathbf{p}_{1,\mathcal{K}}\in\mathcal{P}_{1,\mathcal{K}}$.  Then,
for each $\epsilon > 0$, whenever $t$ is sufficiently large, we have
\begin{align}
&1-\epsilon 
\leq 
\mathbf{p}_{1,\mathcal{K}} (\mathcal{R}^c_t)
\notag \\
& = \hspace{-10pt}
\sum_{\mathbf{\tilde q}_{\mathcal{K}} \in \tilde{\mathcal{Q}}_t(\mathcal{X}^K):
            D(\mathbf{\tilde q}_{\mathcal{K}} \|
            \mathbf{p}_{1,\mathcal{K}}) \geq \epsilon}
\hspace{-20pt}
|\mathcal{R}^c_t\cap \mathcal{T}(\mathbf{\tilde q}_{\mathcal{K}})| \cdot
2^{-t(H(\mathbf{\tilde q}_{\mathcal{K}} ) + D(\mathbf{\tilde
  q}_{\mathcal{K}}  \| \mathbf{p}_{1,\mathcal{K}}))}
\notag \\
& \hspace{10pt} + \hspace{-18pt}
\sum_{\mathbf{\tilde q}_{\mathcal{K}} \in \tilde{\mathcal{Q}}_t(\mathcal{X}^K):
            D(\mathbf{\tilde q}_{\mathcal{K}} \|
            \mathbf{p}_{1,\mathcal{K}}) < \epsilon}
\hspace{-20pt}
|\mathcal{R}^c_t\cap \mathcal{T}(\mathbf{\tilde q}_{\mathcal{K}})| \cdot
2^{-t(H(\mathbf{\tilde q}_{\mathcal{K}} ) + D(\mathbf{\tilde
  q}_{\mathcal{K}}  \| \mathbf{p}_{1,\mathcal{K}}))}
\notag \\
& \leq
\hspace{-10pt}
\sum_{\mathbf{\tilde q}_{\mathcal{K}} \in \tilde{\mathcal{Q}}_t(\mathcal{X}^K):
            D(\mathbf{\tilde q}_{\mathcal{K}} \|
            \mathbf{p}_{1,\mathcal{K}}) \geq \epsilon}
\hspace{-20pt}
|\mathcal{T}(\mathbf{\tilde q}_{\mathcal{K}})| \cdot
2^{-t(H(\mathbf{\tilde q}_{\mathcal{K}} ) + \epsilon)}
\notag \\
& \hspace{10pt} + \hspace{-18pt}
\sum_{\mathbf{\tilde q}_{\mathcal{K}} \in \tilde{\mathcal{Q}}_t(\mathcal{X}^K):
            D(\mathbf{\tilde q}_{\mathcal{K}} \|
            \mathbf{p}_{1,\mathcal{K}}) < \epsilon}
\hspace{-20pt}
|\mathcal{R}^c_t\cap \mathcal{T}(\mathbf{\tilde q}_{\mathcal{K}})| \cdot
2^{-tH(\mathbf{\tilde q}_{\mathcal{K}} )} 
\notag \\
& \leq 
(t+1)^{|\mathcal{X}|^K} 2^{-t\epsilon}
+ \hspace{-25pt}
\sum_{\mathbf{\tilde q}_{\mathcal{K}} \in \tilde{\mathcal{Q}}_t(\mathcal{X}^K):
            D(\mathbf{\tilde q}_{\mathcal{K}} \|
            \mathbf{p}_{1,\mathcal{K}}) < \epsilon}
\hspace{-25pt}
|\mathcal{R}^c_t\cap \mathcal{T}(\mathbf{\tilde q}_{\mathcal{K}})|
            \cdot 2^{-tH(\mathbf{\tilde q}_{\mathcal{K}} )},
\label{equ:1-lambda}
\end{align}
where the equality results from~\cite[Theorem~11.1.2]{info} and the
last inequality is the consequence of~\cite[Theorems~11.1.1
and~11.1.3]{info}.  From~\eqref{equ:1-lambda}, for each $\epsilon >0$,
whenever $t$ is sufficiently large,
\begin{align}
\mu_t & \geq 
  \mathbf{p}_{0,\mathcal{K}}(\mathcal{R}^c_t) 
\notag \\
& \geq
\hspace{-10pt}
\sum_{\mathbf{\tilde q}_{\mathcal{K}} \in \tilde{\mathcal{Q}}_t(\mathcal{X}^K):
            D(\mathbf{\tilde q}_{\mathcal{K}} \|
            \mathbf{p}_{1,\mathcal{K}}) < \epsilon}
\hspace{-30pt}
|\mathcal{R}^c_t\cap \mathcal{T}(\mathbf{\tilde q}_{\mathcal{K}})| \cdot
2^{-t(H(\mathbf{\tilde q}_{\mathcal{K}} ) + D(\mathbf{\tilde
  q}_{\mathcal{K}}  \| \mathbf{p}_{0,\mathcal{K}}))}
\notag \\
& \geq 
2^{-t\max_{\mathbf{\tilde q}_{\mathcal{K}} \in \tilde{\mathcal{Q}}_t(\mathcal{X}^K):
            D(\mathbf{\tilde q}_{\mathcal{K}} \| \mathbf{p}_{1,\mathcal{K}}) < \epsilon} 
            D(\mathbf{\tilde q}_{\mathcal{K}}  \| \mathbf{p}_{0,\mathcal{K}})}
\notag \\
& \hspace{20pt} \cdot
\left( 1 - \epsilon - (t+1)^{|\mathcal{X}|^K} 2^{-t\epsilon} \right)
\notag \\
& \geq
2^{-t\sup_{\mathbf{p}_{\mathcal{K}} \in \mathcal{P}(\mathcal{X}^K):
            D(\mathbf{p}_{\mathcal{K}} \| \mathbf{p}_{1,\mathcal{K}}) < \epsilon} 
            D(\mathbf{p}_{\mathcal{K}}  \| \mathbf{p}_{0,\mathcal{K}})}
\notag \\
& \hspace{20pt} \cdot
\left( 1 - \epsilon - (t+1)^{|\mathcal{X}|^K} 2^{-t\epsilon} \right),
\label{equ:mu>}
\end{align}
for every $\mathbf{p}_{0,\mathcal{K}}\in\mathcal{P}_{0, \mathcal{K}}$
and $\mathbf{p}_{1,\mathcal{K}}\in\mathcal{P}_{1,\mathcal{K}}$.
Further, since
\[
\lim_{\epsilon \rightarrow 0} \sup_{\mathbf{p}_{\mathcal{K}} \in
    \mathcal{P}(\mathcal{X}^K): D(\mathbf{p}_{\mathcal{K}} \|
    \mathbf{p}_{1,\mathcal{K}}) < \epsilon} D(\mathbf{p}_{\mathcal{K}}
  \| \mathbf{p}_{0,\mathcal{K}})
=
D(\mathbf{p}_{1,\mathcal{K}} \| \mathbf{p}_{0,\mathcal{K}}),
\]
we have
$\liminf_{t\rightarrow\infty}-\frac{1}{t}\log_2 \mu_t \leq
D(\mathbf{p}_{1,\mathcal{K}} \| \mathbf{p}_{0,\mathcal{K}})$ for every
$\mathbf{p}_{0,\mathcal{K}}\in\mathcal{P}_{0,\mathcal{K}}$ and
$\mathbf{p}_{1,\mathcal{K}}\in\mathcal{P}_{1,\mathcal{K}}$
from~\eqref{equ:mu>}.  As a result,
\[
\liminf_{t\rightarrow\infty}-\frac{1}{t}\log_2 \mu_t 
\leq \min_{\substack{\mathbf{p}_{0,\mathcal{K}}\in\mathcal{P}_0^K \\
  \mathbf{p}_{1,\mathcal{K}}\in\mathcal{P}_1^K}}
D(\mathbf{p}_{1,\mathcal{K}} \| \mathbf{p}_{0,\mathcal{K}}) =
\alpha_*(d_1), 
\]
which implies $s_* \leq \alpha_*(d_1)$.

\subsubsection*{Proof of (v)}
As in the proof of (iii) above, we have $\beta_*(0) = \infty$ and
$\beta_*(\alpha) = 0$ if $\alpha > \alpha_*(d_1)$.  By setting the
threshold $\gamma$ in the test~\eqref{equ:hypo} to $0$ and to any
value strictly larger than $d_{\max}$, we have $(0,\beta_*(0))$ and
$(\alpha, \beta_*(\alpha))$ for all $\alpha > \alpha_*(d_1)$
achievable using the test~\eqref{equ:hypo}. 

It remains to show the achievability of $(\alpha, \beta_*(\alpha))$
for $\alpha \in (0,\alpha_*(d_1)]$. To that end, set the threshold
$\gamma$ in the test~\eqref{equ:hypo} to $\gamma_*(\alpha)$. Then,
$\mathcal{R}_t =\{\mathbf{x}_{\mathcal{K}}\in \mathcal{X}^{Kt}:
d(\tilde{\mathbf{q}}_{\mathbf{x}_{\mathcal{K}}}) < \gamma_*(\alpha)\}$
is the acceptance region.  By Sanov's theorem again, 
\begin{align*}
  & \mu_t
  =
    \max_{\mathbf{p}_{0,\mathcal{K}}\in\mathcal{P}_{0,\mathcal{K}}} \mathbf{p}_{0,\mathcal{K}}
    (\mathcal{R}_t^c)
    \notag \\
  &\leq
    \max_{\mathbf{p}_{0,\mathcal{K}}\in\mathcal{P}_{0,\mathcal{K}}} (t+1)^{|\mathcal{X}|^K}
    2^{-t \min_{\mathbf{\tilde q}_{\mathcal{K}} \in
    \tilde{\mathcal{Q}}_t(\mathcal{X}^K):
    d(\mathbf{\tilde q}_{\mathcal{K}}) \geq \gamma_*(\alpha)}
    D(\mathbf{\tilde q}_{\mathcal{K}}\| \mathbf{p}_{0,\mathcal{K}})}
    \notag \\
  & \leq
    (t+1)^{|\mathcal{X}|^K} \cdot
    2^{-t\alpha_*(\gamma_*(\alpha))}
    \notag \\
  &\leq
    (t+1)^{|\mathcal{X}|^K} \cdot2^{-t \alpha},
\end{align*}
which leads to
$\liminf_{t\rightarrow\infty}-\frac{1}{t}\log_2 \mu_t \geq \alpha$.
Similarly,
\begin{align*}
&\lambda_t
=
\max_{\mathbf{p}_{1,\mathcal{K}}\in\mathcal{P}_{1,\mathcal{K}}}
\mathbf{p}_{1,\mathcal{K}} (\mathcal{R}_t)
\notag \\
&\leq 
\max_{\mathbf{p}_{1,\mathcal{K}}\in\mathcal{P}_{1,\mathcal{K}}}
(t+1)^{|\mathcal{X}|^K}
2^{-t \min_{\mathbf{\tilde q}_{\mathcal{K}} \in
    \tilde{\mathcal{Q}}_t(\mathcal{X}^K):
    d(\mathbf{\tilde q}_{\mathcal{K}}) < \gamma_*(\alpha)}
D(\mathbf{\tilde q}_{\mathcal{K}}\| \mathbf{p}_{1,\mathcal{K}})}
\notag \\
& \leq
(t+1)^{|\mathcal{X}|^K} \cdot2^{-t \beta_*(\alpha)},
\end{align*}
which leads to
$\liminf_{t\rightarrow\infty}-\frac{1}{t}\log_2 \lambda_t \geq
\beta_*(\alpha)$. 

As shown in the proof of (iv) above, $\beta_*(\alpha) > 0$ for any
$\gamma < d_1$ and $\alpha \leq \alpha_*(\gamma)$. Thus, the
achievability of $(\alpha, \beta_*(\alpha))$ by the
test~\eqref{equ:hypo} and the continuity of $\alpha_*(\cdot)$ implies
that
$\sup\{ \alpha \text{~achieved by the test~\eqref{equ:hypo}}\} =
\alpha_*(d_1)$.

\section{Proofs of Useful Lemmas} \label{sec:prooflemmas}
In this appendix, we give the proofs of Lemmas~\ref{lem:useful2}
and~\ref{lem:useful3}, The proof of Lemma~\ref{lem:useful4} is trivial
and is omitted.

\subsection{Proof of Lemma~\ref{lem:useful2}}

Let
$\mathcal{J}={\mathcal{K}} \setminus (\mathcal{I} \cup
\mathcal{L} )$. Then for any
$\boldsymbol{\sigma}_{\mathcal{K}\setminus\mathcal{L}} \in
\mathcal{N}_m^{(K-L)|\mathcal{X}|}$,
\begin{align}
&
p_{\boldsymbol{\Sigma}_{\mathbf{R}_{\mathcal{I},\mathcal{K}\setminus\mathcal{L}}}
\mid \mathbf{R}_{\mathcal{I},\mathcal{L}}}
(\boldsymbol{\sigma}_{\mathcal{K}\setminus\mathcal{L}} \mid
\mathbf{r}_{\mathcal{I},\mathcal{L}})
\notag \\
&=
\sum_{\mathbf{r}_{\mathcal{I},\mathcal{J}}\in\mathcal{N}_m^{2(K-L-2)|\mathcal{X}|}}
p_{\boldsymbol{\Sigma}_{\mathbf{R}_{\mathcal{I},\mathcal{K}\setminus\mathcal{L}}}
\mid \mathbf{R}_{\mathcal{I},\mathcal{J}} , \mathbf{R}_{\mathcal{I},\mathcal{L}}}
(\boldsymbol{\sigma}_{\mathcal{K}\setminus\mathcal{L}} \mid
            \mathbf{r}_{\mathcal{I},\mathcal{J}}, \mathbf{r}_{\mathcal{I},\mathcal{L}})
\notag \\
& \hspace{60pt}
\cdot p_{\mathbf{R}_{\mathcal{I},\mathcal{J}}\mid  \mathbf{R}_{\mathcal{I},\mathcal{L}}}
(\mathbf{r}_{\mathcal{I},\mathcal{J}}\mid  \mathbf{r}_{\mathcal{I},\mathcal{L}})
\notag \\
&=
2^{-2(K-L-2)m|\mathcal{X}|} 
\notag \\
& \hspace{10pt} \cdot \hspace{-10pt}
\sum_{\mathbf{r}_{\mathcal{I},\mathcal{J}}\in\mathcal{N}_m^{2(K-L-2)|\mathcal{X}|}}
p_{\boldsymbol{\Sigma}_{\mathbf{R}_{\mathcal{I},\mathcal{I}}
\mid \mathbf{R}_{\mathcal{I},\mathcal{J}},
            \mathbf{R}_{\mathcal{I},\mathcal{L}}}}
(\boldsymbol{\sigma}_{\mathcal{I}}
\mid \mathbf{r}_{\mathcal{I},\mathcal{J}}, \mathbf{r}_{\mathcal{I},\mathcal{L}})
\notag \\
& \hspace{40pt} \cdot
p_{\boldsymbol{\Sigma}_{\mathbf{R}_{\mathcal{I},\mathcal{J}} \mid \mathbf{R}_{\mathcal{I},\mathcal{J}}}}
(\boldsymbol{\sigma}_{\mathcal{J}} \mid \mathbf{r}_{\mathcal{I},\mathcal{J}}),
\label{equ:lem2-0}
\end{align}
where the second equality results because
$\mathbf{R}_{\mathcal{I},\mathcal{J}}$ and $\mathbf{R}_{\mathcal{I},\mathcal{L}}$ contains i.i.d. uniform
elements and
$\boldsymbol{\Sigma}_{\mathbf{R}_{\mathcal{I},\mathcal{J}}}$ is a
function of $\mathbf{R}_{\mathcal{I},\mathcal{J}}$. More specifically,
this latter fact gives
\begin{equation}\label{equ:lem2-1}
p_{\boldsymbol{\Sigma}_{\mathbf{R}_{\mathcal{I},\mathcal{J}}}\mid
\mathbf{R}_{\mathcal{I},\mathcal{J}}}
(\boldsymbol{\sigma}_{\mathcal{J}}\mid \mathbf{r}_{\mathcal{I},\mathcal{J}})
=\prod_{j\in\mathcal{J}}\delta\left(\boldsymbol{\sigma}_j\ominus\boldsymbol{\Sigma}_{\mathbf{r}_{\mathcal{I},j}}\right).
\end{equation}

Since $\mathbf{R}_{1,2}$, $\mathbf{R}_{2,1}$, and the elements of
$\mathbf{R}_{\mathcal{I},\mathcal{K} \setminus \mathcal{I}}$ are
i.i.d. uniform, letting
$\mathbf{W}=\mathbf{R}_{2,1} \ominus \mathbf{R}_{1,2}$ gives that
$\mathbf{W}$ is uniform and independent of
$\mathbf{R}_{\mathcal{I},\mathcal{K} \setminus \mathcal{I}}$, i.e.,
for any $\mathbf{w}\in\mathcal{N}_m^{|\mathcal{X}|}$ and
$\mathbf{r}_{\mathcal{I},\mathcal{K}\setminus \mathcal{I}} \in
\mathcal{N}_m^{2(K-2)|\mathcal{X}|}$,
$p_{\mathbf{W}\mid \mathbf{R}_{\mathcal{I},\mathcal{K}\setminus
    \mathcal{I}}} (\mathbf{w}\mid
\mathbf{r}_{\mathcal{I},\mathcal{K}\setminus \mathcal{I}})
=p_{\mathbf{W}}(\mathbf{w})=2^{-m|\mathcal{X}|}$.  In addition,
because
$\boldsymbol{\Sigma}_{\mathbf{R}_{\mathcal{I},1}}= \mathbf{W} \ominus
\bigoplus_{k\in\mathcal{K} \setminus \mathcal{I}} \mathbf{R}_{1,k}$
and
$\boldsymbol{\Sigma}_{\mathbf{R}_{\mathcal{I},2}}=
\ominus\mathbf{W}\ominus\bigoplus_{k\in\mathcal{K} \setminus
  \mathcal{I}} \mathbf{R}_{2,k}$, we have
\begin{align}
&p_{\boldsymbol{\Sigma}_{\mathbf{R}_{\mathcal{I},\mathcal{I}}} \mid
 \mathbf{R}_{\mathcal{I},\mathcal{K}\setminus \mathcal{I}}}
(\boldsymbol{\sigma}_{\mathcal{I}}\mid
                \mathbf{r}_{\mathcal{I},\mathcal{K} \setminus \mathcal{I}})
\notag \\
&=
p_{\mathbf{W}} \left(\boldsymbol{\sigma}_1 \oplus \bigoplus_{k\in
            \mathcal{K} \setminus \mathcal{I}} \mathbf{r}_{1,k} \right)
\cdot \delta\left(\boldsymbol{\sigma}_1\oplus\boldsymbol{\sigma}_2
\oplus\bigoplus_{k\in \mathcal{K} \setminus \mathcal{I}}
            \boldsymbol{\Sigma}_{\mathbf{r}_{\mathcal{I},k}} \right)
\notag \\
&=
2^{-m|\mathcal{X}|}\cdot \delta\left(\boldsymbol{\sigma}_1\oplus\boldsymbol{\sigma}_2
\oplus\bigoplus_{j\in\mathcal{J}}\boldsymbol{\Sigma}_{\mathbf{r}_{\mathcal{I},j}}
\oplus\bigoplus_{l\in\mathcal{L}}\boldsymbol{\Sigma}_{\mathbf{r}_{\mathcal{I},l}}
\right).
\label{equ:lem2-2}
\end{align}
Now, inserting~\eqref{equ:lem2-1} and~\eqref{equ:lem2-2}
into~\eqref{equ:lem2-0} yields
\begin{align*}
&
p_{\boldsymbol{\Sigma}_{\mathbf{R}_{\mathcal{I}, \mathcal{K} \setminus
                 \mathcal{L}}}\mid \mathbf{R}_{\mathcal{I},\mathcal{L}}}
(\boldsymbol{\sigma}_{\mathcal{K}\setminus\mathcal{L}}\mid \mathbf{r}_{\mathcal{I},\mathcal{L}})
\\
&=
2^{-(2(K-L-2)+1) m |\mathcal{X}|} \hspace{-10pt}
\sum_{\mathbf{r}_{\mathcal{I},\mathcal{J}}\in\mathcal{N}_m^{2(K-L-2)|\mathcal{X}|}}
\prod_{j\in\mathcal{J}} \delta\left( \boldsymbol{\sigma}_j \ominus
     \boldsymbol{\Sigma}_{\mathbf{r}_{\mathcal{I},j}} \right)
\\
& \hspace{10pt}
\cdot \delta\left(\boldsymbol{\sigma}_1\oplus
\boldsymbol{\sigma}_2
\oplus \bigoplus_{j\in\mathcal{J}}\boldsymbol{\Sigma}_{\mathbf{r}_{\mathcal{I},j}}
\oplus\bigoplus_{l\in\mathcal{L}} \boldsymbol{\Sigma}_{\mathbf{r}_{\mathcal{I},l}}
\right)\\
&=
2^{-m(K-L-1)|\mathcal{X}|}\cdot
\delta\left(\boldsymbol{\Sigma}_{\boldsymbol{\sigma}_{\mathcal{K}\setminus\mathcal{L}}}
\oplus\bigoplus_{l\in\mathcal{L}}\boldsymbol{\Sigma}_{\mathbf{r}_{\mathcal{I},l}}\right),
\end{align*}
where the second equality is due to simple counting.

\subsection{Proof of Lemma~\ref{lem:useful3}}

Since the generation of $W$ is PPT and $\hat{B}(Y,V,W)$ is PPT,
$\hat{B}_0(Y,U,V)$ is PPT by construction. In addition, for any
$b\in\{0,1\}$, $y\in\mathcal{Y}$, $u\in\mathcal{U}$, and
$v\in\mathcal{V}$,
\begin{align*}
&p_{\hat{B}_0\mid Y,U,V}(b\mid y,u,v) 
\\
&=\sum_{w\in\mathcal{W}} P_{\hat{B}\mid Y,V,W}(b\mid y,v,w)
\cdot p_{W\mid Y,U,V}(w\mid y,u,v).
\end{align*}
Hence,
\begin{align*}
&
\Pr(\hat{B}(Y,V,W)=B\mid U=u,V=v)
\\
&=
\sum_{b\in\{0,1\}}\sum_{y\in\mathcal{Y}}\sum_{w\in\mathcal{W}}
p_{\hat{B}\mid Y,V,W}(b\mid y,v,w)
\\
& \hspace{30pt}
\cdot p_{W\mid Y,U,V}(w\mid y,u,v)
\cdot p_{B,Y\mid U,V}(b,y\mid u,v)
\\
&=
\sum_{b\in\{0,1\}}\sum_{y\in\mathcal{Y}}
p_{\hat{B}_0\mid Y,U,V}(b\mid y,u,v)
\cdot  p_{B,Y\mid U,V}(b,y\mid u,v)
\\
&=
\Pr(\hat{B}_0(Y,U,V)=B\mid U=u,V=v).
\end{align*}

\section{Proof of~\eqref{equ:A3}} \label{app:A3}
\begin{figure*}
\begin{align}
& \hspace{-5pt}
\Pr(\hat{B}_3 ( \mathbf{G}^B_{\mathcal{K}\setminus\mathcal{L}},
\mathbf{Q}^0_{\mathcal{K}\setminus\mathcal{L}},
\mathbf{Q}^1_{\mathcal{K}\setminus\mathcal{L}},
{\mathbf{\bar R}}_{\mathcal{K}\setminus\mathcal{I},\mathcal{K}\setminus\mathcal{L}},
\mathbf{R}_{\mathcal{K}\setminus\mathcal{L},\mathcal{L}},
\boldsymbol{\Phi}^n_{\mathcal{K}\setminus\mathcal{L}})=B
\mid 
\mathbb{Q}(\mathbf{q}^0_{\mathcal{K}\setminus\mathcal{L}},
\mathbf{q}^1_{\mathcal{K}\setminus\mathcal{L}},\mathbf{q}_{\mathcal{L}}),
{\mathbf{\bar R}}_{\mathcal{K}\setminus\mathcal{I},\mathcal{K}\setminus\mathcal{L}}
={\mathbf{\bar r}}_{\mathcal{K}\setminus\mathcal{I},\mathcal{K}\setminus\mathcal{L}},
\notag \\
& \hspace{20pt}
\mathbf{R}_{\mathcal{K}\setminus\mathcal{L},\mathcal{L}}
=\mathbf{r}_{\mathcal{K}\setminus\mathcal{L},\mathcal{L}},
\boldsymbol{\Phi}^n_{\mathcal{K}\setminus\mathcal{L}}
=\boldsymbol{\phi}^n_{\mathcal{K}\setminus\mathcal{L}})
\notag \\
&=
\frac{1}{2}
\sum_{\mathbf{g}_{\mathcal{K}\setminus \mathcal{L}}\in\mathcal{N}_m^{(K-L)|\mathcal{X}|}}
\sum_{b\in\{0,1\}}
p_{\hat{B}_3\mid 
\mathbf{G}^B_{\mathcal{K}\setminus\mathcal{L}},
\mathbf{Q}^0_{\mathcal{K}\setminus\mathcal{L}},
\mathbf{Q}^1_{\mathcal{K}\setminus\mathcal{L}},
{\mathbf{\bar R}}_{\mathcal{K}\setminus\mathcal{I},\mathcal{K}\setminus\mathcal{L}},
\mathbf{R}_{\mathcal{K}\setminus\mathcal{L},\mathcal{L}},
\boldsymbol{\Phi}^n_{\mathcal{K}\setminus\mathcal{L}}}
(b \mid 
\mathbf{g}_{\mathcal{K}\setminus\mathcal{L}},
\mathbf{q}^0_{\mathcal{K}\setminus\mathcal{L}},
\mathbf{q}^1_{\mathcal{K}\setminus\mathcal{L}},
{\mathbf{\bar r}}_{\mathcal{K}\setminus\mathcal{I},\mathcal{K}\setminus\mathcal{L}},
\mathbf{r}_{\mathcal{K}\setminus\mathcal{L},\mathcal{L}},
\boldsymbol{\phi}^n_{\mathcal{K}\setminus\mathcal{L}})
\notag \\
& \hspace{60pt}
\cdot p_{\mathbf{G}^B_{\mathcal{K}\setminus\mathcal{L}}
\mid\mathbf{Q}^0_{\mathcal{K}\setminus\mathcal{L}},
\mathbf{Q}^1_{\mathcal{K}\setminus\mathcal{L}}, \mathbf{Q}_{\mathcal{L}},
{\mathbf{\bar R}}_{\mathcal{K}\setminus\mathcal{I},\mathcal{K}\setminus\mathcal{L}},
\mathbf{R}_{\mathcal{K}\setminus\mathcal{L},\mathcal{L}},
\boldsymbol{\Phi}^n_{\mathcal{K}\setminus\mathcal{L}},
B}
(\mathbf{g}_{\mathcal{K}\setminus\mathcal{L}}
\mid \mathbf{q}^0_{\mathcal{K}\setminus\mathcal{L}},
\mathbf{q}^1_{\mathcal{K}\setminus\mathcal{L}}, \mathbf{q}_{\mathcal{L}},
{\mathbf{\bar r}}_{\mathcal{K}\setminus\mathcal{I},\mathcal{K}\setminus\mathcal{L}},
\mathbf{r}_{\mathcal{K}\setminus\mathcal{L},\mathcal{L}},
\boldsymbol{\phi}^n_{\mathcal{K}\setminus\mathcal{L}},
b)
\notag \\
&=
\frac{1}{2}
\sum_{\mathbf{g}_{\mathcal{K}\setminus \mathcal{L}}\in\mathcal{N}_m^{(K-L)|\mathcal{X}|}}
\sum_{b\in\{0,1\}}
p_{\hat{B}_3 \mid 
\mathbf{G}^B_{\mathcal{K}\setminus\mathcal{L}},
\mathbf{Q}^0_{\mathcal{K}\setminus\mathcal{L}},
\mathbf{Q}^1_{\mathcal{K}\setminus\mathcal{L}},
{\mathbf{\bar R}}_{\mathcal{K}\setminus\mathcal{I},\mathcal{K}\setminus\mathcal{L}},
\mathbf{R}_{\mathcal{K}\setminus\mathcal{L},\mathcal{L}},
\boldsymbol{\Phi}^n_{\mathcal{K}\setminus\mathcal{L}}}
(b \mid
\mathbf{g}_{\mathcal{K}\setminus\mathcal{L}},
\mathbf{q}^0_{\mathcal{K}\setminus\mathcal{L}},
\mathbf{q}^1_{\mathcal{K}\setminus\mathcal{L}},
{\mathbf{\bar r}}_{\mathcal{K}\setminus\mathcal{I},\mathcal{K}\setminus\mathcal{L}},
\mathbf{r}_{\mathcal{K}\setminus\mathcal{L},\mathcal{L}},
\boldsymbol{\phi}^n_{\mathcal{K}\setminus\mathcal{L}})
\notag \\
& \hspace{80pt}
\cdot p_{\boldsymbol{\Sigma}_{\mathbf{R}_{\mathcal{I},\mathcal{K}\setminus\mathcal{L}}}
\mid
\mathbf{R}_{\mathcal{I},\mathcal{L}}}
(\mathbf{g}_{\mathcal{K}\setminus\mathcal{L}} \ominus \mathbf{q}^b_{\mathcal{K}\setminus\mathcal{L}}
\ominus \boldsymbol{\Sigma}_{\mathbf{r}_{\mathcal{K}\setminus\mathcal{I},\mathcal{K}\setminus\mathcal{L}}}
\mid 
\mathbf{r}_{\mathcal{I},\mathcal{L}})
\notag \\
&=
\frac{1}{2}
\cdot 2^{-m(K-L-1)|\mathcal{X}|} \cdot
\sum_{\mathbf{g}_{\mathcal{K}\setminus \mathcal{L}}\in\mathcal{N}_m^{(K-L)|\mathcal{X}|}}
\delta\left(
\boldsymbol{\Sigma}_{\mathbf{g}_{\mathcal{K}\setminus\mathcal{L}}}
\ominus 
\boldsymbol{\Sigma}_{\mathbf{q}^0_{\mathcal{K}\setminus\mathcal{L}}}
\ominus 
\bigoplus_{l\in\mathcal{K}\setminus\mathcal{L}}
\boldsymbol{\Sigma}_{\mathbf{r}_{\mathcal{K}\setminus\mathcal{I},l}}
\oplus
\bigoplus_{l\in\mathcal{L}}
\boldsymbol{\Sigma}_{\mathbf{r}_{\mathcal{I},l}}
\right)
\notag \\
& \hspace{30pt}
\cdot \underbrace{\sum_{b\in\{0,1\}}
p_{\hat{B}_3 \mid 
\mathbf{G}^B_{\mathcal{K}\setminus\mathcal{L}},
\mathbf{Q}^0_{\mathcal{K}\setminus\mathcal{L}},
\mathbf{Q}^1_{\mathcal{K}\setminus\mathcal{L}},
{\mathbf{\bar R}}_{\mathcal{K}\setminus\mathcal{I},\mathcal{K}\setminus\mathcal{L}},
\mathbf{R}_{\mathcal{K}\setminus\mathcal{L},\mathcal{L}},
\boldsymbol{\Phi}^n_{\mathcal{K}\setminus\mathcal{L}}}
(b\mid
\mathbf{g}_{\mathcal{K}\setminus\mathcal{L}},
\mathbf{q}^0_{\mathcal{K}\setminus\mathcal{L}},
\mathbf{q}^1_{\mathcal{K}\setminus\mathcal{L}},
{\mathbf{\bar r}}_{\mathcal{K}\setminus\mathcal{I},\mathcal{K}\setminus\mathcal{L}},
\mathbf{r}_{\mathcal{K}\setminus\mathcal{L},\mathcal{L}},
\boldsymbol{\phi}^n_{\mathcal{K}\setminus\mathcal{L}})}_{1}
\notag \\
&=
\frac{1}{2}
\cdot 
2^{-m(K-L-1)|\mathcal{X}|} \cdot
\sum_{\mathbf{g}_{\mathcal{K}\setminus \mathcal{L}}\in\mathcal{N}_m^{(K-L)|\mathcal{X}|}}
\delta\left(
\boldsymbol{\Sigma}_{\mathbf{g}_{\mathcal{K}\setminus\mathcal{L}}}
\ominus 
\boldsymbol{\Sigma}_{\mathbf{q}^0_{\mathcal{K}\setminus\mathcal{L}}}
\ominus 
\bigoplus_{l\in\mathcal{K}\setminus\mathcal{L}}
\boldsymbol{\Sigma}_{\mathbf{r}_{\mathcal{K}\setminus\mathcal{I},l}}
\oplus
\bigoplus_{l\in\mathcal{L}}
\boldsymbol{\Sigma}_{\mathbf{r}_{\mathcal{I},l}}
\right)
=\frac{1}{2}.
\label{equ:a3-1}
\end{align}
\hrule
\end{figure*}
Equation~\eqref{equ:a3-1} shown on top of the next page provides the
steps to establish the second equality in~\eqref{equ:A3}, where the
first equality is due to the functional form of $\hat{B}_{3}$, the
second equality results because
\[
  \mathbf{G}^B_{\mathcal{K}\setminus\mathcal{L}}
=\mathbf{Q}^B_{\mathcal{K}\setminus\mathcal{L}} \oplus
\boldsymbol{\Sigma}_{\mathbf{R}_{\mathcal{K}\setminus\mathcal{I},\mathcal{K}\setminus\mathcal{L}}}
\oplus
\boldsymbol{\Sigma}_{\mathbf{R}_{\mathcal{I},\mathcal{K}\setminus\mathcal{L}}},
\]
$\mathbf{R}_{\mathcal{I},\mathcal{K}\setminus\mathcal{L}}$, and hence
$\boldsymbol{\Sigma}_{\mathbf{R}_{\mathcal{I},\mathcal{K}\setminus\mathcal{L}}}$,
are conditionally independent of
$[\mathbf{Q}^0_{\mathcal{K}\setminus\mathcal{L}},
\mathbf{Q}^1_{\mathcal{K}\setminus\mathcal{L}},
\mathbf{Q}_{\mathcal{L}}, {\mathbf{\bar
    R}}_{\mathcal{K}\setminus\mathcal{I},\mathcal{K}\setminus\mathcal{L}},
\mathbf{R}_{\mathcal{J},\mathcal{L}},
\boldsymbol{\Phi}^n_{\mathcal{K}\setminus\mathcal{L}}, B]$ given
$\mathbf{R}_{\mathcal{I},\mathcal{L}}$, and
$\boldsymbol{\Sigma}_{\mathbf{R}_{\mathcal{K}\setminus\mathcal{I},
    \mathcal{K}\setminus\mathcal{L}}}$ is a deterministic function of
$[{\mathbf{\bar    R}}_{\mathcal{K}\setminus\mathcal{I},\mathcal{K}\setminus\mathcal{L}},
\mathbf{R}_{\mathcal{K}\setminus\mathcal{L},\mathcal{L}}]$, the third
equality is due to Lemma~\ref{lem:useful2} and that
$\boldsymbol{\Sigma}_{\mathbf{g}_{\mathcal{K}\setminus\mathcal{L}}
  \ominus \mathbf{q}^b_{\mathcal{K}\setminus\mathcal{L}} \ominus
  \boldsymbol{\Sigma}_{\mathbf{R}_{\mathcal{K}\setminus\mathcal{I},\mathcal{K}\setminus\mathcal{L}}}}
= \boldsymbol{\Sigma}_{\mathbf{g}_{\mathcal{K}\setminus\mathcal{L}}}
\ominus
\boldsymbol{\Sigma}_{\mathbf{q}^0_{\mathcal{K}\setminus\mathcal{L}}}
\ominus \bigoplus_{l\in\mathcal{K}\setminus\mathcal{L}}
\boldsymbol{\Sigma}_{\mathbf{r}_{\mathcal{K}\setminus\mathcal{I},l}}$
for both $b=0$ and $1$ (recall
$\boldsymbol{\Sigma}_{\mathbf{q}^0_{\mathcal{K}\setminus\mathcal{L}}}
=
\boldsymbol{\Sigma}_{\mathbf{q}^1_{\mathcal{K}\setminus\mathcal{L}}}$),
and the last equality results because the number of elements
$\mathbf{g}_{\mathcal{K}\setminus\mathcal{L}}
\in\mathcal{N}_m^{(K-L)|\mathcal{X}|}$ that
$\boldsymbol{\Sigma}_{\mathbf{g}_{\mathcal{K}\setminus\mathcal{L}}}$
equals any specific element in $\mathcal{N}_m^{(K-L)|\mathcal{X}|}$ is
exactly $2^{m(K-L-1)|\mathcal{X}|}$.

\bibliographystyle{IEEEtran}

\begin{thebibliography}{10}
\providecommand{\url}[1]{#1}
\csname url@samestyle\endcsname
\providecommand{\newblock}{\relax}
\providecommand{\bibinfo}[2]{#2}
\providecommand{\BIBentrySTDinterwordspacing}{\spaceskip=0pt\relax}
\providecommand{\BIBentryALTinterwordstretchfactor}{4}
\providecommand{\BIBentryALTinterwordspacing}{\spaceskip=\fontdimen2\font plus
\BIBentryALTinterwordstretchfactor\fontdimen3\font minus \fontdimen4\font\relax}
\providecommand{\BIBforeignlanguage}[2]{{%
\expandafter\ifx\csname l@#1\endcsname\relax
\typeout{** WARNING: IEEEtran.bst: No hyphenation pattern has been}%
\typeout{** loaded for the language `#1'. Using the pattern for}%
\typeout{** the default language instead.}%
\else
\language=\csname l@#1\endcsname
\fi
#2}}
\providecommand{\BIBdecl}{\relax}
\BIBdecl

\bibitem{kruskal1952use}
W.~H. Kruskal and W.~A. Wallis, ``Use of ranks in one-criterion analysis of variance,'' \emph{Journal of the American Statistical Association}, vol.~47, no. 260, pp. 583--621, 1952.

\bibitem{terry1952some}
M.~E. Terry, ``Some rank order tests which are most powerful against specific parametric alternatives,'' \emph{The Annals of Mathematical Statistics}, pp. 346--366, 1952.

\bibitem{puri1965some}
M.~L. Puri, ``Some distribution-free k-sample rank tests of homogeneity against ordered alternatives,'' 1965.

\bibitem{mack1981k}
G.~A. Mack and D.~A. Wolfe, ``K-sample rank tests for umbrella alternatives,'' \emph{Journal of the American Statistical Association}, vol.~76, no. 373, pp. 175--181, 1981.

\bibitem{scholz1987k}
F.~W. Scholz and M.~A. Stephens, ``K-sample {Anderson-Darling} tests,'' \emph{Journal of the American Statistical Association}, vol.~82, no. 399, pp. 918--924, 1987.

\bibitem{murakami2006k}
H.~Murakami, ``A k-sample rank test based on modified baumgartner statistic and its power comparison,'' \emph{Journal of the Japanese Society of Computational Statistics}, vol.~19, no.~1, pp. 1--13, 2006.

\bibitem{quade1966analysis}
D.~Quade, ``On analysis of variance for the k-sample problem,'' \emph{The Annals of Mathematical Statistics}, vol.~37, no.~6, pp. 1747--1758, 1966.

\bibitem{szekely2004testing}
G.~J. Sz{\'e}kely, M.~L. Rizzo \emph{et~al.}, ``Testing for equal distributions in high dimension,'' \emph{InterStat}, vol.~5, no. 16.10, pp. 1249--1272, 2004.

\bibitem{kiefer1959k}
J.~Kiefer, ``K-sample analogues of the {Kolmogorov-Smirnov} and {Cram{\'e}r-V. Mises} tests,'' \emph{The Annals of Mathematical Statistics}, pp. 420--447, 1959.

\bibitem{chen2014bayesian}
Y.~Chen and T.~E. Hanson, ``Bayesian nonparametric k-sample tests for censored and uncensored data,'' \emph{Computational Statistics \& Data Analysis}, vol.~71, pp. 335--346, 2014.

\bibitem{Zeitouni1991}
O.~Zeitouni and M.~Gutman, ``On universal hypotheses testing via large deviations,'' \emph{IEEE Transactions on Information Theory}, vol.~37, no.~2, pp. 285--290, 1991.

\bibitem{lihong}
H.~Li, L.~Sun, H.~Zhu, X.~Lu, and X.~Cheng, ``Achieving privacy preservation in {WiFi} fingerprint-based localization,'' in \emph{Proc. of 2014 IEEE Conference on Computer Communications (INFOCOM)}, pp. 2337--2345.

\bibitem{mine}
X.~Wang, Y.~Liu, Z.~Shi, X.~Lu, and L.~Sun, ``A privacy-preserving fuzzy localization scheme with {CSI} fingerprint,'' in \emph{Proc. of 2015 IEEE Global Communications Conference (GLOBECOM)}.

\bibitem{wang2020privacy}
N.~Wang, J.~Le, W.~Li, L.~Jiao, Z.~Li, and K.~Zeng, ``Privacy protection and efficient incumbent detection in spectrum sharing based on federated learning,'' in \emph{Proc. of 2020 IEEE Conference on Communications and Network Security (CNS)}.

\bibitem{multics2022}
J.~Liang, D.~Xiao, H.~Huang, and M.~Li, ``Multilevel privacy preservation scheme based on compressed sensing,'' \emph{IEEE Transactions on Industrial Informatics}, vol.~19, no.~6, pp. 7435--7444, 2023.

\bibitem{dwork2006}
C.~Dwork, F.~McSherry, K.~Nissim, and A.~Smith, ``Calibrating noise to sensitivity in private data analysis,'' in \emph{Theory of Cryptography: Third Theory of Cryptography Conference, TCC 2006}.\hskip 1em plus 0.5em minus 0.4em\relax Springer, pp. 265--284.

\bibitem{wei2020federated}
K.~Wei, J.~Li, M.~Ding, C.~Ma, H.~H. Yang, F.~Farokhi, S.~Jin, T.~Q. Quek, and H.~V. Poor, ``Federated learning with differential privacy: Algorithms and performance analysis,'' \emph{IEEE Transactions on Information Forensics and Security}, vol.~15, pp. 3454--3469, 2020.

\bibitem{seif2020wireless}
M.~Seif, R.~Tandon, and M.~Li, ``Wireless federated learning with local differential privacy,'' in \emph{Proc. of 2020 IEEE International Symposium on Information Theory (ISIT)}, pp. 2604--2609.

\bibitem{taoshu}
T.~Shu, Y.~Chen, J.~Yang, and A.~Williams, ``Multi-lateral privacy-preserving localization in pervasive environments,'' in \emph{Proc. of 2014 IEEE Conference on Computer Communications (INFOCOM)}, pp. 2319--2327.

\bibitem{ukil2010privacy}
A.~Ukil, ``Privacy preserving data aggregation in wireless sensor networks,'' in \emph{Proc. of 2010 6th International Conference on Wireless and Mobile Communications}.\hskip 1em plus 0.5em minus 0.4em\relax IEEE, pp. 435--440.

\bibitem{danezis2013smart}
G.~Danezis, C.~Fournet, M.~Kohlweiss, and S.~Zanella-B{\'e}guelin, ``Smart meter aggregation via secret-sharing,'' in \emph{Proc. of the first ACM workshop on Smart energy grid security}, 2013, pp. 75--80.

\bibitem{modulo}
M.~Hayashi and T.~Koshiba, ``Secure modulo zero-sum randomness as cryptographic resource,'' \emph{Cryptology ePrint Archive}, 2018.

\bibitem{secagg}
K.~Bonawitz, V.~Ivanov, B.~Kreuter, A.~Marcedone, H.~B. McMahan, S.~Patel, D.~Ramage, A.~Segal, and K.~Seth, ``Practical secure aggregation for privacy-preserving machine learning,'' in \emph{Proc. of 2017 ACM SIGSAC Conference on Computer and Communications Security (CCS)}, pp. 1175--1191.

\bibitem{so2022lightsecagg}
J.~So, C.~He, C.-S. Yang, S.~Li, Q.~Yu, R.~E~Ali, B.~Guler, and S.~Avestimehr, ``Lightsecagg: a lightweight and versatile design for secure aggregation in federated learning,'' \emph{Proceedings of Machine Learning and Systems}, vol.~4, pp. 694--720, 2022.

\bibitem{liu2023long}
Z.~Liu, H.-Y. Lin, and Y.~Liu, ``Long-term privacy-preserving aggregation with user-dynamics for federated learning,'' \emph{IEEE Transactions on Information Forensics and Security}, vol.~18, pp. 2398--2412, 2023.

\bibitem{hd}
I.~Sason and S.~Verd{\'u}, ``$ f $-divergence inequalities,'' \emph{IEEE Transactions on Information Theory}, vol.~62, no.~11, pp. 5973--6006, 2016.

\bibitem{info}
T.~M. Cover, \emph{Elements of information theory (2nd edition)}.\hskip 1em plus 0.5em minus 0.4em\relax John Wiley \& Sons, 2006.

\bibitem{intro}
J.~Katz and Y.~Lindell, \emph{Introduction to modern cryptography}.\hskip 1em plus 0.5em minus 0.4em\relax CRC press, 2020.

\bibitem{elgamal}
Y.~Tsiounis and M.~Yung, ``On the security of {ElGamal} based encryption,'' in \emph{Public Key Cryptography: First International Workshop on Practice and Theory in Public Key Cryptography, PKC'98}.\hskip 1em plus 0.5em minus 0.4em\relax Springer, 2006, pp. 117--134.

\bibitem{rsa}
E.~Kiltz, A.~O’Neill, and A.~Smith, ``Instantiability of {RSA-OAEP} under chosen-plaintext attack,'' \emph{Journal of Cryptology}, vol.~30, no.~3, pp. 889--919, 2017.

\bibitem{Hoeffding1965}
W.~Hoeffding, ``Asymptotically optimal tests for multinomial distributions,'' \emph{The Annals of Mathematical Statistics}, pp. 369--401, 1965.

\bibitem{<1km}
\BIBentryALTinterwordspacing
{Wireless Innovation Forum (WINNF)}, ``Requirements for commercial operation in the {U.S.} 3550-3700 mhz citizens broadband radio service band,'' Dec 2022, Version V1.10.0. [Online]. Available: \url{https://winnf.memberclicks.net/assets/CBRS/WINNF-TS-0112.pdf}
\BIBentrySTDinterwordspacing

\bibitem{>1km}
E.~F. Drocella, J.~Richards, R.~Sole, F.~Najmy, A.~Lundy, and P.~McKenna, \emph{3.5 GHz exclusion zone analyses and methodology}.\hskip 1em plus 0.5em minus 0.4em\relax US Department of Commerce, NTIA Technical Report TR-15-517, 2015.

\bibitem{kmcpa}
M.~Bellare, A.~Boldyreva, and S.~Micali, ``Public-key encryption in a multi-user setting: Security proofs and improvements,'' in \emph{Advances in Cryptology—EUROCRYPT 2000: International Conference on the Theory and Application of Cryptographic Techniques}.\hskip 1em plus 0.5em minus 0.4em\relax Springer, 2000, pp. 259--274.

\end{thebibliography}

\begin{IEEEbiographynophoto}{Xiaoshan Wang}
received the B.E. degree in automation and the M.S. degree in control theory and control engineering from Beijing Jiaotong University, China, in 2008 and 2011 respectively, and the Ph.D. degree in information security from the University of Chinese Academy of Sciences, China, in 2018. He is currently pursuing the Ph.D. degree with the University of Florida, USA. His research interests include privacy preservation, wireless communication security, information theory, and machine learning.
\end{IEEEbiographynophoto}


\begin{IEEEbiographynophoto}{Tan F. Wong} received the B.Sc. degree (Hons.) from the Chinese University of Hong Kong in 1991, and the M.S.E.E. and Ph.D. degrees from Purdue University in 1992 and 1997, respectively, all in electrical engineering. He was a Research Engineer with the Department of Electronics, Macquarie University, Sydney, Australia. He also served as a Postdoctoral Research Associate with the School of Electrical and Computer Engineering, Purdue University. Since 1998, he has been with the University of Florida, where he is currently a Professor of Electrical and Computer Engineering. 

\end{IEEEbiographynophoto}

\end{document}